\documentclass[aps,twocolumn,prx,superscriptaddress,nofootinbib]{revtex4-2}
\usepackage{bm}
\usepackage[T1]{fontenc}
\usepackage{amsmath,amssymb,mathtools,dsfont}
\usepackage{braket}
\usepackage{graphicx}
\usepackage[caption=false]{subfig}
\renewcommand{\thesubfigure}{(\alph{subfigure})}
\usepackage[hypertexnames=false]{hyperref}
\usepackage[capitalize,nameinlink]{cleveref}
\usepackage{autonum}
\makeatletter
\AtBeginDocument{\autonum@generatePatchedReferenceCSL{Cref}}
\makeatother
\usepackage{quantikz}
\usetikzlibrary{arrows.meta,calc}
\usepackage{physics}
\usepackage{algpseudocode}
\usepackage{amsthm}

\newtheorem{lemma}{Lemma}
\newtheorem{theorem}{Theorem}
\theoremstyle{definition}

\newtheorem{definition}{Definition}
\newtheorem{corollary}{Corollary}

\newcommand{\BQP}{\mathsf{BQP}}
\newcommand{\BPP}{\mathsf{BPP}}
\newcommand{\FBQP}{\mathsf{FBQP}}
\newcommand{\poly}{\mathrm{poly}}
\newcommand{\Hhist}{H_{\mathrm{hist}}}
\newcommand{\Dhist}{\Delta_{\mathrm{hist}}}
\newcommand{\pacc}{p_{\mathrm{acc}}}
\newcommand{\ueff}{u_{\mathrm{eff}}}
\newcommand{\Pf}{P_f}
\newcommand{\Oacc}{\widehat O_{\mathrm{acc}}}

\newcounter{algorithm}
\crefname{algorithm}{Algorithm}{Algorithms}

\begin{document}

\noindent
\hspace{\fill} YITP-26-130
\begingroup
\let\newpage\relax
% \maketitle
\endgroup

\title{Chern-number estimation:\\ nearly optimal quantum algorithm and provable quantum speedup}
\author{Soichiro Imamura}\email{imamura-soichiro524@g.ecc.u-tokyo.ac.jp}
\affiliation{Department of Physics, The University of Tokyo, Tokyo 113-0033, Japan}
\author{Shintaro Ae}
\affiliation{Physical and Theoretical Chemistry Laboratory, Department of Chemistry, University of Oxford, South Parks Road, Oxford OX1 3QZ, United Kingdom}
\author{Kazuki Sakamoto}
\affiliation{Graduate School of Engineering Science, The University of Osaka, 1-3 Machikaneyama, Toyonaka, Osaka 560-8531, Japan}
\author{Ryu Hayakawa}
\affiliation{Yukawa Institute for Theoretical Physics \& The Hakubi Center, Kyoto University, Japan}
\author{Chusei Kiumi}
\affiliation{Center for Quantum Information and Quantum Biology,
The University of Osaka, 1-2 Machikaneyama, Toyonaka, Osaka, 560-0043, Japan}

\date{\today}

\begin{abstract}
    Topological phases of matter are characterized by global properties of quantum states beyond conventional local order parameters.
    The Chern number is a central invariant in this characterization, linking the geometry of quantum states to robust physical observables and making its determination key to understanding quantum matter. We construct a nearly optimal quantum algorithm for estimating the Chern number of a smooth Hamiltonian family over a two-dimensional parameter torus with a unique, gapped ground state, including interacting many-body systems.
    The algorithm transports and reuses two ground-state copies prepared at a single reference point and uses generalized quantum signal processing to coherently accumulate geometric phases.
    It achieves a Hamiltonian-oracle evolution time of $\widetilde O(L_xL_y/\Delta_{\min}^3)$, where \(L_x,L_y\) bound the parameter derivatives of the Hamiltonian and \(\Delta_{\min}\) is a known spectral-gap lower bound.
    We prove a matching worst-case oracle lower bound, even when the Chern number is promised to be zero or one, establishing optimality up to polylogarithmic factors.
For local Hamiltonian families with an inverse-polynomial gap and a supplied guiding state, exact Chern-number computation is in \(\FBQP\) and is \(\BQP\)-hard, and merely deciding whether the Chern number is zero or one is already \(\BQP\)-complete.
These results establish how efficiently quantum computers can estimate a quantized topological invariant and provide evidence for quantum advantage in characterizing topological phases.

\end{abstract}

\maketitle

\section{Introduction}
\label{sec:introduction}

\begin{figure*}[t]
  \centering
  \captionsetup[subfloat]{farskip=0pt,nearskip=0pt,listofformat=subsimple}
  \subfloat{%
  \label{fig:intro_overview_a}%
  \begin{minipage}[t]{0.48\textwidth}
    \centering
    \raggedright\textbf{\thesubfigure}\par\vspace{-2pt}
    \centering
\begin{tikzpicture}[
        x=1.25cm,y=1.25cm,
        arrow/.style={-{Latex[length=4pt]},line width=0.6pt},
        lab/.style={align=center,font=\scriptsize},
        wire/.style={line width=0.5pt},
        gate/.style={draw,fill=white,line width=0.5pt,minimum width=0.55cm,minimum height=0.42cm,inner sep=1pt,font=\scriptsize}
    ]
      % Panel canvas: 8.5 cm x 6.2 cm.
      \path[use as bounding box] (-5.77cm,-1.96cm) rectangle (2.73cm,4.24cm);
      % ---------- torus with mesh ----------
      \begin{scope}[shift={(-2.9,1.45)},scale=0.85]
        \draw[gray!60,line width=0.3pt] (1.419,0.019) -- (1.418,0.039) -- (1.415,0.058) -- (1.411,0.077) -- (1.406,0.096) -- (1.399,0.115) -- (1.392,0.133) -- (1.384,0.151) -- (1.374,0.168) -- (1.364,0.185) -- (1.352,0.202) -- (1.340,0.218) -- (1.326,0.233) -- (1.312,0.248) -- (1.297,0.262) -- (1.281,0.276) -- (1.264,0.288) -- (1.247,0.300) -- (1.229,0.311) -- (1.210,0.321) -- (1.191,0.330) -- (1.171,0.339) -- (1.151,0.346) -- (1.130,0.353) -- (1.109,0.358) -- (1.087,0.363) -- (1.066,0.366) -- (1.044,0.369) -- (1.022,0.370) -- (1.000,0.371) -- (0.978,0.370) -- (0.956,0.369) -- (0.934,0.366) -- (0.913,0.363) -- (0.891,0.358) -- (0.870,0.353) -- (0.849,0.346) -- (0.829,0.339) -- (0.809,0.330) -- (0.790,0.321) -- (0.771,0.311) -- (0.753,0.300) -- (0.736,0.288) -- (0.719,0.276) -- (0.703,0.262) -- (0.688,0.248) -- (0.674,0.233) -- (0.660,0.218) -- (0.648,0.202) -- (0.636,0.185) -- (0.626,0.168) -- (0.616,0.151) -- (0.608,0.133) -- (0.601,0.115) -- (0.594,0.096) -- (0.589,0.077) -- (0.585,0.058) -- (0.582,0.039) -- (0.581,0.019);
        \draw[gray!60,line width=0.3pt] (1.225,0.472) -- (1.212,0.484) -- (1.198,0.495) -- (1.184,0.506) -- (1.168,0.515) -- (1.152,0.524) -- (1.135,0.532) -- (1.118,0.539) -- (1.100,0.544) -- (1.082,0.549) -- (1.063,0.553) -- (1.044,0.556) -- (1.024,0.557) -- (1.005,0.558) -- (0.985,0.558) -- (0.964,0.556) -- (0.944,0.554) -- (0.924,0.550) -- (0.904,0.546) -- (0.883,0.541) -- (0.863,0.534) -- (0.843,0.527) -- (0.823,0.518) -- (0.804,0.509) -- (0.785,0.499) -- (0.766,0.488) -- (0.748,0.476) -- (0.730,0.463) -- (0.713,0.450) -- (0.696,0.435) -- (0.680,0.420) -- (0.664,0.405) -- (0.650,0.389) -- (0.636,0.372) -- (0.622,0.354) -- (0.610,0.337) -- (0.598,0.318) -- (0.588,0.300) -- (0.578,0.281) -- (0.569,0.262) -- (0.562,0.242) -- (0.555,0.222) -- (0.549,0.203) -- (0.544,0.183) -- (0.541,0.163) -- (0.538,0.143) -- (0.536,0.124) -- (0.536,0.104) -- (0.536,0.085) -- (0.538,0.066) -- (0.541,0.047) -- (0.544,0.029) -- (0.549,0.011) -- (0.555,-0.007) -- (0.562,-0.024) -- (0.569,-0.040) -- (0.578,-0.056) -- (0.588,-0.071) -- (0.598,-0.086);
        \draw[gray!60,line width=0.3pt] (0.869,0.719) -- (0.856,0.723) -- (0.842,0.726) -- (0.828,0.727) -- (0.814,0.728) -- (0.799,0.728) -- (0.784,0.726) -- (0.769,0.724) -- (0.754,0.720) -- (0.738,0.715) -- (0.723,0.710) -- (0.707,0.703) -- (0.692,0.695) -- (0.676,0.686) -- (0.661,0.676) -- (0.645,0.666) -- (0.630,0.654) -- (0.615,0.642) -- (0.601,0.628) -- (0.586,0.614) -- (0.572,0.599) -- (0.559,0.583) -- (0.545,0.567) -- (0.533,0.550) -- (0.520,0.532) -- (0.508,0.514) -- (0.497,0.496) -- (0.486,0.476) -- (0.476,0.457) -- (0.467,0.437) -- (0.458,0.417) -- (0.450,0.397) -- (0.442,0.376) -- (0.436,0.355) -- (0.430,0.335) -- (0.425,0.314) -- (0.420,0.293) -- (0.417,0.273) -- (0.414,0.252) -- (0.412,0.232) -- (0.411,0.212) -- (0.410,0.193) -- (0.411,0.173) -- (0.412,0.155) -- (0.414,0.136) -- (0.417,0.118) -- (0.420,0.101) -- (0.425,0.085) -- (0.430,0.069) -- (0.436,0.054) -- (0.442,0.039) -- (0.450,0.026) -- (0.458,0.013) -- (0.467,0.001) -- (0.476,-0.010) -- (0.486,-0.020) -- (0.497,-0.029) -- (0.508,-0.037) -- (0.520,-0.044);
        \draw[gray!60,line width=0.3pt] (0.456,0.847) -- (0.448,0.847) -- (0.440,0.845) -- (0.432,0.843) -- (0.424,0.839) -- (0.416,0.834) -- (0.408,0.829) -- (0.399,0.822) -- (0.391,0.814) -- (0.383,0.805) -- (0.374,0.795) -- (0.366,0.783) -- (0.358,0.772) -- (0.349,0.759) -- (0.341,0.745) -- (0.333,0.730) -- (0.325,0.715) -- (0.317,0.698) -- (0.310,0.681) -- (0.302,0.664) -- (0.295,0.646) -- (0.288,0.627) -- (0.282,0.607) -- (0.275,0.587) -- (0.269,0.567) -- (0.263,0.546) -- (0.258,0.526) -- (0.253,0.504) -- (0.248,0.483) -- (0.243,0.461) -- (0.239,0.440) -- (0.236,0.418) -- (0.233,0.397) -- (0.230,0.375) -- (0.227,0.354) -- (0.225,0.333) -- (0.224,0.312) -- (0.223,0.291) -- (0.222,0.271) -- (0.222,0.252) -- (0.222,0.232) -- (0.223,0.214) -- (0.224,0.196) -- (0.225,0.178) -- (0.227,0.162) -- (0.230,0.146) -- (0.233,0.131) -- (0.236,0.116) -- (0.239,0.103) -- (0.243,0.091) -- (0.248,0.079) -- (0.253,0.068) -- (0.258,0.059) -- (0.263,0.050) -- (0.269,0.043) -- (0.275,0.036) -- (0.282,0.031) -- (0.288,0.027) -- (0.295,0.024);
        \draw[gray!60,line width=0.3pt] (0.000,0.888) -- (0.000,0.886) -- (0.000,0.883) -- (0.000,0.879) -- (0.000,0.873) -- (0.000,0.867) -- (0.000,0.859) -- (0.000,0.850) -- (0.000,0.840) -- (0.000,0.829) -- (0.000,0.818) -- (0.000,0.805) -- (0.000,0.791) -- (0.000,0.777) -- (0.000,0.761) -- (0.000,0.745) -- (0.000,0.728) -- (0.000,0.710) -- (0.000,0.692) -- (0.000,0.673) -- (0.000,0.654) -- (0.000,0.634) -- (0.000,0.613) -- (0.000,0.592) -- (0.000,0.571) -- (0.000,0.550) -- (0.000,0.528) -- (0.000,0.506) -- (0.000,0.484) -- (0.000,0.462) -- (0.000,0.440) -- (0.000,0.418) -- (0.000,0.397) -- (0.000,0.375) -- (0.000,0.354) -- (0.000,0.333) -- (0.000,0.312) -- (0.000,0.292) -- (0.000,0.272) -- (0.000,0.253) -- (0.000,0.235) -- (0.000,0.217) -- (0.000,0.200) -- (0.000,0.183) -- (0.000,0.167) -- (0.000,0.152) -- (0.000,0.139) -- (0.000,0.125) -- (0.000,0.113) -- (0.000,0.102) -- (0.000,0.092) -- (0.000,0.083) -- (0.000,0.075) -- (0.000,0.068) -- (0.000,0.062) -- (0.000,0.057) -- (0.000,0.054) -- (0.000,0.051) -- (0.000,0.050);
        \draw[gray!60,line width=0.3pt] (-0.456,0.847) -- (-0.448,0.847) -- (-0.440,0.845) -- (-0.432,0.843) -- (-0.424,0.839) -- (-0.416,0.834) -- (-0.408,0.829) -- (-0.399,0.822) -- (-0.391,0.814) -- (-0.383,0.805) -- (-0.374,0.795) -- (-0.366,0.783) -- (-0.358,0.772) -- (-0.349,0.759) -- (-0.341,0.745) -- (-0.333,0.730) -- (-0.325,0.715) -- (-0.317,0.698) -- (-0.310,0.681) -- (-0.302,0.664) -- (-0.295,0.646) -- (-0.288,0.627) -- (-0.282,0.607) -- (-0.275,0.587) -- (-0.269,0.567) -- (-0.263,0.546) -- (-0.258,0.526) -- (-0.253,0.504) -- (-0.248,0.483) -- (-0.243,0.461) -- (-0.239,0.440) -- (-0.236,0.418) -- (-0.233,0.397) -- (-0.230,0.375) -- (-0.227,0.354) -- (-0.225,0.333) -- (-0.224,0.312) -- (-0.223,0.291) -- (-0.222,0.271) -- (-0.222,0.252) -- (-0.222,0.232) -- (-0.223,0.214) -- (-0.224,0.196) -- (-0.225,0.178) -- (-0.227,0.162) -- (-0.230,0.146) -- (-0.233,0.131) -- (-0.236,0.116) -- (-0.239,0.103) -- (-0.243,0.091) -- (-0.248,0.079) -- (-0.253,0.068) -- (-0.258,0.059) -- (-0.263,0.050) -- (-0.269,0.043) -- (-0.275,0.036) -- (-0.282,0.031) -- (-0.288,0.027) -- (-0.295,0.024);
        \draw[gray!60,line width=0.3pt] (-0.869,0.719) -- (-0.856,0.723) -- (-0.842,0.726) -- (-0.828,0.727) -- (-0.814,0.728) -- (-0.799,0.728) -- (-0.784,0.726) -- (-0.769,0.724) -- (-0.754,0.720) -- (-0.738,0.715) -- (-0.723,0.710) -- (-0.707,0.703) -- (-0.692,0.695) -- (-0.676,0.686) -- (-0.661,0.676) -- (-0.645,0.666) -- (-0.630,0.654) -- (-0.615,0.642) -- (-0.601,0.628) -- (-0.586,0.614) -- (-0.572,0.599) -- (-0.559,0.583) -- (-0.545,0.567) -- (-0.533,0.550) -- (-0.520,0.532) -- (-0.508,0.514) -- (-0.497,0.496) -- (-0.486,0.476) -- (-0.476,0.457) -- (-0.467,0.437) -- (-0.458,0.417) -- (-0.450,0.397) -- (-0.442,0.376) -- (-0.436,0.355) -- (-0.430,0.335) -- (-0.425,0.314) -- (-0.420,0.293) -- (-0.417,0.273) -- (-0.414,0.252) -- (-0.412,0.232) -- (-0.411,0.212) -- (-0.410,0.193) -- (-0.411,0.173) -- (-0.412,0.155) -- (-0.414,0.136) -- (-0.417,0.118) -- (-0.420,0.101) -- (-0.425,0.085) -- (-0.430,0.069) -- (-0.436,0.054) -- (-0.442,0.039) -- (-0.450,0.026) -- (-0.458,0.013) -- (-0.467,0.001) -- (-0.476,-0.010) -- (-0.486,-0.020) -- (-0.497,-0.029) -- (-0.508,-0.037) -- (-0.520,-0.044);
        \draw[gray!60,line width=0.3pt] (-1.225,0.472) -- (-1.212,0.484) -- (-1.198,0.495) -- (-1.184,0.506) -- (-1.168,0.515) -- (-1.152,0.524) -- (-1.135,0.532) -- (-1.118,0.539) -- (-1.100,0.544) -- (-1.082,0.549) -- (-1.063,0.553) -- (-1.044,0.556) -- (-1.024,0.557) -- (-1.005,0.558) -- (-0.985,0.558) -- (-0.964,0.556) -- (-0.944,0.554) -- (-0.924,0.550) -- (-0.904,0.546) -- (-0.883,0.541) -- (-0.863,0.534) -- (-0.843,0.527) -- (-0.823,0.518) -- (-0.804,0.509) -- (-0.785,0.499) -- (-0.766,0.488) -- (-0.748,0.476) -- (-0.730,0.463) -- (-0.713,0.450) -- (-0.696,0.435) -- (-0.680,0.420) -- (-0.664,0.405) -- (-0.650,0.389) -- (-0.636,0.372) -- (-0.622,0.354) -- (-0.610,0.337) -- (-0.598,0.318) -- (-0.588,0.300) -- (-0.578,0.281) -- (-0.569,0.262) -- (-0.562,0.242) -- (-0.555,0.222) -- (-0.549,0.203) -- (-0.544,0.183) -- (-0.541,0.163) -- (-0.538,0.143) -- (-0.536,0.124) -- (-0.536,0.104) -- (-0.536,0.085) -- (-0.538,0.066) -- (-0.541,0.047) -- (-0.544,0.029) -- (-0.549,0.011) -- (-0.555,-0.007) -- (-0.562,-0.024) -- (-0.569,-0.040) -- (-0.578,-0.056) -- (-0.588,-0.071) -- (-0.598,-0.086);
        \draw[gray!60,line width=0.3pt] (-1.419,0.019) -- (-1.418,0.039) -- (-1.415,0.058) -- (-1.411,0.077) -- (-1.406,0.096) -- (-1.399,0.115) -- (-1.392,0.133) -- (-1.384,0.151) -- (-1.374,0.168) -- (-1.364,0.185) -- (-1.352,0.202) -- (-1.340,0.218) -- (-1.326,0.233) -- (-1.312,0.248) -- (-1.297,0.262) -- (-1.281,0.276) -- (-1.264,0.288) -- (-1.247,0.300) -- (-1.229,0.311) -- (-1.210,0.321) -- (-1.191,0.330) -- (-1.171,0.339) -- (-1.151,0.346) -- (-1.130,0.353) -- (-1.109,0.358) -- (-1.087,0.363) -- (-1.066,0.366) -- (-1.044,0.369) -- (-1.022,0.370) -- (-1.000,0.371) -- (-0.978,0.370) -- (-0.956,0.369) -- (-0.934,0.366) -- (-0.913,0.363) -- (-0.891,0.358) -- (-0.870,0.353) -- (-0.849,0.346) -- (-0.829,0.339) -- (-0.809,0.330) -- (-0.790,0.321) -- (-0.771,0.311) -- (-0.753,0.300) -- (-0.736,0.288) -- (-0.719,0.276) -- (-0.703,0.262) -- (-0.688,0.248) -- (-0.674,0.233) -- (-0.660,0.218) -- (-0.648,0.202) -- (-0.636,0.185) -- (-0.626,0.168) -- (-0.616,0.151) -- (-0.608,0.133) -- (-0.601,0.115) -- (-0.594,0.096) -- (-0.589,0.077) -- (-0.585,0.058) -- (-0.582,0.039) -- (-0.581,0.019);
        \draw[gray!60,line width=0.3pt] (-1.312,-0.255) -- (-1.311,-0.236) -- (-1.310,-0.216) -- (-1.307,-0.196) -- (-1.303,-0.176) -- (-1.299,-0.157) -- (-1.293,-0.137) -- (-1.286,-0.117) -- (-1.278,-0.098) -- (-1.270,-0.079) -- (-1.260,-0.060) -- (-1.249,-0.041) -- (-1.238,-0.023) -- (-1.225,-0.005) -- (-1.212,0.012) -- (-1.198,0.029) -- (-1.184,0.045) -- (-1.168,0.061) -- (-1.152,0.076) -- (-1.135,0.090) -- (-1.118,0.104) -- (-1.100,0.117) -- (-1.082,0.128) -- (-1.063,0.140) -- (-1.044,0.150) -- (-1.024,0.159) -- (-1.005,0.167) -- (-0.985,0.175) -- (-0.964,0.181) -- (-0.944,0.187) -- (-0.924,0.191) -- (-0.904,0.195) -- (-0.883,0.197) -- (-0.863,0.198) -- (-0.843,0.199) -- (-0.823,0.198) -- (-0.804,0.196) -- (-0.785,0.194) -- (-0.766,0.190) -- (-0.748,0.185) -- (-0.730,0.179) -- (-0.713,0.172) -- (-0.696,0.165) -- (-0.680,0.156) -- (-0.664,0.146) -- (-0.650,0.136) -- (-0.636,0.125) -- (-0.622,0.112);
        \draw[gray!60,line width=0.3pt] (-1.249,-0.445) -- (-1.260,-0.430) -- (-1.270,-0.415) -- (-1.278,-0.399) -- (-1.286,-0.383) -- (-1.293,-0.366) -- (-1.299,-0.349) -- (-1.303,-0.331) -- (-1.307,-0.312) -- (-1.310,-0.293) -- (-1.311,-0.274) -- (-1.312,-0.255);
        \draw[gray!60,line width=0.3pt] (-1.004,-0.471) -- (-1.004,-0.452) -- (-1.002,-0.432) -- (-1.000,-0.412) -- (-0.998,-0.391) -- (-0.994,-0.371) -- (-0.990,-0.350) -- (-0.984,-0.329) -- (-0.978,-0.309) -- (-0.972,-0.288) -- (-0.964,-0.267) -- (-0.956,-0.247) -- (-0.947,-0.227) -- (-0.938,-0.207) -- (-0.928,-0.187) -- (-0.917,-0.168) -- (-0.906,-0.150) -- (-0.894,-0.132) -- (-0.882,-0.114) -- (-0.869,-0.097) -- (-0.856,-0.081) -- (-0.842,-0.065) -- (-0.828,-0.050) -- (-0.814,-0.036) -- (-0.799,-0.022) -- (-0.784,-0.010) -- (-0.769,0.002) -- (-0.754,0.012) -- (-0.738,0.022) -- (-0.723,0.031) -- (-0.707,0.039) -- (-0.692,0.046) -- (-0.676,0.051) -- (-0.661,0.056) -- (-0.645,0.060) -- (-0.630,0.062) -- (-0.615,0.064) -- (-0.601,0.064) -- (-0.586,0.064) -- (-0.572,0.062) -- (-0.559,0.059) -- (-0.545,0.055);
        \draw[gray!60,line width=0.3pt] (-0.894,-0.708) -- (-0.906,-0.701) -- (-0.917,-0.693) -- (-0.928,-0.684) -- (-0.938,-0.674) -- (-0.947,-0.663) -- (-0.956,-0.651) -- (-0.964,-0.638) -- (-0.972,-0.625) -- (-0.978,-0.610) -- (-0.984,-0.595) -- (-0.990,-0.579) -- (-0.994,-0.563) -- (-0.998,-0.545) -- (-1.000,-0.528) -- (-1.002,-0.509) -- (-1.004,-0.491) -- (-1.004,-0.471);
        \draw[gray!60,line width=0.3pt] (-0.543,-0.616) -- (-0.543,-0.596) -- (-0.543,-0.576) -- (-0.541,-0.556) -- (-0.540,-0.535) -- (-0.538,-0.514) -- (-0.536,-0.492) -- (-0.533,-0.471) -- (-0.530,-0.449) -- (-0.526,-0.428) -- (-0.522,-0.406) -- (-0.517,-0.385) -- (-0.513,-0.363) -- (-0.508,-0.342) -- (-0.502,-0.321) -- (-0.496,-0.300) -- (-0.490,-0.280) -- (-0.484,-0.260) -- (-0.477,-0.241) -- (-0.470,-0.222) -- (-0.463,-0.204) -- (-0.456,-0.186) -- (-0.448,-0.169) -- (-0.440,-0.153) -- (-0.432,-0.137) -- (-0.424,-0.123) -- (-0.416,-0.109) -- (-0.408,-0.096) -- (-0.399,-0.084) -- (-0.391,-0.073) -- (-0.383,-0.063) -- (-0.374,-0.054) -- (-0.366,-0.046) -- (-0.358,-0.039) -- (-0.349,-0.033) -- (-0.341,-0.028) -- (-0.333,-0.025) -- (-0.325,-0.022) -- (-0.317,-0.021) -- (-0.310,-0.021);
        \draw[gray!60,line width=0.3pt] (-0.470,-0.844) -- (-0.477,-0.841) -- (-0.484,-0.837) -- (-0.490,-0.831) -- (-0.496,-0.825) -- (-0.502,-0.817) -- (-0.508,-0.809) -- (-0.513,-0.799) -- (-0.517,-0.788) -- (-0.522,-0.777) -- (-0.526,-0.764) -- (-0.530,-0.751) -- (-0.533,-0.737) -- (-0.536,-0.722) -- (-0.538,-0.706) -- (-0.540,-0.689) -- (-0.541,-0.672) -- (-0.543,-0.654) -- (-0.543,-0.635) -- (-0.543,-0.616);
        \draw[gray!60,line width=0.3pt] (-0.000,-0.667) -- (-0.000,-0.647) -- (-0.000,-0.627) -- (-0.000,-0.606) -- (-0.000,-0.585) -- (-0.000,-0.564) -- (-0.000,-0.542) -- (-0.000,-0.521) -- (-0.000,-0.499) -- (-0.000,-0.477) -- (-0.000,-0.455) -- (-0.000,-0.433) -- (-0.000,-0.411) -- (-0.000,-0.389) -- (-0.000,-0.368) -- (-0.000,-0.347) -- (-0.000,-0.326) -- (-0.000,-0.305) -- (-0.000,-0.285) -- (-0.000,-0.266) -- (-0.000,-0.247) -- (-0.000,-0.229) -- (-0.000,-0.211) -- (-0.000,-0.194) -- (-0.000,-0.178) -- (-0.000,-0.162) -- (-0.000,-0.148) -- (-0.000,-0.134) -- (-0.000,-0.121) -- (-0.000,-0.109) -- (-0.000,-0.099) -- (-0.000,-0.089) -- (-0.000,-0.080) -- (-0.000,-0.072) -- (-0.000,-0.066) -- (-0.000,-0.060) -- (-0.000,-0.056) -- (-0.000,-0.053) -- (-0.000,-0.050);
        \draw[gray!60,line width=0.3pt] (-0.000,-0.889) -- (-0.000,-0.888) -- (-0.000,-0.885) -- (-0.000,-0.882) -- (-0.000,-0.877) -- (-0.000,-0.871) -- (-0.000,-0.864) -- (-0.000,-0.856) -- (-0.000,-0.847) -- (-0.000,-0.837) -- (-0.000,-0.826) -- (-0.000,-0.814) -- (-0.000,-0.800) -- (-0.000,-0.786) -- (-0.000,-0.772) -- (-0.000,-0.756) -- (-0.000,-0.739) -- (-0.000,-0.722) -- (-0.000,-0.704) -- (-0.000,-0.686) -- (-0.000,-0.667);
        \draw[gray!60,line width=0.3pt] (0.543,-0.616) -- (0.543,-0.596) -- (0.543,-0.576) -- (0.541,-0.556) -- (0.540,-0.535) -- (0.538,-0.514) -- (0.536,-0.492) -- (0.533,-0.471) -- (0.530,-0.449) -- (0.526,-0.428) -- (0.522,-0.406) -- (0.517,-0.385) -- (0.513,-0.363) -- (0.508,-0.342) -- (0.502,-0.321) -- (0.496,-0.300) -- (0.490,-0.280) -- (0.484,-0.260) -- (0.477,-0.241) -- (0.470,-0.222) -- (0.463,-0.204) -- (0.456,-0.186) -- (0.448,-0.169) -- (0.440,-0.153) -- (0.432,-0.137) -- (0.424,-0.123) -- (0.416,-0.109) -- (0.408,-0.096) -- (0.399,-0.084) -- (0.391,-0.073) -- (0.383,-0.063) -- (0.374,-0.054) -- (0.366,-0.046) -- (0.358,-0.039) -- (0.349,-0.033) -- (0.341,-0.028) -- (0.333,-0.025) -- (0.325,-0.022) -- (0.317,-0.021) -- (0.310,-0.021);
        \draw[gray!60,line width=0.3pt] (0.470,-0.844) -- (0.477,-0.841) -- (0.484,-0.837) -- (0.490,-0.831) -- (0.496,-0.825) -- (0.502,-0.817) -- (0.508,-0.809) -- (0.513,-0.799) -- (0.517,-0.788) -- (0.522,-0.777) -- (0.526,-0.764) -- (0.530,-0.751) -- (0.533,-0.737) -- (0.536,-0.722) -- (0.538,-0.706) -- (0.540,-0.689) -- (0.541,-0.672) -- (0.543,-0.654) -- (0.543,-0.635) -- (0.543,-0.616);
        \draw[gray!60,line width=0.3pt] (1.004,-0.471) -- (1.004,-0.452) -- (1.002,-0.432) -- (1.000,-0.412) -- (0.998,-0.391) -- (0.994,-0.371) -- (0.990,-0.350) -- (0.984,-0.329) -- (0.978,-0.309) -- (0.972,-0.288) -- (0.964,-0.267) -- (0.956,-0.247) -- (0.947,-0.227) -- (0.938,-0.207) -- (0.928,-0.187) -- (0.917,-0.168) -- (0.906,-0.150) -- (0.894,-0.132) -- (0.882,-0.114) -- (0.869,-0.097) -- (0.856,-0.081) -- (0.842,-0.065) -- (0.828,-0.050) -- (0.814,-0.036) -- (0.799,-0.022) -- (0.784,-0.010) -- (0.769,0.002) -- (0.754,0.012) -- (0.738,0.022) -- (0.723,0.031) -- (0.707,0.039) -- (0.692,0.046) -- (0.676,0.051) -- (0.661,0.056) -- (0.645,0.060) -- (0.630,0.062) -- (0.615,0.064) -- (0.601,0.064) -- (0.586,0.064) -- (0.572,0.062) -- (0.559,0.059) -- (0.545,0.055);
        \draw[gray!60,line width=0.3pt] (0.894,-0.708) -- (0.906,-0.701) -- (0.917,-0.693) -- (0.928,-0.684) -- (0.938,-0.674) -- (0.947,-0.663) -- (0.956,-0.651) -- (0.964,-0.638) -- (0.972,-0.625) -- (0.978,-0.610) -- (0.984,-0.595) -- (0.990,-0.579) -- (0.994,-0.563) -- (0.998,-0.545) -- (1.000,-0.528) -- (1.002,-0.509) -- (1.004,-0.491) -- (1.004,-0.471);
        \draw[gray!60,line width=0.3pt] (1.312,-0.255) -- (1.311,-0.236) -- (1.310,-0.216) -- (1.307,-0.196) -- (1.303,-0.176) -- (1.299,-0.157) -- (1.293,-0.137) -- (1.286,-0.117) -- (1.278,-0.098) -- (1.270,-0.079) -- (1.260,-0.060) -- (1.249,-0.041) -- (1.238,-0.023) -- (1.225,-0.005) -- (1.212,0.012) -- (1.198,0.029) -- (1.184,0.045) -- (1.168,0.061) -- (1.152,0.076) -- (1.135,0.090) -- (1.118,0.104) -- (1.100,0.117) -- (1.082,0.128) -- (1.063,0.140) -- (1.044,0.150) -- (1.024,0.159) -- (1.005,0.167) -- (0.985,0.175) -- (0.964,0.181) -- (0.944,0.187) -- (0.924,0.191) -- (0.904,0.195) -- (0.883,0.197) -- (0.863,0.198) -- (0.843,0.199) -- (0.823,0.198) -- (0.804,0.196) -- (0.785,0.194) -- (0.766,0.190) -- (0.748,0.185) -- (0.730,0.179) -- (0.713,0.172) -- (0.696,0.165) -- (0.680,0.156) -- (0.664,0.146) -- (0.650,0.136) -- (0.636,0.125) -- (0.622,0.112);
        \draw[gray!60,line width=0.3pt] (1.249,-0.445) -- (1.260,-0.430) -- (1.270,-0.415) -- (1.278,-0.399) -- (1.286,-0.383) -- (1.293,-0.366) -- (1.299,-0.349) -- (1.303,-0.331) -- (1.307,-0.312) -- (1.310,-0.293) -- (1.311,-0.274) -- (1.312,-0.255);
        \draw[gray!60,line width=0.3pt] (-1.418,-0.035) -- (-1.412,-0.070) -- (-1.403,-0.104) -- (-1.389,-0.139) -- (-1.372,-0.173) -- (-1.351,-0.206) -- (-1.326,-0.239) -- (-1.297,-0.271) -- (-1.265,-0.303) -- (-1.230,-0.333) -- (-1.191,-0.363) -- (-1.149,-0.392) -- (-1.104,-0.420) -- (-1.055,-0.446) -- (-1.004,-0.471) -- (-0.950,-0.495) -- (-0.894,-0.518) -- (-0.835,-0.539) -- (-0.773,-0.559) -- (-0.710,-0.577) -- (-0.645,-0.594) -- (-0.578,-0.609) -- (-0.509,-0.622) -- (-0.439,-0.634) -- (-0.368,-0.644) -- (-0.295,-0.652) -- (-0.222,-0.658) -- (-0.148,-0.663) -- (-0.074,-0.666) -- (-0.000,-0.667) -- (0.074,-0.666) -- (0.148,-0.663) -- (0.222,-0.658) -- (0.295,-0.652) -- (0.368,-0.644) -- (0.439,-0.634) -- (0.509,-0.622) -- (0.578,-0.609) -- (0.645,-0.594) -- (0.710,-0.577) -- (0.773,-0.559) -- (0.835,-0.539) -- (0.894,-0.518) -- (0.950,-0.495) -- (1.004,-0.471) -- (1.055,-0.446) -- (1.104,-0.420) -- (1.149,-0.392) -- (1.191,-0.363) -- (1.230,-0.333) -- (1.265,-0.303) -- (1.297,-0.271) -- (1.326,-0.239) -- (1.351,-0.206) -- (1.372,-0.173) -- (1.389,-0.139) -- (1.403,-0.104) -- (1.412,-0.070) -- (1.418,-0.035);
        \draw[gray!60,line width=0.3pt] (1.297,0.262) -- (1.295,0.294) -- (1.290,0.326) -- (1.281,0.357) -- (1.269,0.389) -- (1.253,0.420) -- (1.234,0.450) -- (1.211,0.480) -- (1.185,0.510) -- (1.156,0.539);
        \draw[gray!60,line width=0.3pt] (-1.156,0.539) -- (-1.185,0.510) -- (-1.211,0.480) -- (-1.234,0.450) -- (-1.253,0.420) -- (-1.269,0.389) -- (-1.281,0.357) -- (-1.290,0.326) -- (-1.295,0.294) -- (-1.297,0.262) -- (-1.295,0.230) -- (-1.290,0.199) -- (-1.281,0.167) -- (-1.269,0.136) -- (-1.253,0.105) -- (-1.234,0.074) -- (-1.211,0.044) -- (-1.185,0.015) -- (-1.156,-0.014) -- (-1.123,-0.042) -- (-1.088,-0.069) -- (-1.049,-0.096) -- (-1.008,-0.121) -- (-0.964,-0.145) -- (-0.917,-0.168) -- (-0.868,-0.190) -- (-0.816,-0.211) -- (-0.762,-0.230) -- (-0.706,-0.248) -- (-0.648,-0.265) -- (-0.589,-0.280) -- (-0.528,-0.294) -- (-0.465,-0.306) -- (-0.401,-0.317) -- (-0.336,-0.326) -- (-0.270,-0.333) -- (-0.203,-0.339) -- (-0.136,-0.343) -- (-0.068,-0.346) -- (-0.000,-0.347) -- (0.068,-0.346) -- (0.136,-0.343) -- (0.203,-0.339) -- (0.270,-0.333) -- (0.336,-0.326) -- (0.401,-0.317) -- (0.465,-0.306) -- (0.528,-0.294) -- (0.589,-0.280) -- (0.648,-0.265) -- (0.706,-0.248) -- (0.762,-0.230) -- (0.816,-0.211) -- (0.868,-0.190) -- (0.917,-0.168) -- (0.964,-0.145) -- (1.008,-0.121) -- (1.049,-0.096) -- (1.088,-0.069) -- (1.123,-0.042) -- (1.156,-0.014) -- (1.185,0.015) -- (1.211,0.044) -- (1.234,0.074) -- (1.253,0.105) -- (1.269,0.136) -- (1.281,0.167) -- (1.290,0.199) -- (1.295,0.230) -- (1.297,0.262);
        \draw[gray!60,line width=0.3pt] (1.000,0.371) -- (0.999,0.395) -- (0.995,0.420) -- (0.988,0.444) -- (0.978,0.468) -- (0.966,0.492) -- (0.951,0.516) -- (0.934,0.539) -- (0.914,0.562) -- (0.891,0.584) -- (0.866,0.606) -- (0.839,0.627) -- (0.809,0.647) -- (0.777,0.666) -- (0.743,0.685) -- (0.707,0.703) -- (0.669,0.720) -- (0.629,0.736) -- (0.588,0.751) -- (0.545,0.765) -- (0.500,0.777) -- (0.454,0.789) -- (0.407,0.800) -- (0.358,0.809) -- (0.309,0.817) -- (0.259,0.824) -- (0.208,0.830) -- (0.156,0.835) -- (0.105,0.838) -- (0.052,0.840) -- (0.000,0.840) -- (-0.052,0.840) -- (-0.105,0.838) -- (-0.156,0.835) -- (-0.208,0.830) -- (-0.259,0.824) -- (-0.309,0.817) -- (-0.358,0.809) -- (-0.407,0.800) -- (-0.454,0.789) -- (-0.500,0.777) -- (-0.545,0.765) -- (-0.588,0.751) -- (-0.629,0.736) -- (-0.669,0.720) -- (-0.707,0.703) -- (-0.743,0.685) -- (-0.777,0.666) -- (-0.809,0.647) -- (-0.839,0.627) -- (-0.866,0.606) -- (-0.891,0.584) -- (-0.914,0.562) -- (-0.934,0.539) -- (-0.951,0.516) -- (-0.966,0.492) -- (-0.978,0.468) -- (-0.988,0.444) -- (-0.995,0.420) -- (-0.999,0.395) -- (-1.000,0.371) -- (-0.999,0.346) -- (-0.995,0.322) -- (-0.988,0.297) -- (-0.978,0.273) -- (-0.966,0.249) -- (-0.951,0.226) -- (-0.934,0.203) -- (-0.914,0.180) -- (-0.891,0.158) -- (-0.866,0.136) -- (-0.839,0.115) -- (-0.809,0.095) -- (-0.777,0.075) -- (-0.743,0.057) -- (-0.707,0.039) -- (-0.669,0.022) -- (-0.629,0.006) -- (-0.588,-0.009) -- (-0.545,-0.023) -- (-0.500,-0.036) -- (-0.454,-0.047) -- (-0.407,-0.058) -- (-0.358,-0.067) -- (-0.309,-0.076) -- (-0.259,-0.083) -- (-0.208,-0.088) -- (-0.156,-0.093) -- (-0.105,-0.096) -- (-0.052,-0.098) -- (-0.000,-0.099) -- (0.052,-0.098) -- (0.105,-0.096) -- (0.156,-0.093) -- (0.208,-0.088) -- (0.259,-0.083) -- (0.309,-0.076) -- (0.358,-0.067) -- (0.407,-0.058) -- (0.454,-0.047) -- (0.500,-0.036) -- (0.545,-0.023) -- (0.588,-0.009) -- (0.629,0.006) -- (0.669,0.022) -- (0.707,0.039) -- (0.743,0.057) -- (0.777,0.075) -- (0.809,0.095) -- (0.839,0.115) -- (0.866,0.136) -- (0.891,0.158) -- (0.914,0.180) -- (0.934,0.203) -- (0.951,0.226) -- (0.966,0.249) -- (0.978,0.273) -- (0.988,0.297) -- (0.995,0.322) -- (0.999,0.346) -- (1.000,0.371);
        \draw[gray!60,line width=0.3pt] (0.703,0.262) -- (0.702,0.279) -- (0.699,0.297) -- (0.694,0.314) -- (0.688,0.331) -- (0.679,0.348) -- (0.669,0.364) -- (0.656,0.380) -- (0.642,0.396) -- (0.626,0.412) -- (0.609,0.427) -- (0.590,0.442) -- (0.569,0.456) -- (0.546,0.470) -- (0.522,0.483) -- (0.497,0.496) -- (0.470,0.507) -- (0.442,0.519) -- (0.413,0.529) -- (0.383,0.539) -- (0.352,0.548) -- (0.319,0.556) -- (0.286,0.564) -- (0.252,0.570) -- (0.217,0.576) -- (0.182,0.581) -- (0.146,0.585) -- (0.110,0.588) -- (0.073,0.590) -- (0.037,0.592) -- (0.000,0.592) -- (-0.037,0.592) -- (-0.073,0.590) -- (-0.110,0.588) -- (-0.146,0.585) -- (-0.182,0.581) -- (-0.217,0.576) -- (-0.252,0.570) -- (-0.286,0.564) -- (-0.319,0.556) -- (-0.352,0.548) -- (-0.383,0.539) -- (-0.413,0.529) -- (-0.442,0.519) -- (-0.470,0.507) -- (-0.497,0.496) -- (-0.522,0.483) -- (-0.546,0.470) -- (-0.569,0.456) -- (-0.590,0.442) -- (-0.609,0.427) -- (-0.626,0.412) -- (-0.642,0.396) -- (-0.656,0.380) -- (-0.669,0.364) -- (-0.679,0.348) -- (-0.688,0.331) -- (-0.694,0.314) -- (-0.699,0.297) -- (-0.702,0.279) -- (-0.703,0.262) -- (-0.702,0.245) -- (-0.699,0.228) -- (-0.694,0.211) -- (-0.688,0.194) -- (-0.679,0.177) -- (-0.669,0.160) -- (-0.656,0.144) -- (-0.642,0.128) -- (-0.626,0.112);
        \draw[gray!60,line width=0.3pt] (0.626,0.112) -- (0.642,0.128) -- (0.656,0.144) -- (0.669,0.160) -- (0.679,0.177) -- (0.688,0.194) -- (0.694,0.211) -- (0.699,0.228) -- (0.702,0.245) -- (0.703,0.262);
        \draw[gray!60,line width=0.3pt] (0.579,0.014) -- (0.577,0.028) -- (0.573,0.043) -- (0.567,0.057) -- (0.560,0.070) -- (0.552,0.084) -- (0.541,0.098) -- (0.530,0.111) -- (0.517,0.124) -- (0.502,0.136) -- (0.486,0.148) -- (0.469,0.160) -- (0.451,0.171) -- (0.431,0.182) -- (0.410,0.193) -- (0.388,0.202) -- (0.365,0.212) -- (0.341,0.220) -- (0.316,0.228) -- (0.290,0.236) -- (0.263,0.243) -- (0.236,0.249) -- (0.208,0.254) -- (0.179,0.259) -- (0.150,0.263) -- (0.121,0.266) -- (0.091,0.269) -- (0.061,0.271) -- (0.030,0.272) -- (0.000,0.272) -- (-0.030,0.272) -- (-0.061,0.271) -- (-0.091,0.269) -- (-0.121,0.266) -- (-0.150,0.263) -- (-0.179,0.259) -- (-0.208,0.254) -- (-0.236,0.249) -- (-0.263,0.243) -- (-0.290,0.236) -- (-0.316,0.228) -- (-0.341,0.220) -- (-0.365,0.212) -- (-0.388,0.202) -- (-0.410,0.193) -- (-0.431,0.182) -- (-0.451,0.171) -- (-0.469,0.160) -- (-0.486,0.148) -- (-0.502,0.136) -- (-0.517,0.124) -- (-0.530,0.111) -- (-0.541,0.098) -- (-0.552,0.084) -- (-0.560,0.070) -- (-0.567,0.057) -- (-0.573,0.043) -- (-0.577,0.028) -- (-0.579,0.014);
        \draw[gray!60,line width=0.3pt] (0.569,-0.068) -- (0.546,-0.055) -- (0.522,-0.041) -- (0.497,-0.029) -- (0.470,-0.017) -- (0.442,-0.006) -- (0.413,0.005) -- (0.383,0.015) -- (0.352,0.024) -- (0.319,0.032) -- (0.286,0.039) -- (0.252,0.046) -- (0.217,0.052) -- (0.182,0.057) -- (0.146,0.061) -- (0.110,0.064) -- (0.073,0.066) -- (0.037,0.067) -- (0.000,0.068) -- (-0.037,0.067) -- (-0.073,0.066) -- (-0.110,0.064) -- (-0.146,0.061) -- (-0.182,0.057) -- (-0.217,0.052) -- (-0.252,0.046) -- (-0.286,0.039) -- (-0.319,0.032) -- (-0.352,0.024) -- (-0.383,0.015) -- (-0.413,0.005) -- (-0.442,-0.006) -- (-0.470,-0.017) -- (-0.497,-0.029) -- (-0.522,-0.041) -- (-0.546,-0.055) -- (-0.569,-0.068);
        \draw[gray!60,line width=0.3pt] (-1.049,-0.620) -- (-1.008,-0.645) -- (-0.964,-0.670) -- (-0.917,-0.693) -- (-0.868,-0.715) -- (-0.816,-0.735) -- (-0.762,-0.755) -- (-0.706,-0.773) -- (-0.648,-0.790) -- (-0.589,-0.805) -- (-0.528,-0.818) -- (-0.465,-0.831) -- (-0.401,-0.841) -- (-0.336,-0.850) -- (-0.270,-0.858) -- (-0.203,-0.864) -- (-0.136,-0.868) -- (-0.068,-0.870) -- (-0.000,-0.871) -- (0.068,-0.870) -- (0.136,-0.868) -- (0.203,-0.864) -- (0.270,-0.858) -- (0.336,-0.850) -- (0.401,-0.841) -- (0.465,-0.831) -- (0.528,-0.818) -- (0.589,-0.805) -- (0.648,-0.790) -- (0.706,-0.773) -- (0.762,-0.755) -- (0.816,-0.735) -- (0.868,-0.715) -- (0.917,-0.693) -- (0.964,-0.670) -- (1.008,-0.645) -- (1.049,-0.620);
        \draw[line width=0.5pt] (-1.418,-0.035) -- (-1.412,-0.070) -- (-1.403,-0.104) -- (-1.389,-0.139) -- (-1.372,-0.173) -- (-1.351,-0.206) -- (-1.326,-0.239) -- (-1.297,-0.271) -- (-1.265,-0.303) -- (-1.230,-0.333) -- (-1.191,-0.363) -- (-1.149,-0.392) -- (-1.104,-0.420) -- (-1.055,-0.446) -- (-1.004,-0.471) -- (-0.950,-0.495) -- (-0.894,-0.518) -- (-0.835,-0.539) -- (-0.773,-0.559) -- (-0.710,-0.577) -- (-0.645,-0.594) -- (-0.578,-0.609) -- (-0.509,-0.622) -- (-0.439,-0.634) -- (-0.368,-0.644) -- (-0.295,-0.652) -- (-0.222,-0.658) -- (-0.148,-0.663) -- (-0.074,-0.666) -- (-0.000,-0.667) -- (0.074,-0.666) -- (0.148,-0.663) -- (0.222,-0.658) -- (0.295,-0.652) -- (0.368,-0.644) -- (0.439,-0.634) -- (0.509,-0.622) -- (0.578,-0.609) -- (0.645,-0.594) -- (0.710,-0.577) -- (0.773,-0.559) -- (0.835,-0.539) -- (0.894,-0.518) -- (0.950,-0.495) -- (1.004,-0.471) -- (1.055,-0.446) -- (1.104,-0.420) -- (1.149,-0.392) -- (1.191,-0.363) -- (1.230,-0.333) -- (1.265,-0.303) -- (1.297,-0.271) -- (1.326,-0.239) -- (1.351,-0.206) -- (1.372,-0.173) -- (1.389,-0.139) -- (1.403,-0.104) -- (1.412,-0.070) -- (1.418,-0.035);
        \draw[line width=0.5pt] (0.579,0.014) -- (0.577,0.028) -- (0.573,0.043) -- (0.567,0.057) -- (0.560,0.070) -- (0.552,0.084) -- (0.541,0.098) -- (0.530,0.111) -- (0.517,0.124) -- (0.502,0.136) -- (0.486,0.148) -- (0.469,0.160) -- (0.451,0.171) -- (0.431,0.182) -- (0.410,0.193) -- (0.388,0.202) -- (0.365,0.212) -- (0.341,0.220) -- (0.316,0.228) -- (0.290,0.236) -- (0.263,0.243) -- (0.236,0.249) -- (0.208,0.254) -- (0.179,0.259) -- (0.150,0.263) -- (0.121,0.266) -- (0.091,0.269) -- (0.061,0.271) -- (0.030,0.272) -- (0.000,0.272) -- (-0.030,0.272) -- (-0.061,0.271) -- (-0.091,0.269) -- (-0.121,0.266) -- (-0.150,0.263) -- (-0.179,0.259) -- (-0.208,0.254) -- (-0.236,0.249) -- (-0.263,0.243) -- (-0.290,0.236) -- (-0.316,0.228) -- (-0.341,0.220) -- (-0.365,0.212) -- (-0.388,0.202) -- (-0.410,0.193) -- (-0.431,0.182) -- (-0.451,0.171) -- (-0.469,0.160) -- (-0.486,0.148) -- (-0.502,0.136) -- (-0.517,0.124) -- (-0.530,0.111) -- (-0.541,0.098) -- (-0.552,0.084) -- (-0.560,0.070) -- (-0.567,0.057) -- (-0.573,0.043) -- (-0.577,0.028) -- (-0.579,0.014);
        \fill[blue!25] (-0.000,-0.667) -- (0.056,-0.666) -- (0.111,-0.665) -- (0.167,-0.662) -- (0.222,-0.658) -- (0.277,-0.654) -- (0.331,-0.648) -- (0.385,-0.642) -- (0.439,-0.634) -- (0.491,-0.625) -- (0.543,-0.616) -- (0.543,-0.616) -- (0.543,-0.586) -- (0.541,-0.556) -- (0.539,-0.524) -- (0.536,-0.492) -- (0.531,-0.460) -- (0.526,-0.428) -- (0.520,-0.395) -- (0.513,-0.363) -- (0.505,-0.331) -- (0.496,-0.300) -- (0.496,-0.300) -- (0.449,-0.309) -- (0.401,-0.317) -- (0.352,-0.324) -- (0.303,-0.330) -- (0.253,-0.335) -- (0.203,-0.339) -- (0.152,-0.342) -- (0.102,-0.345) -- (0.051,-0.346) -- (-0.000,-0.347) -- (-0.000,-0.347) -- (-0.000,-0.379) -- (-0.000,-0.411) -- (-0.000,-0.444) -- (-0.000,-0.477) -- (-0.000,-0.510) -- (-0.000,-0.542) -- (-0.000,-0.575) -- (-0.000,-0.606) -- (-0.000,-0.637) -- (-0.000,-0.667) -- cycle;
        \draw[blue!70!black,line width=0.6pt] (-0.000,-0.667) -- (0.056,-0.666) -- (0.111,-0.665) -- (0.167,-0.662) -- (0.222,-0.658) -- (0.277,-0.654) -- (0.331,-0.648) -- (0.385,-0.642) -- (0.439,-0.634) -- (0.491,-0.625) -- (0.543,-0.616) -- (0.543,-0.616) -- (0.543,-0.586) -- (0.541,-0.556) -- (0.539,-0.524) -- (0.536,-0.492) -- (0.531,-0.460) -- (0.526,-0.428) -- (0.520,-0.395) -- (0.513,-0.363) -- (0.505,-0.331) -- (0.496,-0.300) -- (0.496,-0.300) -- (0.449,-0.309) -- (0.401,-0.317) -- (0.352,-0.324) -- (0.303,-0.330) -- (0.253,-0.335) -- (0.203,-0.339) -- (0.152,-0.342) -- (0.102,-0.345) -- (0.051,-0.346) -- (-0.000,-0.347) -- (-0.000,-0.347) -- (-0.000,-0.379) -- (-0.000,-0.411) -- (-0.000,-0.444) -- (-0.000,-0.477) -- (-0.000,-0.510) -- (-0.000,-0.542) -- (-0.000,-0.575) -- (-0.000,-0.606) -- (-0.000,-0.637) -- (-0.000,-0.667) -- cycle;
        \node[lab] at (0,-1.25) {parameter torus\\with finite mesh};
      \end{scope}

      % ---------- cell operations and repetition ----------
      \node[draw=orange!85!black,rounded corners=6pt,line width=0.5pt,inner sep=0pt,
            minimum width=3.15cm,minimum height=4.4cm] (repeatbox) at (0.6,1.6) {};
      \node[draw=black!65,rounded corners=6pt,line width=0.5pt,inner sep=0pt,
            minimum width=2.8cm,minimum height=3.6cm] (cellbox) at (0.6,1.45) {};
      \node[font=\scriptsize] at ($(cellbox.north)+(0,0.21)$) {repeat for all cells};
      \begin{scope}[shift={(0.6,1.80)}]
        \node[draw=blue!70!black,line width=0.6pt,inner sep=0pt,outer sep=0pt,
              minimum width=2.25cm,minimum height=2.25cm] (meshcell) at (0,0) {};
        \fill (-0.9,-0.9) circle (2pt);
        \fill ( 0.9,-0.9) circle (2pt);
        % (1) phase accumulation
        \draw[arrow] (-0.74,-0.74) -- (0.74,-0.74) -- (0.74,0.74) -- (-0.74,0.74) -- (-0.74,-0.50);
        \node[lab] at (0,0) {(1) phase\\accumulation};
        % (2) state transport
        \draw[arrow] (-0.9,-1.07) .. controls (-0.36,-1.33) and (0.36,-1.33) .. (0.9,-1.07);
        \node[lab] at (0,-1.50) {(2) state transport};
      \end{scope}
      \draw[blue!70!black,line width=0.4pt,dashed] (-2.9+0.85*0.496,1.45-0.85*0.300) -- (meshcell.north west);
      \draw[blue!70!black,line width=0.4pt,dashed] (-2.9+0.85*0.543,1.45-0.85*0.616) -- (meshcell.south west);

      % ---------- circuit ----------
      \begin{scope}[shift={(-3.3,-1.3)}]
        \node[font=\small,left] at (0,0.6) {\(\ket{+}\)};
        \node[font=\small,left] at (0,0) {\(\ket{u(\mathbf k_0)}^{\otimes2}\)};
        \draw[wire] (0,0.6) -- (4.75,0.6);
        \draw[wire] (0,0) -- (3.4,0);
        \draw[wire] (0.08,0.51) -- (0.20,0.69);
        \draw[wire] (0.08,-0.09) -- (0.20,0.09);
        \node[gate,draw=orange!85!black] (g1) at (0.5,0) {};
        \node[gate,draw=orange!85!black] (g2) at (1.3,0) {};
        \node[font=\small,fill=white,inner sep=2pt] at (2.0,0) {\(\cdots\)};
        \node[gate,draw=orange!85!black] (g3) at (2.7,0) {};
        \foreach \x/\gatename in {0.5/g1,1.3/g2,2.7/g3}{
          \fill (\x,0.6) circle (1.8pt);
          \draw[wire] (\x,0.6) -- (\gatename.north);
        }
        \node[gate,minimum width=0.8cm] (qft) at (3.55,0.6) {$\mathrm{QFT}^{\dagger}$};
        % meter
        \node[draw,fill=white,line width=0.5pt,minimum width=0.5cm,minimum height=0.42cm,inner sep=0pt] (meter) at (4.5,0.6) {};
        \draw[wire] (4.35,0.5) arc (180:0:0.15);
        \draw[wire,-{Latex[length=2pt]}] (4.5,0.5) -- (4.65,0.72);
        \node[font=\small,right=2pt] at (meter.east) {\(=C\)};
        % zoom lines
        \draw[orange!85!black,line width=0.4pt,dashed] (g1.north west) -- ($(repeatbox.south west)+(6pt,0)$);
        \draw[orange!85!black,line width=0.4pt,dashed] (g1.north east) -- ($(repeatbox.south east)+(-6pt,0)$);
      \end{scope}
  \end{tikzpicture}
  \end{minipage}%
  }\hfill
  \subfloat{%
  \label{fig:intro_overview_b}%
  \begin{minipage}[t]{0.48\textwidth}
    \centering
    \raggedright\textbf{\thesubfigure}\par\vspace{-2pt}
    \centering
\begin{tikzpicture}[
      x=1.29cm,y=1.164cm,
      arrow/.style={-{Latex[length=4pt]},line width=0.6pt},
      circuit/.style={align=center,font=\small},
      phase/.style={align=center,font=\scriptsize},
      full/.style={draw,rounded corners=2pt,align=center,
        inner sep=4pt,minimum width=8.4cm,text width=7.8cm,font=\small}
  ]
    % Panel canvas: 8.5 cm x 6.2 cm.
    \path[use as bounding box] (-4.25cm,-2.75cm) rectangle (4.25cm,3.45cm);
    \node[circuit] (accepting) at (-2.0,2.55)
      {Accepting circuit\\\(\pacc\ge2/3\)};
    \node[circuit] (rejecting) at (2.0,2.55)
      {Rejecting circuit\\\(\pacc\le1/3\)};

    \node[full] (fullfamily) at (0,1.2)
      {Local Hamiltonian family \(H(\mathbf{k})\)\\
       encoding the circuit};
    \draw[arrow] (accepting.south) -- (accepting.south |- fullfamily.north);
    \draw[arrow] (rejecting.south) -- (rejecting.south |- fullfamily.north);

    \fill[blue!6] (-3,-0.5) rectangle (0,0);
    \fill[blue!18] (-3,-0.5) rectangle (-1,0);
    \fill[gray!10] (0,-0.5) rectangle (3,0);
    \fill[gray!28] (1,-0.5) rectangle (3,0);
    \draw[line width=0.6pt] (-3,-0.5) -- (3,-0.5);
    \draw[dashed,line width=0.7pt] (0,-1.15) -- (0,0);
    \draw[arrow] (accepting.south |- fullfamily.south) -- (-2.0,0.08);
    \draw[arrow] (rejecting.south |- fullfamily.south) -- (2.0,0.08);
    \node[align=center,font=\scriptsize,fill=white,inner sep=1pt]
      at (0,0.4) {Chern number};

    \node[phase] at (-1.55,-0.88)
      {topological phase\\\(C=1\)};
    \node[phase] at (1.55,-0.88)
      {trivial phase\\\(C=0\)};
    \node[align=center,font=\scriptsize,fill=white,inner sep=1pt]
      at (0,-1.5) {gap closing};
    \draw[arrow] (3,-1.9) -- (-3,-1.9);
    \node[font=\scriptsize,below] at (-2.7,-1.95) {\(\pacc=1\)};
    \node[font=\scriptsize,below] at (2.7,-1.95) {\(\pacc=0\)};
    \node[font=\small,below] at (0,-1.95) {\(\pacc\)};
  \end{tikzpicture}
  \end{minipage}%
  }
  \caption{Quantum algorithm and computational hardness of the Chern number $C$
  for a nondegenerate ground state on a parameter torus.
  \textbf{(a)} Generalized quantum signal processing accumulates fractional
  geometric phases cell by cell and transports two ground-state copies
  from the lower-left vertex of each cell to that of the next cell in the scan,
  using each copy in turn as the other's energy reference.
  An inverse quantum Fourier transform recovers $C$ with bounded error.
  For $L_x,L_y\geq\Delta_{\min}/2$, the oracle evolution time
  $\widetilde O\qty(L_xL_y/\Delta_{\min}^3)$ is worst-case optimal up to
  polylogarithmic factors, even for $C\in\qty{0,1}$; state-preparation costs
  are not included in the oracle time.
  Here $L_\mu$ bounds $\norm{\partial_{k_\mu}H}$, and $\Delta_{\min}>0$
  is a known spectral-gap lower bound.
  \textbf{(b)} Feynman--Kitaev encoding makes circuit acceptance control an
  effective Qi--Wu--Zhang mass, yielding $C=1$ for $\pacc\geq2/3$ and $C=0$
  for $\pacc\leq1/3$.
  An exact Schrieffer--Wolff transformation and gapped interpolation transfer
  this invariant to the full 5-local Hamiltonian.
  With polynomial norm and derivative bounds, an inverse-polynomial gap,
  and a supplied circuit preparing a guiding state with inverse-polynomial
  ground-state overlap at $\mathbf{k}_0$, exact computation from succinct local descriptions is
  in $\FBQP$ and $\BQP$-hard; its binary restriction is $\BQP$-complete.
  }
  \label{fig:intro_overview}
\end{figure*}
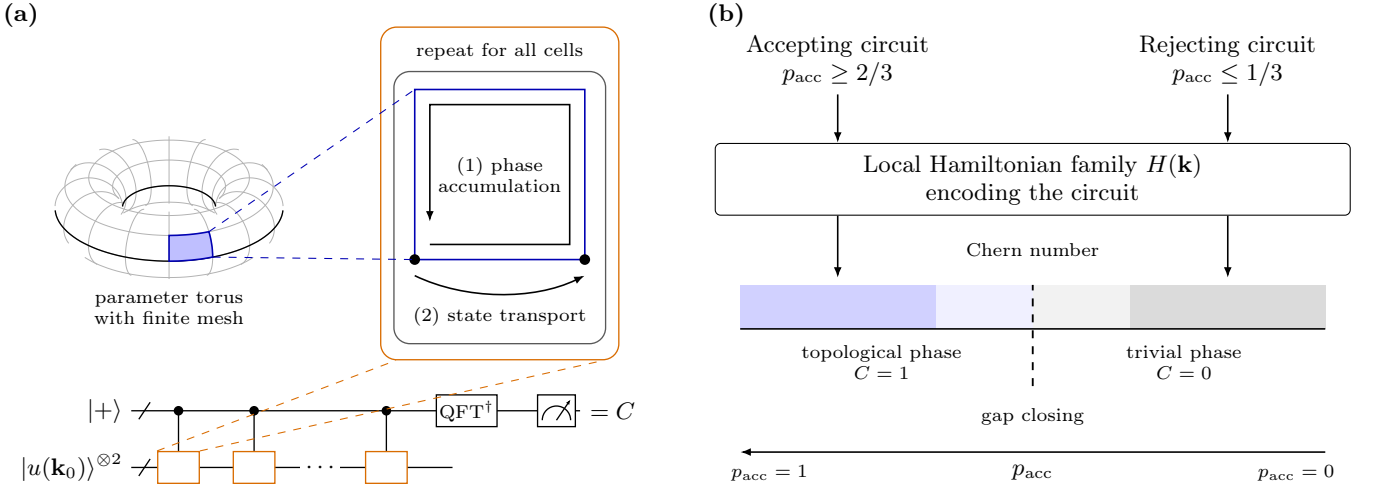

Topological phases of matter reveal distinctions between quantum states that cannot be captured by conventional local order parameters, motivating their description in terms of global invariants~\cite{hasan_kane_topological_2010}.
Among these invariants, the Chern number plays a central role by connecting the geometry of
parameter-dependent quantum states to observable properties of quantum
matter~\cite{simon_holonomy_1983,berry_quantal_1984}.
It determines quantized Hall and adiabatic pumping
responses~\cite{thouless_hall_1982,thouless_pump_1983,lohse_pump_2016,nakajima_pump_2016}
and the net chirality of edge modes in Chern
insulators~\cite{hatsugai_edge_1993}.
Chern topology also places fundamental constraints on how quantum states can vary across parameter
space: a nonzero band Chern number prevents the construction of an orthonormal basis of
exponentially localized Wannier states~\cite{brouder_wannier_2007} and imposes lower bounds on the integrated quantum metric, which quantifies how rapidly quantum states change under parameter variations~\cite{ozawa_mera_metric_2021}.
These geometric constraints can bound the superfluid weight in topological flat-band
systems~\cite{peotta_superfluidity_2015}.

For interacting systems, single-particle band topology is generally insufficient to characterize
the many-body state.
A corresponding invariant can instead be defined by varying phases inserted across the periodic
boundaries; the resulting many-body Chern number characterizes the quantized Hall
response~\cite{niu_thouless_wu_1985,goldman_ozawa_hall_2024}.
Recent studies of interaction-enabled
pumping~\cite{viebahn_interactions_pumping_2024}
and correlated Chern phases in moir\'e materials, including MoTe$_2$ and
graphene~\cite{wang_fci_mote2_2024,lu_graphene_fqah_2024},
motivate computing these invariants in interacting systems.
Such computations can also help assess the stability of correlated topological phases against
competing orders~\cite{chen_nonabelian_mote2_2025,he_competing_mote2_2025}.
Finite-size effects and the computational cost of resolving interacting many-body states
motivate methods capable of estimating Chern numbers in larger systems.
Reliable estimates at these scales could clarify the emergence and stability of correlated
topological phases and guide the design of materials and quantum simulators with desired
topological properties.

Classical methods for Chern-number estimation include discretized eigenstate
overlaps~\cite{FukuiHatsugaiSuzuki2005}, real-space markers for band
systems~\cite{bianco_resta_marker_2011}, and Green-function formulations for interacting
insulators~\cite{wang_zhang_green_2012}.
Many-body calculations employ exact diagonalization on finite
clusters~\cite{wang_fci_mote2_2024,chen_nonabelian_mote2_2025} and density-matrix renormalization
group methods that extract quantized responses through flux
insertion~\cite{grushin_fci_2015,he_competing_mote2_2025}.
Several strategies reduce the computational cost: gauge-invariant discretizations can recover the
integer on coarse meshes~\cite{FukuiHatsugaiSuzuki2005}, tensor networks compress
suitably structured states~\cite{orus_tensor_networks_2014}, and boundary-twist insensitivity can permit accurate estimates without
full twist-space integration in sufficiently large gapped
systems~\cite{watanabe_twist_2018,kudo_chern_2019}.
Despite these advances, the exponential growth of the many-body Hilbert space presents a potential
obstacle to classical calculations, suggesting that quantum computation could offer an advantage.
However, quantization and geometric structure can simplify Chern-number estimation, so a quantum
advantage is not evident from the Hilbert-space dimension alone.

This raises central questions: can quantum algorithms achieve a quantum advantage in Chern-number
estimation, how efficiently can they perform this task, and what fundamental limits constrain their
performance?
For a nondegenerate eigenstate followed around a closed loop of Hamiltonian parameters, the
Berry phase is the geometric part of the acquired phase, distinct from the energy-dependent
dynamical phase~\cite{berry_quantal_1984}.
Quantum approaches to the related problem of Berry-phase estimation include gate-based adiabatic
simulation~\cite{murta_berry_2020}, variational methods~\cite{tamiya_berry_2021}, and techniques for
suppressing adiabatic phase errors~\cite{kiumi_adiabatic_2026a}.
A decision formulation is $\BQP$-complete under
suitable gap, regularity, and guiding-state promises~\cite{hayakawa_berry_complexity_2025}.
That result concerns a generally continuous phase defined modulo a full turn; it does not establish the
complexity of determining an integer Chern number, which is invariant under gap-preserving
deformations. More directly, Ref.~\cite{Niedermeier2024} proposes adiabatic circuits for Chern-number estimation
and demonstrates them for noninteracting two-band models.
Although that work discusses finite-time and measurement errors, it does not establish an overall
runtime bound in terms of the spectral gap, Hamiltonian derivatives, and target failure probability.
Establishing both constructive runtime guarantees and matching lower bounds is therefore essential
to determining the optimal quantum resources for Chern-number estimation.

We obtain two main results for the same Chern-number problem,
using distinct computational models.

\paragraph*{\textnormal{\textbf{Nearly optimal quantum algorithm.}}}
In the Hamiltonian-oracle model, our algorithm returns the exact integer Chern number
with a prescribed target failure probability (\cref{thm:nonadiabatic_algorithm}).
Up to polylogarithmic factors in the inverse failure probability and problem parameters,
its evolution time is bounded by the product of the two first-derivative bounds divided
by the cube of the spectral-gap lower bound.
The costs of ground-state preparation at a single fixed reference point
and of input-independent gates are not included in this oracle-time bound.
For any fixed positive target failure probability below one half, we prove a matching
worst-case lower bound under the same oracle access and cost convention
(\cref{thm:chern_lower_bound}).
The algorithm is therefore optimal up to polylogarithmic factors, and the lower bound
persists even when the Chern number is promised to be either zero or one
(\cref{cor:binary_chern_lower_bound}).

Our algorithm replaces the adiabatic transport of previous Chern-number estimation
circuits~\cite{Niedermeier2024} with generalized-quantum-signal-processing (GQSP),
requiring no adiabatic approximation.
It prepares two ground-state copies at a single reference point and transports them
along the parameter torus, using each as a local energy reference for the other,
without further state-preparation calls.
Since the total geometric phase is an integer number of full turns and is invisible
to direct accumulation, we instead accumulate fractional phases coherently over a finite mesh
and read out the integer with an inverse quantum Fourier transform (QFT);
see \cref{fig:intro_overview_a}.

\paragraph*{\textnormal{\textbf{Complexity from a succinct classical input.}}}
For local Hamiltonian families given by succinct classical descriptions under the promises
of \cref{def:guided_cne}, exact Chern-number computation is in $\FBQP$, the
function analogue of $\BQP$, and is $\BQP$-hard
(\cref{thm:guided_cne_complexity}).
These promises include polynomial norm and first-derivative bounds, an inverse-polynomial
gap, and a supplied guiding-state circuit with inverse-polynomial overlap.
Restricting the Chern number to zero or one gives a $\BQP$-complete decision problem
(\cref{cor:guided_cne_bqp_complete}).
The hardness proof encodes a quantum circuit into a local Hamiltonian family whose
ground-state Chern number is $1$ or $0$ according to whether the circuit accepts or
rejects (\cref{fig:intro_overview_b}).
Here uniform quantum gate complexity includes state preparation and Hamiltonian simulation
from the classical input.
Our $\BQP$-hardness result gives complexity-theoretic evidence of a quantum advantage:
unless $\BPP=\BQP$, no classical polynomial-time algorithm solves this promise problem in the worst-case.

Together, these results open a route to using quantum computers to identify and characterize interacting topological phases 
that are difficult to study classically.

This paper is organized as follows.
In \cref{sec:chern_definition}, we introduce the problem setting and state our main results.
In \cref{sec:upper_bound_algorithm}, we present the quantum algorithm for Chern-number estimation.
We establish the oracle lower bound in \cref{sec:chern_lower_bound} and study the computational
complexity of the exact function problem in \cref{sec:bqp_hardness}.
Finally, \cref{sec:discussion} discusses the implications of our results and directions for future
work.

\section{Problem Setting and Main Results}
\label{sec:chern_definition}

\subsection{Problem setting and computational models}

\noindent\textbf{Chern-number estimation.}
We consider a smooth finite-dimensional Hamiltonian family $H(\mathbf{k})$ parametrized by a
two-dimensional torus $X:=\mathbb{T}^2$, represented by $[0,2\pi]^2$ with periodic boundary
conditions.
Here $\mathbf{k}:=(k_x,k_y)$ denotes general periodic parameters, such as the momentum
of a Bloch Hamiltonian or boundary twist angles of an interacting many-body system.
Let $E_0(\mathbf{k})$ denote the ground-state energy.
Throughout the main cost analysis, we assume that the ground state is nondegenerate and separated
from the remaining spectrum by a positive gap.
Writing $E_m(\mathbf{k})$ for the remaining eigenvalues, define
\begin{align}
    \Delta(\mathbf{k})
    &:=
    \min_{m\ne 0}\abs{E_m(\mathbf{k})-E_0(\mathbf{k})},
    \label{eq:gap_def}
\end{align}
and let $\Delta_{\min}>0$ denote a supplied lower bound with
$\Delta(\mathbf{k})\geq\Delta_{\min}$ for all $\mathbf{k}\in X$.
Only this bound enters the algorithm and the cost analysis; the true
minimum gap is never required.

The Berry connection and Berry curvature are
\begin{align}
    A_\mu(\mathbf{k})
    &:= i\ev{\partial_{k_\mu}}{u(\mathbf{k})},
    \qquad \mu\in\{x,y\},
    \\
    \Omega(\mathbf{k})
    &:=
    \partial_{k_x}A_y(\mathbf{k})
    -
    \partial_{k_y}A_x(\mathbf{k}),
    \label{eq:berry_connection_curvature_def}
\end{align}
where $\ket{u(\mathbf{k})}$ is a normalized ground state chosen in a smooth local gauge.
A local gauge fixes the otherwise arbitrary phase of this eigenstate
smoothly on a neighborhood.  When the Chern number is nonzero, no single
smooth periodic choice exists over the whole torus.
The Chern number is
\begin{align}
    C
    &:=
    \frac{1}{2\pi}
    \int_X d^2\mathbf{k}\,\Omega(\mathbf{k})
    \in\mathbb{Z}.
    \label{eq:chern_number_def}
\end{align}

Let $P(\mathbf{k}):=\ket{u(\mathbf{k})}\bra{u(\mathbf{k})}$ denote the ground-state projector.
The Berry curvature has the gauge-invariant representation
\begin{align}
    \Omega(\mathbf{k})
    &=
    i\operatorname{Tr}\qty(
        P[\partial_{k_x}P,\partial_{k_y}P]
    ).
    \label{eq:berry_curvature_projector_form}
\end{align}
Let $L_x,L_y>0$ be supplied upper bounds on the Hamiltonian derivatives,
satisfying
\begin{align}
    \norm{\partial_{k_\mu}H(\mathbf{k})}
    &\leq L_\mu,
    \qquad \mu\in\qty{x,y},
    \label{eq:hamiltonian_derivative_bounds}
\end{align}
for every $\mathbf{k}\in X$.

We study the same ground-state Chern-number problem, with invariant $C[P]:=C$,
in two computational models.
\Cref{thm:nonadiabatic_algorithm,thm:chern_lower_bound} concern
Hamiltonian-oracle complexity under a common access model, whereas
\cref{thm:guided_cne_complexity} concerns the succinct classical-input
complexity of the same invariant.
For estimation, the required output $\widehat C$ satisfies $\Pr[|\widehat C-C|<1/2]\geq1-\eta$;
rounding then recovers the integer $C$.

\noindent\textbf{Hamiltonian-oracle model.}
The model provides coherent parameter addressing, controlled forward and
inverse Hamiltonian evolution, and ground-state preparation only
at the fixed reference point $\mathbf{k}_0:=(0,0)$.
Coherent parameter addressing means that a quantum address register
stores a mesh-point index and selects the corresponding Hamiltonian
acting on the system register. The address may be in a superposition:
each branch undergoes $e^{\pm iH(\mathbf{k})t}$ for its specified
$\mathbf{k}$, preserving coherence between branches without measuring
the address register. \Cref{eq:alg_hamiltonian_oracle} gives the
corresponding block-diagonal oracle.

The algorithm and lower bound in
\cref{thm:nonadiabatic_algorithm,thm:chern_lower_bound}
use the following common Hamiltonian-oracle model.
Let $\mathcal O_H$ denote the oracle in
\cref{eq:alg_hamiltonian_oracle}.
Introducing a control qubit $\mathrm c$, define
\begin{align}
    \widetilde{\mathcal O}_H
    :=\ket{1}\!\bra{1}_{\mathrm c}\otimes\mathcal O_H.
    \label{eq:controlled_hamiltonian_oracle}
\end{align}
Evolution under this oracle implements controlled Hamiltonian
evolution; fixing the control qubit to $\ket1$ gives uncontrolled
evolution.
The algorithm evolves according to
\begin{align}
    i\frac{d}{dt}\ket{\psi_H(t)}
    &=\qty[H_D(t)+g(t)\widetilde{\mathcal O}_H]
      \ket{\psi_H(t)},
    \notag\\
    g(t)&\in\mathbb R,\qquad \abs{g(t)}\leq1,
    \label{eq:hamiltonian_oracle_evolution}
\end{align}
where $H_D(t)$ is an arbitrary input-independent
driver~\cite{farhi_quantum_2008,mochon_hamiltonian_2007}.
Negative $g(t)$ permits inverse evolution.
Input-independent gates can implement coherent control conditions
and route the oracle to different data registers.

We measure Hamiltonian-oracle cost by the total evolution time
supplied by the oracle.
For final physical time $T$, define
\begin{align}
    T_{\mathrm{oracle}}
    :=\int_0^Tdt\,\abs{g(t)}.
    \label{eq:oracle_time_definition}
\end{align}
This resource measure includes neither an additional
Hamiltonian-norm factor nor the gate cost of implementing
Hamiltonian evolution.
Input-independent gates and driver evolution incur no
Hamiltonian-oracle-time cost.
The cost of ground-state preparation at $\mathbf{k}_0$ is not
included in $T_{\mathrm{oracle}}$.

The algorithm realizes its oracle calls within this model by
setting $H_D(t)=0$ and $g(t)=\pm1$ during each call, and applying
input-independent gates with $g(t)=0$.
Consequently, a forward or inverse oracle call of duration $|t|$
contributes $|t|$ to \cref{eq:oracle_time_definition}, including
when controlled or applied with coherent parameter addressing.

\noindent\textbf{Succinct classical-input model.}
The input consists of a succinct classical description of a local
Hamiltonian family with efficiently evaluable coefficients and a
guiding-state preparation circuit, under the promises of
\cref{def:guided_cne}: inverse-polynomial gap and guiding-state overlap,
and polynomial norm and first-derivative bounds.
We consider uniform polynomial-time quantum computation, counting the
gate costs of state preparation and Hamiltonian simulation from the
classical input. Hamiltonian evolution is not supplied as an oracle.

\noindent\textbf{Complexity classes.}
The class \(\BQP\) consists of promise decision problems solvable with
bounded error by a uniform polynomial-size quantum circuit family.
Its function analogue \(\FBQP\) consists of functions whose classical
output is produced correctly with bounded error by such a circuit family.
We call a function problem \(\BQP\)-hard when every problem in \(\BQP\)
reduces to evaluating it with polynomial-time classical pre- and
postprocessing.
All complexity statements below are understood on inputs satisfying the
stated promises.

\subsection{Hamiltonian-oracle upper and lower bounds}

The explicit algorithm is described in
\cref{sec:upper_bound_algorithm}.
The early-return condition is $\min\qty{L_x,L_y}<\Delta_{\min}/2$,
as in \cref{eq:alg_zero_check}.
Here $\widetilde O$ omits polylogarithmic factors in the inverse
failure probability $1/\eta$ and the problem parameters.

\begin{theorem}[Chern-number estimation algorithm]
\label{thm:nonadiabatic_algorithm}
Let $H(\mathbf{k})$ be a smooth finite-dimensional Hamiltonian
family on $X:=\mathbb T^2$. Suppose the following conditions hold:
\begin{enumerate}
    \item The ground state is nondegenerate throughout $X$.

    \item Bounds $\Delta_{\min},H_{\max},L_x,L_y>0$ are supplied, with
    \begin{align}
        \Delta(\mathbf{k})&\geq\Delta_{\min},
        \notag\\
        \norm{H(\mathbf{k})}&\leq H_{\max},
        \notag\\
        \norm{\partial_{k_\mu}H(\mathbf{k})}&\leq L_\mu,
        \qquad \mu\in\qty{x,y},
        \label{eq:alg_input_bounds}
    \end{align}
    for every $\mathbf{k}\in X$.

    \item Controlled forward and inverse Hamiltonian evolution
    with coherent parameter addressing is available, together
    with a coherent circuit preparing the ground state
    $\ket{u(\mathbf{k}_0)}$ at $\mathbf{k}_0:=(0,0)$.
\end{enumerate}
Then \cref{alg:chern_estimation} returns the exact ground-state Chern number $C$ 
defined in \cref{eq:chern_number_def} with probability at least $1-\eta$ for every $0<\eta<1/2$.
For the supplied bounds and target failure probability $\eta$, let
$T_{\mathrm{upper}}$ denote the worst-case Hamiltonian-oracle time of
\cref{alg:chern_estimation}, measured according to
\cref{eq:oracle_time_definition}.
In the early-return case defined by
\cref{eq:alg_zero_check}, no Hamiltonian evolution is used, so
\begin{align}
    T_{\mathrm{upper}}=0.
    \label{eq:alg_runtime_zero}
\end{align}
In the quantum case,
\begin{align}
    T_{\mathrm{upper}}
    =\widetilde O\qty(\frac{L_xL_y}{\Delta_{\min}^3}).
    \label{eq:alg_runtime}
\end{align}
\end{theorem}

For the lower bound, write $a_\mu:=L_\mu/\Delta_{\min}$ for
$\mu\in\qty{x,y}$.

\begin{theorem}[Chern-number estimation lower bound]
\label{thm:chern_lower_bound}
Let $\mathfrak{H}$ be a class of smooth finite-dimensional Hamiltonian
families on $X:=\mathbb T^2$. Suppose the following conditions hold:
\begin{enumerate}
    \item For every $H\in\mathfrak{H}$, the ground state is nondegenerate
    throughout $X$.

    \item Bounds $\Delta_{\min},H_{\max},L_x,L_y>0$ are supplied, with
    \begin{align}
        \Delta(\mathbf{k})&\geq\Delta_{\min},
        \notag\\
        \norm{H(\mathbf{k})}&\leq H_{\max},
        \label{eq:lower_bound_norm_constraint}\\
        \norm{\partial_{k_\mu}H(\mathbf{k})}&\leq L_\mu,
        \qquad \mu\in\qty{x,y},
        \label{eq:lower_bound_derivative_constraints}
    \end{align}
    for every $\mathbf{k}\in X$ and every $H\in\mathfrak{H}$.

    \item Controlled forward and inverse Hamiltonian evolution
    with coherent parameter addressing is available, together
    with a coherent circuit preparing the ground state
    $\ket{u(\mathbf{k}_0)}$ at $\mathbf{k}_0:=(0,0)$.
    Ground-state preparation is allowed only at this reference point.
    The Hamiltonian-oracle access model and cost convention are defined in
    \cref{eq:hamiltonian_oracle_evolution,eq:oracle_time_definition}.

    \item If $a_x,a_y\geq1/2$, $\mathfrak{H}$ contains
    the corresponding two-band hard family constructed in
    \cref{sec:lower_bound_hard_instance}, either directly
    or after adding the same trivial unoccupied bands to every Hamiltonian
    in that family.
     In this case, the supplied values of $L_x,L_y,\Delta_{\min}$, and $H_{\max}$ are
    held fixed across the hard family.
    All its Hamiltonians are accessed on the same oracle grid defined in
    \cref{eq:alg_mesh}.
\end{enumerate}
For the supplied bounds and any fixed $0<\eta<1/2$, let
$T_{\mathrm{lower}}$ denote the lower bound specified below on the
worst-case Hamiltonian-oracle time of any algorithm that estimates the
ground-state Chern number to additive error less than $1/2$ with
probability at least $1-\eta$ on every $H\in\mathfrak{H}$.
In the early-return case $\min\qty{a_x,a_y}<1/2$, we have $C=0$
for every $H\in\mathfrak{H}$, and returning zero attains the lower bound
\begin{align}
    T_{\mathrm{lower}}=0.
    \label{eq:chern_lower_bound_theorem_zero}
\end{align}
In the quantum case $a_x,a_y\geq1/2$, every such algorithm has
worst-case Hamiltonian-oracle time
\begin{align}
    T_{\mathrm{lower}}
    &=c_\eta\frac{L_xL_y}{\Delta_{\min}^3},
    \label{eq:chern_lower_bound_theorem_result}\\
    c_\eta&:=\frac{1-2\sqrt{\eta(1-\eta)}}4>0.
    \label{eq:lower_bound_uniform_constant}
\end{align}
The constant $c_\eta$ is independent of $H_{\max}$ and of the
distances of $a_x,a_y$ from $1/2$.
\end{theorem}

Both theorems use coherent parameter addressing, controlled forward and
inverse Hamiltonian evolution, and ground-state preparation only at
$\mathbf{k}_0:=(0,0)$, with the same Hamiltonian-oracle-time accounting.
Neither oracle-time bound includes the costs of state preparation
or input-independent gates.
This makes the bounds directly comparable.
On the input classes covered by \cref{thm:chern_lower_bound}, the algorithm
is optimal up to polylogarithmic factors for general integer Chern-number
estimation. The early-return regime uses no Hamiltonian evolution,
as stated in \cref{eq:alg_runtime_zero,eq:chern_lower_bound_theorem_zero}.
The following corollary gives the same asymptotic lower bound for the
binary promise under an additional hard-family inclusion assumption.

\begin{corollary}[Binary-Chern lower bound]
\label{cor:binary_chern_lower_bound}
Under the assumptions of \cref{thm:chern_lower_bound}, with $a_x,a_y\geq1/2$, suppose additionally
that $\mathfrak{H}$ contains the binary-Chern hard family of
\cref{sec:lower_bound_binary_family}, with the same convention for
multiband embedding.
Even if the algorithm is required to succeed only on inputs in
$\mathfrak{H}$ satisfying $C\in\qty{0,1}$, its worst-case oracle time obeys
\begin{align}
    T_{\mathrm{lower}}
    &=\frac{c_\eta}{3}\frac{L_xL_y}{\Delta_{\min}^3}
    \notag\\
    &=\frac{1-2\sqrt{\eta(1-\eta)}}{12}
        \frac{L_xL_y}{\Delta_{\min}^3}.
    \label{eq:binary_chern_lower_bound}
\end{align}
Thus the same asymptotic lower bound holds under this binary promise,
with the same oracle and reference-state access assumptions.
\end{corollary}

\subsection{Hardness results for guided Chern-number estimation}

\Cref{thm:nonadiabatic_algorithm,thm:chern_lower_bound} characterize
Hamiltonian-oracle complexity. We now consider the same Chern-number
computation task with a succinct classical Hamiltonian description and
the gate-cost accounting of this model.

\begin{definition}[Guided exact Chern-number computation problem]
\label{def:guided_cne}
\leavevmode\\
\noindent\textbf{Input.}
  For a fixed constant \(k=O(1)\), a succinct classical description of a
  smooth periodic \(k\)-local
  \(n\)-qubit Hamiltonian family
  \begin{align}
    H(\mathbf{k})&:=\sum_{\ell=1}^{L}h_\ell(\mathbf{k}),
    \qquad \mathbf{k}\in X,\qquad L=\poly(n),
  \end{align}
  given by the support of each local term and a classical algorithm that,
  on a \(b\)-bit specification of \(\mathbf{k}\), evaluates its coefficients
  to \(b\) bits in time \(\poly(n,b)\); and a polynomial-size circuit over a
  fixed universal gate set preparing a guiding state \(\ket c\).  The reference
  point is fixed at \(\mathbf{k}_0:=(0,0)\), rather than supplied as part of
  the input.

\medskip
\noindent\textbf{Promise.}
  The ground state is nondegenerate on all of \( X\) and
  \(\Delta_{\min}\ge1/\poly(n)\).  The Hamiltonian norm and its first 
  parameter derivatives are polynomially bounded.  At \(\mathbf{k}_0\),
  the guiding state has inverse-polynomial overlap,
\[
\abs{\braket{c}{u(\mathbf{k}_0)}}^2
  \ge \frac{1}{\poly(n)}.
\]
  The polynomial upper bounds on the norm and derivatives, and the
  inverse-polynomial lower bounds on the gap and guiding-state overlap,
  are fixed as part of the promise problem and are available to the algorithm.

\medskip
\noindent\textbf{Output.}
  Output the integer ground-state Chern number \(C[P]\) in signed binary.
\end{definition}

The following theorem classifies the complexity of this function problem.

\begin{theorem}[Complexity of guided Chern-number computation]
\label{thm:guided_cne_complexity}
The guided exact Chern-number computation problem of
\cref{def:guided_cne} belongs to \(\FBQP\) and is \(\BQP\)-hard.  The
hardness already holds for 5-local Hamiltonian families.
\end{theorem}

Restricting the Chern number to two values gives the corresponding decision
problem as an immediate consequence.

\begin{corollary}[\(\BQP\)-complete restriction]
\label{cor:guided_cne_bqp_complete}
Even for 5-local Hamiltonian families, restricting to the additional promise
\(C[P]\in\{0,1\}\) gives a \(\BQP\)-complete problem of deciding whether
\(C[P]=0\) or \(C[P]=1\).
\end{corollary}

We prove \cref{thm:guided_cne_complexity} by establishing membership in
\(\FBQP\) in \cref{sec:bqp_containment} and \(\BQP\)-hardness in
\cref{sec:bqp_hardness_reduction}.
\Cref{cor:guided_cne_bqp_complete} then follows immediately because the
hardness construction produces instances with \(C[P]=0\) or \(C[P]=1\).

\section{Quantum Algorithm for Chern Number Estimation}
\label{sec:upper_bound_algorithm}

We now prove \cref{thm:nonadiabatic_algorithm} by constructing the algorithm
for the Hamiltonian $H(\mathbf{k})$, using the bounds specified in
\cref{sec:chern_definition}.
Two ground-state copies are initially prepared at $\mathbf{k}_0$
and subsequently transported along the parameter torus.
Each copy in turn serves as an energy reference for the other,
with the reference energy given by the ground-state energy
at the current parameter point of the reference copy.
The same two copies are reused throughout, without additional
ground-state preparation.

The construction has three parts. First, we express $C$ exactly as a sum
of principal geometric phases on a finite mesh. Second, we synthesize
ground-state transport and fractional phase modulation from Hamiltonian
evolution using GQSP. Third, we combine these
operations into closed scans and recover the signed integer by an inverse QFT. 

\subsection{Finite-mesh formulation and local-phase identity}
\label{sec:alg_mesh}

We first reduce Chern-number estimation to an exact sum of geometric
phases associated with the cells of a finite mesh.
The gap and first derivatives control both the ground-state geometry and
the required mesh resolution:
\begin{align}
    \norm{\partial_{k_\mu}P(\mathbf{k})}
    &\leq\frac{\norm{\partial_{k_\mu}H(\mathbf{k})}}{\Delta(\mathbf{k})}
    \leq\frac{L_\mu}{\Delta_{\min}}, \quad \mu\in\qty{x,y},
    \label{eq:alg_projector_bound}\\
    \abs{\Omega(\mathbf{k})}
    &\leq\frac{2L_xL_y}{\Delta_{\min}^2}.
    \label{eq:berry_curvature_hamiltonian_bound}
\end{align}
Their derivation appears in \cref{app:alg_geometry}.
Integrating \cref{eq:berry_curvature_hamiltonian_bound} over the torus
gives the integer bound
\begin{align}
    \abs{C}\leq C_{\max}
    :=\left\lceil\frac{4\pi L_xL_y}{\Delta_{\min}^2}\right\rceil.
    \label{eq:alg_chern_bound}
\end{align}

Before constructing the mesh, we check whether the supplied bounds
already force the Chern number to vanish.
If
\begin{align}
    \min\qty{L_x,L_y}<\frac{\Delta_{\min}}{2},
    \label{eq:alg_zero_check}
\end{align}
the ground-state bundle admits a smooth periodic normalized section and
$C=0$ as proved in \cref{app:alg_geometry}.
The algorithm performs a classical check of this sufficient
condition and returns zero in the early-return case before preparing
any quantum state.

Otherwise, the algorithm proceeds with the quantum case and chooses the rectangular mesh
\begin{align}
    N_\mu&:=\left\lceil\frac{128\pi L_\mu}{\Delta_{\min}}\right\rceil,
    \quad \delta k_\mu:=\frac{2\pi}{N_\mu}, \quad \mu\in\{x,y\},
    \label{eq:alg_mesh}\\
    N_{\square}&:=N_xN_y
    =\Theta\qty(\frac{L_xL_y}{\Delta_{\min}^2}).
    \label{eq:alg_cell_count}
\end{align}
The mesh vertices are indexed by $j_\mu=0,\ldots,N_\mu-1$
for $\mu\in\qty{x,y}$, with indices interpreted modulo $N_\mu$.
We assume controlled forward and inverse Hamiltonian evolution with
coherent parameter addressing. The mesh vertices and the corresponding
multiplication oracle are
\begin{align}
    \mathbf k_{j_x,j_y}
    &:=\qty(j_x\,\delta k_x,j_y\,\delta k_y),
    \label{eq:alg_mesh_points}\\
    \mathcal O_H
    &:=\sum_{j_x,j_y}\ket{j_x,j_y}\bra{j_x,j_y}
       \otimes H(\mathbf k_{j_x,j_y}).
    \label{eq:alg_hamiltonian_oracle}
\end{align}
Oracle costs are measured by the total Hamiltonian evolution time
defined in \cref{eq:oracle_time_definition}.
Input-independent gates, including the inverse QFT, are taken to be
exact in this oracle model. The gate costs of implementing these
operations and Hamiltonian evolution are not included in the oracle time.

A coherent circuit prepares
$\ket{u(\mathbf{k}_0)}$ at $\mathbf{k}_0:=(0,0)$ and is applied twice.
The cost of these two preparation calls is not included in the oracle time.
All subsequent ground-state transfers are constructed from the Hamiltonian oracle.

The purpose of the chosen mesh resolution is to control the variation
of the Hamiltonian and its ground-state projector within each cell.
Any two points $v,w$ in the same rectangular cell of the mesh obey
\begin{align}
    \norm{H(v)-H(w)}&\leq\frac{\Delta_{\min}}{32},
    \notag\\
    \norm{P(v)-P(w)}&\leq\rho:=\frac{1}{32}
    \label{eq:alg_cell_bounds}
\end{align}
as proved in \cref{app:alg_geometry}.

We use the oriented cell boundaries to define the local geometric phases.
Write $\mathcal C_{j_x,j_y}:=(a,b,c,d)$ for the cell with
lower-left vertex $a:=\mathbf k_{j_x,j_y}$ and successive
vertices $b:=a+(\delta k_x,0)$, $c:=a+(\delta k_x,\delta k_y)$,
and $d:=a+(0,\delta k_y)$.
The order $a\to b\to c\to d\to a$ is counterclockwise.
Cells crossing a periodic seam are represented as rectangles in lifted
coordinates.

For each oriented cell, define the four-point Bargmann
invariant~\cite{mukunda_bargmann_2003}
and its principal phase by
\begin{align}
    B_{\mathcal C}
    &:=
    \operatorname{Tr}\qty[P(a)P(b)P(c)P(d)],
    &
    \phi_{\mathcal C}
    &:=
    \operatorname{Arg}B_{\mathcal C},
    \label{eq:alg_bargmann}
\end{align}
where $\operatorname{Arg}z\in(-\pi,\pi]$.
Equivalently, $B_{\mathcal C}$ is the product of the four
overlaps between ground states at adjacent
vertices of the cell $\mathcal C$ in counterclockwise order.
Independent phase choices at the vertices cancel in this
product, so $\phi_{\mathcal C}$ is the same gauge-invariant
phase that underlies lattice formulations of the
Chern number~\cite{FukuiHatsugaiSuzuki2005}.
It remains to show that summing these principal phases recovers
the Chern number exactly.

The mesh bound ensures that all overlaps are nonzero
and that $\abs{\phi_{\mathcal C}}<\pi/4$, as proved in
\cref{app:alg_bargmann_bounds}.
With the Berry-connection convention of
\cref{eq:berry_connection_curvature_def}, the exact identity is
\begin{align}
    2\pi C
    &=
    -\sum_{\mathcal C}\phi_{\mathcal C}.
    \label{eq:alg_phase_sum}
\end{align}
To obtain it, apply Stokes' theorem in smooth local gauges
to express $2\pi C$ as the sum of counterclockwise
Berry-connection integrals around all cells.
Pair the two oppositely oriented contributions on each
shared edge. For each cell $\mathcal C$, define $E_{\mathcal C}$
as the sum of these paired contributions on its right and
upper edges. Then
$2\pi C=\sum_{\mathcal C}E_{\mathcal C}$.
Here $E_{\mathcal C}$ groups shared-edge contributions
and need not equal the Berry-curvature integral through
$\mathcal C$.

In projector-based gauges anchored at the lower-left
vertices, these integrals reduce to endpoint differences
of transition phases, giving
$e^{iE_{\mathcal C}}=e^{-i\phi_{\mathcal C}}$.
The mesh bounds imply $|E_{\mathcal C}|<\pi$; together
with $|\phi_{\mathcal C}|<\pi/4$, this fixes
$E_{\mathcal C}=-\phi_{\mathcal C}$ as a real equality.
The full argument is in \cref{app:alg_phase_sum}.
The local operations below implement the phase accumulation required
by this identity.

\subsection{Local algorithmic primitives}
\label{sec:alg_primitives}

The global algorithm needs two local operations: a direct rotation
$D_{wv}$ that transports a ground state from vertex $v$ to vertex $w$
within the same cell,
and a fractional phase modulation $W_{\mathcal C}(\alpha)$ that returns
the anchor state with phase $e^{-i\alpha\phi_{\mathcal C}}$.
We construct them through three GQSP transformations.
\Cref{fig:oracle_hierarchy} summarizes the construction of the local
algorithmic primitives from Hamiltonian evolution.

During each complete local operation, one copy is held fixed as an
energy reference. Local error guarantees include all work registers
and hold uniformly over the other system and control registers when
this reference is exact.

\begin{figure}[bt]
    \centering
    \resizebox{0.3\textwidth}{!}{%
    \begin{tikzpicture}[
        x=1pt,y=1pt,
        box/.style={draw,rounded corners=2pt,align=center,
            inner sep=5pt,text width=155pt,font=\small},
        arrow/.style={-{Latex[length=4pt]},line width=0.6pt},
        tag/.style={fill=white,inner sep=2pt,align=center,font=\scriptsize}
    ]
        \node[box] (hamiltonian) at (0,0)
            {Hamiltonian evolution\\$e^{\pm iH(\mathbf{k})t}$};
        \node[box] (reflection) at (0,-65)
            {Ground-state reflection\\$\mathcal R_v:=2P(v)-I$};
        \node[box] (rotation) at (0,-130)
            {Direct rotation\\
            $D_{wv}:=\qty(\mathcal R_w\mathcal R_v)^{1/2}$};
        \node[box] (fractional) at (0,-195)
            {Fractional phase modulation\\
            $W_{\mathcal C}(\alpha):=\qty(D_{ad}D_{dc}D_{cb}D_{ba})^\alpha$};
        \draw[arrow] (hamiltonian) --
            node[tag]{Phase-step GQSP\\\cref{sec:alg_reflection}} (reflection);
        \draw[arrow] (reflection) --
            node[tag]{Square-root GQSP\\\cref{sec:alg_rotation}} (rotation);
        \draw[arrow] (rotation) --
            node[tag]{Fractional-power GQSP\\
            \cref{sec:alg_fractional}} (fractional);
    \end{tikzpicture}%
    }
    \caption{Construction of the local algorithmic primitives. Three GQSP
    transformations synthesize ground-state reflections, direct rotations,
    and fractional phase modulations from Hamiltonian evolution.
    Here $v,w$ are vertices of cell $\mathcal C$, and $a,b,c,d$ label
    its vertices in counterclockwise order.
    }
    \label{fig:oracle_hierarchy}
\end{figure}
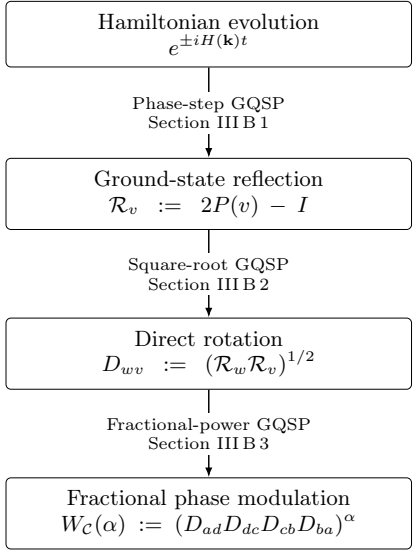

\subsubsection{Hamiltonian evolutions to ground-state reflection}
\label{sec:alg_reflection}

An exact ground-state copy at $a$ supplies an energy reference for a
reflection at any vertex $v$ in the same cell. Define
\begin{align}
    K_{v|a}:=I\otimes H(v)-H(a)\otimes I.
    \label{eq:alg_relative_hamiltonian}
\end{align}
On the sector with reference state $\ket{u(a)}$, its eigenvalues are
$E_m(v)-E_0(a)$. By the variational principle and
\cref{eq:alg_cell_bounds},
\begin{align}
    \abs{E_0(v)-E_0(a)}&\leq\frac{\Delta_{\min}}{32},
    \notag\\
    E_m(v)-E_0(a)&\geq\frac{31\Delta_{\min}}{32}\qquad(m\ne0).
    \label{eq:alg_relative_energy_bounds}
\end{align}
Thus the reference separates the
ground state from the excited states without estimating $E_0(a)$.
The relative evolution can be implemented directly with the available
Hamiltonian oracles.
The commuting factors give
\begin{align}
    e^{-iK_{v|a}t}
    =e^{+iH(a)t}\otimes e^{-iH(v)t},
    \label{eq:alg_relative_evolution}
\end{align}
at oracle time $2|t|$.

We convert this energy separation into a reflection by a GQSP
approximation to a sign function.
Shift the relative energy by the threshold $\Delta_{\min}/2$ and set
\begin{align}
    \tau_R&:=\frac{\pi}{4H_{\max}},
    &U_{v|a}&:=e^{-i\tau_R\qty(K_{v|a}-\Delta_{\min}I/2)}.
    \label{eq:alg_reflection_signal}
\end{align}
Writing its eigenvalues as $e^{-i\varphi}$, all phases in the
exact-reference sector lie in $(-\pi/2,\pi/2)$. Ground and excited
phases lie on opposite sides of zero with separation
$\Theta(\Delta_{\min}/H_{\max})$. A bounded polynomial approximation
to $-\operatorname{sgn}(\sin\varphi)$, implemented by
GQSP, therefore realizes
\begin{align}
    \mathcal R_v:=2P(v)-I.
    \label{eq:alg_reflection}
\end{align}
The construction uses the sign approximation of
Ref.~\cite[Lemma~14]{gilyen_qsvt_2019} and the Laurent-polynomial
implementation of
Ref.~\cite[Corollary~5 and Theorem~6]{motlagh_generalized_qsp_2024}.
For state error $\epsilon$, including the signal ancilla, it requires oracle time
\begin{align}
    T_{\mathcal R}(\epsilon)
    =O\qty(\frac{1}{\Delta_{\min}}\log\frac{1}{\epsilon})
    \label{eq:alg_reflection_cost}
\end{align}
and constant-size GQSP ancilla overhead. The spectral and error guarantees
are proved in \cref{app:alg_reflection}.

Controls act on the complete
specified circuit, retaining its phase, and inverses reverse the
circuit. Both preserve oracle time and the relevant state-error
guarantees, as established in \cref{app:alg_controls}.
These reflections can now be combined to transport ground states
between vertices of a cell.

\subsubsection{Ground-state reflections to direct rotation}
\label{sec:alg_rotation}

Direct rotations transport ground states between vertices of a cell with a fixed phase convention.
For two vertices $v,w$ in the same cell, $P(v)$ and $P(w)$ project
onto one-dimensional complex ground-state subspaces, called rays.
Each ray represents a normalized ground state up to an overall phase.
The product of the corresponding reflections rotates the plane
spanned by these rays by twice the angle between them. Its principal
square root is the direct rotation~\cite{bravyi2011schriefferwolff}
\begin{align}
    D_{wv}&:=\qty(\mathcal R_w\mathcal R_v)^{1/2},
    \label{eq:alg_direct_rotation}\\
    D_{wv}\ket{u(v)}
    &=\ket{u(w)}^{(v)}
    :=\frac{P(w)\ket{u(v)}}{\norm{P(w)\ket{u(v)}}}.
    \label{eq:alg_rotation_action}
\end{align}
The state $\ket{u(w)}^{(v)}$ is the normalized ground state at $w$
whose overlap with $\ket{u(v)}$ is real and strictly positive.
This condition uniquely fixes the phase of $\ket{u(w)}^{(v)}$ relative to
$\ket{u(v)}$~\cite{pancharatnam_generalized_1956}.
Here and throughout, $U^\alpha$ denotes the principal power of a unitary $U$,
defined spectrally by $e^{i\theta}\mapsto e^{i\alpha\theta}$ with
$\theta\in(-\pi,\pi)$, unless another branch is explicitly specified.

The small angle between the ground rays also makes the direct rotation
efficient to synthesize.
Since $\norm{P(v)-P(w)}\leq\rho$, the reflection
product has its spectrum in a fixed arc around $1$, separated from the
principal branch cut at $-1$. Square-root GQSP implements $D_{wv}$ using
$O\qty(\log(1/\epsilon))$ queries to the reflection product and its
inverse~\cite[Theorem~10]{motlagh_generalized_qsp_2024}.
The theorem's dependence on the minimum singular value of the phase
generator contributes only a constant here, as shown in \cref{app:alg_qsp}.
A fixed reference copy at any vertex of the same cell supplies both reflections.
\Cref{app:alg_qsp} proves the spectral statement and
\cref{app:alg_nested_costs} gives the internal precision allocation.

\subsubsection{Direct rotations to fractional phase modulation}
\label{sec:alg_fractional}

A closed sequence of direct rotations converts ground-state transport
into a cell phase.
Hold one copy at the anchor $a$ while applying four direct rotations
to the other:
\begin{align}
    L_{\mathcal C}&:=D_{ad}D_{dc}D_{cb}D_{ba},
    \label{eq:alg_cell_loop}\\
    L_{\mathcal C}\ket{u(a)}
    &=e^{-i\phi_{\mathcal C}}\ket{u(a)}.
    \label{eq:alg_loop_action}
\end{align}
The loop phase is determined by the ground-state overlaps along its boundary.
For an oriented edge $v\to w$, define
$\varphi_{vw}:=\operatorname{Arg}\braket{u(v)}{u(w)}$.
By \cref{eq:alg_rotation_action},
\begin{align}
    D_{wv}\ket{u(v)}
    &=\frac{\braket{u(w)}{u(v)}}{\abs{\braket{u(v)}{u(w)}}}\ket{u(w)}
    \notag\\
    &=e^{-i\varphi_{vw}}\ket{u(w)}.
    \label{eq:alg_edge_overlap_phase}
\end{align}
Multiplying the four phase factors gives
\begin{align}
    e^{-i\qty(\varphi_{ab}+\varphi_{bc}+\varphi_{cd}+\varphi_{da})}
    &=\frac{B_{\mathcal C}^*}{\abs{B_{\mathcal C}}}
    =e^{-i\phi_{\mathcal C}}.
    \label{eq:alg_cell_phase_product}
\end{align}
Every reflection query
remains within the cell of the fixed reference.

To take a fractional power of this loop, we need a branch gap for
the entire loop operator.
Moreover, each direct rotation is close to the identity, giving
\begin{align}
    \norm{L_{\mathcal C}-I}
    \leq8\sin\qty(\frac{1}{2}\arcsin\rho)\leq\frac14.
    \label{eq:alg_loop_branch_gap}
\end{align}
The entire spectrum of the closed product therefore stays a constant
distance from $-1$. Fractional-power GQSP can be applied to this product,
uniformly for known $\alpha\in(0,1)$, to construct
\begin{align}
    W_{\mathcal C}(\alpha)&:=\qty(L_{\mathcal C})^\alpha,
    \label{eq:alg_fractional}\\
    W_{\mathcal C}(\alpha)\ket{u(a)}
    &=e^{-i\alpha\phi_{\mathcal C}}\ket{u(a)}.
    \label{eq:alg_fractional_action}
\end{align}

It uses $O\qty(\log(1/\epsilon))$ loop queries and their
inverses~\cite[Theorem~10]{motlagh_generalized_qsp_2024}.
The theorem's dependence on the minimum singular value of the phase
generator also contributes only a constant for the loop signal, as shown
in \cref{app:alg_qsp}.
In particular, the exponent $\alpha$ is synthesized directly, 
so it introduces no multiplicative oracle-time factor.
The spectral bound and the fractional-power implementation are detailed
in \cref{app:alg_qsp}.

\Cref{tab:alg_primitive_costs} summarizes the costs derived in
\cref{app:alg_nested_costs}, with precision allocated to all nested
calls. The GQSP ancilla overhead remains constant
size because the number of GQSP construction layers is fixed.
\begin{table}[bt]
    \caption{Hamiltonian-oracle time for each completed local primitive
    at state error $\epsilon$, including its internal error allocation.
    }
    \label{tab:alg_primitive_costs}
    \centering
    \begin{tabular}{c c}
        \hline\hline
        Primitive & Oracle time \\
        \hline
        $\mathcal R$ & $O\qty(\Delta_{\min}^{-1}\log(1/\epsilon))$ \\
        $D$ & $O\qty(\Delta_{\min}^{-1}\log^2(1/\epsilon))$ \\
        $W$ & $O\qty(\Delta_{\min}^{-1}\log^3(1/\epsilon))$ \\
        \hline\hline
    \end{tabular}
\end{table}

\subsection{Overall algorithm and resource bounds}
\label{sec:alg_overall}

We now combine the local primitives into a complete algorithm.
A closed scan accumulates a controlled fractional phase over all cells;
several such scans then encode the signed Chern number for Fourier readout.
We then bound the implementation error and total resources.

\subsubsection{Closed scans with two ground-state copies}
\label{sec:alg_scan}

A scan visits every cell anchor in cyclic raster order,
\begin{align}
    a_\ell&:=(j_x\,\delta k_x,j_y\,\delta k_y),
    &\ell&:=j_yN_x+j_x,
    \label{eq:alg_raster}
\end{align}
with successor $v:=a_{(\ell+1)\bmod N_{\square}}$ for the current anchor
$a:=a_\ell$. A move within a row is horizontal. At a row transition,
periodicity identifies the successor with the upper-right vertex of
the last cell in the row.
The final move returns to $\mathbf{k}_0$. Thus every move connects
two vertices of a single cell.
This traversal is valid also for $C\ne0$: it only selects states along
the scan path and requires no smooth periodic gauge on the whole
torus.

Each step consists of phase accumulation at the current anchor
followed by transport to the next anchor.
Initialize two copies at $\mathbf{k}_0$ and fix a readout qubit.
At each anchor, apply $W_{\mathcal C}(m/q)~(m,q \in \mathbb{Z})$ to copy 2 controlled on the
readout being $\ket{1}$, using copy 1 as the fixed reference. Both
copies return to the anchor ground ray, and the cell phase is kicked
back onto the readout qubit.
Then move both copies to the next anchor
$v$ without conditioning on the readout qubit: apply $D_{va}$ to copy
2 using copy 1 as reference, followed by $D_{va}$ on copy 1 using the
moved copy 2 as reference. On the ideal trajectory each reference is
an exact local ground state.

The controlled cell phases multiply to
\begin{align}
    \prod_{\mathcal C}e^{-im\phi_{\mathcal C}/q}
    =e^{2\pi imC/q},
    \label{eq:alg_scan_phase}
\end{align}
by \cref{eq:alg_phase_sum}. The transports are independent of the
readout value, so their accumulated phase is common to both branches.
Consequently, a closed scan $\mathcal S_{q,m}$ restores both
reference-point ground rays and implements the logical gate
\begin{align}
    G_{q,m}:=\operatorname{diag}\qty(1,e^{2\pi imC/q}).
    \label{eq:alg_scan_gate}
\end{align}

Each scan uses $N_{\square}$ fractional phase modulations and
$2N_{\square}$ navigation rotations. The two prepared copies are
reused across all scans.

\subsubsection{Integer reconstruction and algorithm}
\label{sec:alg_readout}

An unfractionated scan would produce $e^{2\pi iC}=1$.
Fractional phases retain the integer residue. Choose
\begin{align}
    r_C&:=\left\lceil\log_2\qty(2C_{\max}+1)\right\rceil,
    &q&:=2^{r_C}>2C_{\max}.
    \label{eq:alg_modulus}
\end{align}
The scans encode the integer residue as a Fourier state of the readout register.
Prepare $r_C$ readout qubits in $\ket{+}^{\otimes r_C}$.
For each $j=0,\ldots,r_C-1$, perform one scan controlled by bit $j$
with $m:=2^j$. The fractional exponents are
$\alpha:=2^j/q\in(0,1/2]$. Since these exponents are implemented
directly, this uses exactly $r_C$ scans. Up to a common transport
phase, the readout state is
\begin{align}
    \frac{1}{\sqrt q}\sum_{n=0}^{q-1}e^{2\pi inC/q}\ket n.
    \label{eq:alg_fourier_state}
\end{align}
The inverse QFT returns $R=C\bmod q$ with certainty in the ideal
circuit. Centered decoding,
\begin{align}
    \operatorname{cent}_q(R):=
    \begin{cases}
        R,&0\leq R<\frac{q}{2},\\
        R-q,&\frac{q}{2}\leq R<q,
    \end{cases}
    \label{eq:alg_decode}
\end{align}
recovers the signed integer uniquely because $\abs{C}<q/2$.
The scan and Fourier identities are proved in
\cref{app:alg_scan_readout}.
\Cref{alg:chern_estimation} summarizes the construction.

\begin{figure}[bt]
\normalsize
\raggedright
\refstepcounter{algorithm}\label{alg:chern_estimation}
\par\medskip\noindent
\textbf{\Cref{alg:chern_estimation}. Chern-number estimation}
\par\smallskip
\begin{algorithmic}[1]
    \Require Controlled Hamiltonian evolution; $L_x,L_y,\Delta_{\min},H_{\max}$;
    coherent ground state preparation at $\mathbf{k}_0$; target failure probability $\eta$.
    \Ensure The ground-state Chern number $C$ with probability at least $1-\eta$.
    \If{$\min\qty{L_x,L_y}<\Delta_{\min}/2$}
        \State \Return $0$.
    \EndIf
    \State Construct the mesh using \cref{eq:alg_mesh,eq:alg_cell_count} and compute
    $C_{\max},r_C,q$ from \cref{eq:alg_chern_bound,eq:alg_modulus}.
    \State Set $\epsilon_{\mathrm{loc}}:=\eta/(3N_{\square}r_C)$ as shown in \cref{eq:alg_local_precision}.
    \State Prepare two copies of $\ket{u(\mathbf{k}_0)}$ and the
    readout state $\ket{+}^{\otimes r_C}$.
    \For{$j=0,\ldots,r_C-1$}
        \For{each cell $\mathcal C$ in cyclic raster order}
            \State Apply $W_{\mathcal C}(2^j/q)$ to copy 2 controlled
            by readout bit $j$, using copy 1 as reference.
            \State Move copy 2 to the next anchor using copy 1
            as reference; then move copy 1 using copy 2.
        \EndFor
    \EndFor
    \State Apply the inverse QFT and measure the readout value $R$.
    \State \Return $\operatorname{cent}_q(R)$ from \cref{eq:alg_decode}.
\end{algorithmic}
\end{figure}

\subsubsection{Correctness and resource bounds}
\label{sec:alg_cost}

We choose the local accuracies to achieve the target success probability
for the implemented circuit.
There are at most $3N_{\square}r_C$ complete local primitives.
Assign each one state error
\begin{align}
    \epsilon_{\mathrm{loc}}:=\frac{\eta}{3N_{\square}r_C}
    \label{eq:alg_local_precision}
\end{align}
on its ideal incoming state, including all work registers. Unitarity
bounds the final state error by the sum of these local errors.
This also propagates errors in the moved reference copies: the local
spectral promises are required only on the ideal trajectory. Since the
ideal inverse-QFT output is exact, the accumulated state error bounds
the probability of returning a wrong integer. The detailed argument is
given in \cref{app:alg_error}.

With this precision choice, the total oracle time follows by counting
completed local operations.
In the quantum case, each of the $r_C$ scans contains
$N_{\square}$ fractional phase modulations and $2N_{\square}$
direct rotations. Let $T_W(\epsilon)$ and $T_D(\epsilon)$ denote
their respective Hamiltonian-oracle times at state error $\epsilon$,
as summarized in \cref{tab:alg_primitive_costs}. With the local
precision in \cref{eq:alg_local_precision}, the total runtime is
\begin{align}
    T_{\mathrm{upper}}(\eta)
    &=O\qty[
        N_{\square}r_C\qty(
            T_W(\epsilon_{\mathrm{loc}})
            +2T_D(\epsilon_{\mathrm{loc}})
        )
    ]
    \notag\\
    &=O\qty[
        \frac{N_{\square}r_C}{\Delta_{\min}}
        \log^3\qty(\frac{N_{\square}r_C}{\eta})
    ]
    \notag\\
    &=\widetilde O\qty(\frac{L_xL_y}{\Delta_{\min}^3}).
    \label{eq:alg_resource_composition}
\end{align}
The last equality uses \cref{eq:alg_cell_count,eq:alg_chern_bound}
and the scan count $r_C=O\qty(\log\qty(2C_{\max}+1))$ from
\cref{eq:alg_modulus}, establishing \cref{eq:alg_runtime} in
\cref{thm:nonadiabatic_algorithm}.
The reflection degree contains
$H_{\max}/\Delta_{\min}$, but each query of the Hamiltonian evolution has duration
$\Theta(H_{\max}^{-1})$; the norm scale cancels in oracle time.

For an $n$-qubit Hamiltonian, the two data registers, mesh address,
readout register, and constant-size GQSP ancilla overhead require the following number of qubits:
\begin{align}
    N_{\mathrm{qubit}}
    &=2n+O\qty[\log N_{\square}+\log\qty(2C_{\max}+1)]+O(1)
    \notag\\
    &=2n+O\qty[\log\qty(1+\frac{L_xL_y}{\Delta_{\min}^2})].
    \label{eq:alg_qubits}
\end{align}
These registers are reused throughout.
The qubit count in \cref{eq:alg_qubits} excludes any additional
workspace needed for state preparation or gate implementation.
Gate-synthesis and classical-processing costs are not included
in the oracle time.
\Cref{app:alg_resources} derives the oracle time and qubit count.

On the input classes covered by
\cref{thm:chern_lower_bound}, the evolution-time scaling matches the
lower bound up to polylogarithmic factors.
For the input classes covered by
\cref{cor:binary_chern_lower_bound}, the same matching scaling holds
under the binary promise $C\in\qty{0,1}$.

\section{Query Lower Bound for Chern-Number Estimation}
\label{sec:chern_lower_bound}

We now prove \cref{thm:chern_lower_bound,cor:binary_chern_lower_bound}.
The lower bound follows the same division of problem parameter space as the
algorithm in \cref{sec:upper_bound_algorithm}.
Write
\begin{align}
    a_\mu:=\frac{L_\mu}{\Delta_{\min}},
    \qquad \mu\in\qty{x,y}.
    \label{eq:lower_bound_dimensionless_parameters}
\end{align}
If $\min\qty{a_x,a_y}<1/2$, the zero-Chern certificate in
\cref{app:alg_geometry} gives $C=0$ for every allowed input.
Returning zero requires no oracle calls, so the optimal oracle time,
and hence its lower bound, is
\begin{align}
    T_{\mathrm{lower}}=0.
    \label{eq:lower_bound_zero_branch}
\end{align}
This is exactly the early-return case of \cref{eq:alg_zero_check}.
In the complementary regime $a_x,a_y\geq1/2$, including its boundary,
we first partition the parameter torus into regions and assign
an independently chosen bit to each region.
A bit value of zero selects an inactive configuration with zero
Chern-number contribution.
A bit value of one selects a smooth Hamiltonian configuration
within the region that contributes $+1$ to the total Chern number;
we call this configuration a positive defect.
The total Chern number is therefore the number of positive defects.
We then adapt this construction in
\cref{sec:lower_bound_binary_family} by placing
fixed negative
defects, each contributing $-1$, in selected regions and restricting the
variable bits so that the total Chern number is zero or one.
A common adversary argument gives a lower bound that matches the
algorithm's evolution-time scaling up to polylogarithmic factors.

\subsection{Hard-instance construction}
\label{sec:lower_bound_hard_instance}

For a partition into $M$ regions, the three basic constructions below encode
the associated bits $z:=(z_1,\ldots,z_M)\in\qty{0,1}^M$ in a two-band
Hamiltonian
\begin{align}
    H_z(\mathbf{k})
    :=\frac{\Delta_{\min}}2\mathbf n_z(\mathbf{k})\cdot\bm\sigma,
    \qquad \abs{\mathbf n_z(\mathbf{k})}=1,
    \label{eq:lower_bound_hard_hamiltonian}
\end{align}
where $\bm\sigma:=(\sigma_x,\sigma_y,\sigma_z)$ is the vector of Pauli
matrices.
We refer to the unit-vector field $\mathbf n_z(\mathbf{k})$ as the texture.
All inputs therefore have eigenvalues $\pm\Delta_{\min}/2$, gap
$\Delta_{\min}$, and norm $\Delta_{\min}/2$.
In each of the three constructions below, we choose $\mathbf n_z$
so that the supplied derivative bounds hold and
$\mathbf n_z(0,0)=(0,0,1)$ for every bit string $z$.
This ensures that $H_z(0,0)=(\Delta_{\min}/2)\sigma_z$,
with the common reference ground state $\ket1$.
In these basic families, flipping a bit inserts or removes a positive
defect and changes the Hamiltonian only in the corresponding
region.
The details of these properties are proved in
\cref{app:chern_lower_bound_proof}.

\subsubsection{Tiles in the interior}
\label{sec:lower_bound_tiles}

When $a_x,a_y>1/2$, both parameter circles can be partitioned into
intervals with enough derivative budget for smooth interpolation.
Choose
\begin{align}
    n_\mu&:=\max\qty{1,\lfloor a_\mu\rfloor},
    &w_\mu&:=\frac{2\pi}{n_\mu}.
    \label{eq:lower_bound_tile_counts}
\end{align}
Since $n_\mu<2a_\mu$, there exists $0<\varepsilon<1/2$ such that
\begin{align}
    \frac{n_\mu}{2(1-\varepsilon)}\leq a_\mu,
    \qquad \mu\in\qty{x,y}.
    \label{eq:lower_bound_smoothing_budget}
\end{align}
Let $B_\varepsilon:[0,1]\to[0,1]$ be the smooth interpolation
constructed in \cref{app:lower_bound_interpolation}, with
\begin{align}
    B_\varepsilon(0)&=0, \quad B_\varepsilon(1)=1,
    \notag\\
    0\leq B_\varepsilon'(u)&\leq\frac1{1-\varepsilon},
    \notag\\
    B_\varepsilon^{(m)}(0)&=B_\varepsilon^{(m)}(1)=0
    \quad(m\geq1).
    \label{eq:lower_bound_interpolation_properties}
\end{align}
The partition counts $n_\mu$ describe the hard instance and are
distinct from the oracle-grid counts $N_\mu$.

On the tile with indices $0\leq i<n_x$ and $0\leq j<n_y$, use
coordinates $s:=(k_x-iw_x)/w_x$ and $t:=(k_y-jw_y)/w_y$ in $[0,1]$.
Define a polar excursion from north to south and back to north by
\begin{align}
    \vartheta(s):=
    \begin{cases}
        \pi B_\varepsilon(2s),&0\leq s\leq\frac12,\\
        \pi\qty[1-B_\varepsilon(2s-1)],&\frac12<s\leq1.
    \end{cases}
    \label{eq:lower_bound_polar_excursion}
\end{align}
For the bit $z_{ij}$, set
\begin{align}
    \varphi_{z_{ij}}(s,t)&:=
    \begin{cases}
        2\pi z_{ij}B_\varepsilon(t),&0\leq s\leq\frac12,\\
        0,&\frac12<s\leq1,
    \end{cases}
    \label{eq:lower_bound_tile_azimuth}\\
    \mathbf m(\theta,\phi)
    &:=\qty(\sin\theta\cos\phi,\sin\theta\sin\phi,\cos\theta),
    \label{eq:lower_bound_spherical_texture}\\
    \mathbf n_z(\mathbf k)
    &:=\mathbf m\qty(\vartheta(s),\varphi_{z_{ij}}(s,t)).
    \label{eq:lower_bound_tile_texture}
\end{align}
An active tile ($z_{ij}=1$) contains a positive defect: its texture
winds once in azimuth during the first
half of the tile. Every tile returns along the same meridian
during its second half.
Flatness of the interpolation at its endpoints makes the texture
smooth at the poles and across tile boundaries, including boundaries
between different bit values.
As proved in \cref{app:lower_bound_tiles}, the construction satisfies
\begin{align}
    \norm{\partial_{k_\mu}H_z}
    &\leq\frac{\Delta_{\min}n_\mu}{2(1-\varepsilon)}
    \leq L_\mu,
    \label{eq:lower_bound_tile_derivatives}\\
    M=n_xn_y&\geq\frac{a_xa_y}{4},\\
    C(H_z)&=\sum_{i,j}z_{ij}.
    \label{eq:lower_bound_tile_chern_count}
\end{align}

\subsubsection{Slabs on the boundary}
\label{sec:lower_bound_slabs}

When $a_y=1/2$ and $a_x>1/2$, retain the $n_x$ polar excursions
along $x$ and use the entire $y$ circle for the azimuthal winding.
In slab $i$, replace \cref{eq:lower_bound_tile_azimuth} by
\begin{align}
    \varphi_{z_i}(s,k_y):=
    \begin{cases}
        z_i k_y,&0\leq s\leq\frac12,\\
        0,&\frac12<s\leq1.
    \end{cases}
    \label{eq:lower_bound_slab_azimuth}
\end{align}
An active slab ($z_i=1$) contains a positive defect, with one azimuthal
winding around the $y$ circle on its first half; an inactive slab
($z_i=0$) has no azimuthal winding.

Choose $\varepsilon$ to satisfy \cref{eq:lower_bound_smoothing_budget}
only for $\mu=x$.
The winding is smooth and periodic in $k_y$ without interpolation
in that direction, giving
\begin{align}
    \norm{\partial_{k_x}H_z}
    &\leq\frac{\Delta_{\min}n_x}{2(1-\varepsilon)}\leq L_x,
    \notag\\
    \norm{\partial_{k_y}H_z}&\leq\frac{\Delta_{\min}}2=L_y,
    \label{eq:lower_bound_slab_derivatives}\\
    M=n_x&\geq\frac{a_x}{2}=a_xa_y,\\
    C(H_z)&=\sum_i z_i.
    \label{eq:lower_bound_slab_chern_count}
\end{align}
For $a_x=1/2$ and $a_y>1/2$, interchange the coordinates and
reverse the azimuthal winding to preserve the positive Chern
contribution of each active slab.
The resulting family has $M=n_y\geq a_xa_y$ and satisfies the
corresponding derivative bounds.
The derivative bounds and Chern-number formulas for
both boundary cases
are proved in \cref{app:lower_bound_slabs}.

\subsubsection{Flattened QWZ construction at the corner}
\label{sec:lower_bound_corner}

At $a_x=a_y=1/2$, compare a constant Hamiltonian with a
spectrally flattened orientation-reversed form of the Qi--Wu--Zhang (QWZ)
model~\cite{qi_topological_2006}:
\begin{align}
    H_0(\mathbf k)&:=\frac{\Delta_{\min}}2\sigma_z,
    \label{eq:lower_bound_corner_hamiltonian_0}\\
    H_1(\mathbf k)&:=\frac{\Delta_{\min}}2
        \frac{\mathbf d(\mathbf k)}{\abs{\mathbf d(\mathbf k)}}
        \cdot\bm\sigma,
    \label{eq:lower_bound_corner_hamiltonian_1}\\
    \mathbf d(\mathbf k)
    &:=\qty(\sin k_x,-\sin k_y,1+\cos k_x+\cos k_y).
    \label{eq:lower_bound_corner_bloch_vector}
\end{align}
Thus \(\mathbf d(\mathbf k)\) is precisely the Bloch vector of
\(H_{\rm QWZ}(\mathbf k;1)\) from \cref{eq:qwz_model_hardness}.
The normalization is smooth, and both Hamiltonians satisfy
$\norm{\partial_{k_\mu}H_z}\leq\Delta_{\min}/2=L_\mu$.
They share the reference ground state $\ket1$ at $(0,0)$ and have
$C(H_0)=0$ and $C(H_1)=1$ with the convention of
\cref{eq:berry_connection_curvature_def}.
The derivative and flattening checks are given in
\cref{app:lower_bound_corner}.
Taking the whole torus
as a single region gives $M=1$, with $H_0$
representing the inactive configuration and $H_1$ a positive defect.

\subsubsection{Common properties and multiband embedding}
\label{sec:lower_bound_common_properties}

All three basic families have the same properties needed for the lower bound.
After enumerating their bit regions by $j=1,\ldots,M$, they satisfy
\begin{align}
    M&\geq\frac{a_xa_y}{4},
    \label{eq:lower_bound_number_of_defects}\\
    C(H_z)&=\sum_{j=1}^M z_j.
    \label{eq:lower_bound_chern_hamming_weight}
\end{align}
If $z^{(j)}$ denotes $z$ with its $j$th bit flipped, then the
neighboring Chern numbers differ by one.
Their Hamiltonians differ only in the assigned region and obey
\begin{align}
    \norm{H_z(\mathbf k)-H_{z^{(j)}}(\mathbf k)}
    \leq\Delta_{\min}.
    \label{eq:lower_bound_neighbor_norm}
\end{align}
The common reference state can be supplied by the same preparation
circuit for every input.

The hard families also embed in any fixed number $N\geq2$ of bands:
\begin{align}
    H_z^{(N)}(\mathbf k)
    :=H_z(\mathbf k)\oplus\frac{\Delta_{\min}}2 I_{N-2}.
    \label{eq:lower_bound_multiband_embedding}
\end{align}
The added unoccupied bands preserve the nondegenerate ground state,
gap, norm, derivative bounds, and Chern number.
They are independent of both $z$ and $\mathbf k$, so neighboring
oracle differences acquire only a zero block.

\subsubsection{\texorpdfstring{Binary-Chern hard family}{Binary-Chern hard family}}
\label{sec:lower_bound_binary_family}

To establish the lower bound under the promise $C\in\qty{0,1}$,
we reserve
some tile or slab regions for fixed negative
defects.
Their contributions shift the total Chern numbers associated with
two adjacent central Hamming layers to zero and one.

For tiles and slabs, reversing the azimuthal winding of a positive
defect produces a negative defect.
On the first half of a tile, use
\begin{align}
    \varphi(s,t):=-2\pi B_\varepsilon(t),
    \label{eq:lower_bound_signed_tile_azimuth}
\end{align}
and retain $\varphi:=0$ on the return half.
For a slab, reverse the azimuthal winding on its first half:
use $-k_y$ when $a_y=1/2<a_x$, and $+k_x$ when $a_x=1/2<a_y$,
while retaining $\varphi:=0$ on the return half.
In both constructions, the Berry curvature changes sign, so the
region contributes $-1$.
The derivative bounds, smoothness, gap, norm, and common reference
ground state are preserved, as verified in
\cref{app:lower_bound_tiles,app:lower_bound_slabs}.

To assemble the binary-Chern family,
let $M$ be the available number of tile or slab regions and set
\begin{align}
    m:=\lfloor(M-1)/3\rfloor,
    \qquad M_{\mathrm{var}}:=2m+1.
    \label{eq:lower_bound_binary_region_counts}
\end{align}
Fix negative defects in $m$ regions and choose another $M_{\mathrm{var}}$ regions as
variable regions; this uses $m+M_{\mathrm{var}}=3m+1\leq M$ regions.
In variable region $j$, a bit $z_j$ specifies the absence or presence
of a positive defect, and all remaining regions are inactive.
Restrict $z\in\qty{0,1}^{M_{\mathrm{var}}}$ to $|z|=m$ or $|z|=m+1$, where
$|z|:=\sum_{j=1}^{M_{\mathrm{var}}} z_j$.
The resulting family, also denoted by $H_z$, satisfies
\begin{align}
    C(H_z)=-m+|z|\in\qty{0,1}.
    \label{eq:lower_bound_binary_chern_values}
\end{align}
All supplied bounds and the reference-state preparation circuit remain
fixed across this family, and adjacent allowed inputs differ in just one
variable region with the norm bound of \cref{eq:lower_bound_neighbor_norm}.
At $a_x=a_y=1/2$, use the existing flattened-QWZ pair with Chern
numbers zero and one directly.

\subsection{Common adversary estimate}
\label{sec:lower_bound_oracle_argument}

For each hard input $H_z$, we apply the common Hamiltonian-oracle
model of \cref{eq:controlled_hamiltonian_oracle,eq:hamiltonian_oracle_evolution}
with $H=H_z$, using the common
mesh specified in \cref{eq:alg_mesh}.
Write $\ket{\psi_z(t)}$ for the corresponding algorithm state.
The oracle-time cost is defined in
\cref{eq:oracle_time_definition}.
All supplied bounds and the reference-state preparation circuit
are fixed across the hard inputs, so the driver and initial state
carry no input-dependent information.

Successful estimation requires the algorithm to distinguish every
neighboring pair of bit strings.
We measure this distinguishability with the squared-distance progress
function~\cite{farhi_quantum_2008,yonge-mallo_adversary_2011}
\begin{align}
    W(t):=\frac1{2^M}\sum_{z\in\qty{0,1}^M}\sum_{j=1}^M
    \norm{\ket{\psi_z(t)}-\ket{\psi_{z^{(j)}}(t)}}^2.
    \label{eq:chern_progress_function}
\end{align}
All inputs start in the same state, including any prepared ground-state
copies, so $W(0)=0$.
Since neighboring Chern numbers differ by one, their successful output
intervals of radius $1/2$ are disjoint.
The final states must therefore have overlap at most
$2\sqrt{\eta(1-\eta)}$, where $\eta\in(0,1/2)$ bounds the estimation
failure probability on every input. This implies
\begin{align}
    W(T)\geq2\qty[1-2\sqrt{\eta(1-\eta)}]M.
    \label{eq:lower_bound_required_progress}
\end{align}
The lower bound on $W(T)$
is derived in
\cref{eq:app_lower_bound_required_progress_derivation}.

The disjoint bit regions limit how rapidly this progress can accumulate.
Their address projectors are mutually orthogonal, and each neighboring
oracle difference is supported on one such region with norm at most
$\Delta_{\min}$ by \cref{eq:lower_bound_neighbor_norm}.
Tensoring with $\ket{1}\!\bra{1}_{\mathrm c}$ preserves the region
support and the norm bound.
Differentiating \cref{eq:chern_progress_function} cancels the common
driver, and summing over the orthogonal regions gives
\begin{align}
    \abs{\frac{dW}{dt}}\leq2\Delta_{\min}\abs{g(t)}.
    \label{eq:key_progress_rate_bound}
\end{align}
For the corner pair, the same argument uses $M=1$ and the identity
as the region projector.
The full calculation appears in \cref{app:chern_lower_bound_adversary}.

Comparing the required progress with the integrated rate proves the
uniform oracle-time lower bound.
Since $W(T)\leq2\Delta_{\min}T_{\mathrm{oracle}}$, we obtain
\begin{align}
    T_{\mathrm{oracle}}
    &\geq\frac{1-2\sqrt{\eta(1-\eta)}}{\Delta_{\min}}M
    \notag\\
    &\geq c_\eta\frac{a_xa_y}{\Delta_{\min}}
    =c_\eta\frac{L_xL_y}{\Delta_{\min}^3}.
    \label{eq:chern_hamiltonian_oracle_scaling}
\end{align}
The second line uses \cref{eq:lower_bound_number_of_defects},
establishing \cref{thm:chern_lower_bound} throughout
$a_x,a_y\geq1/2$.
Together with \cref{eq:lower_bound_zero_branch}, this covers the
same two regimes as the algorithm.

For \cref{cor:binary_chern_lower_bound}, apply the adversary argument to
the relation between the two central Hamming layers defined in
\cref{app:lower_bound_binary_promise}.
This relation has degree $m+1$ on both sides and yields
\begin{align}
    T_{\mathrm{oracle}}
    \geq\frac{1-2\sqrt{\eta(1-\eta)}}{\Delta_{\min}}(m+1).
    \label{eq:lower_bound_binary_layer_runtime}
\end{align}
Since $m+1\geq M/3$ and $M\geq a_xa_y/4$, this proves
\cref{eq:binary_chern_lower_bound} for tiles and slabs.
The corner pair gives $\Omega(\Delta_{\min}^{-1})$, which equals
$\Omega(L_xL_y/\Delta_{\min}^3)$ at $a_x=a_y=1/2$.
The progress function, its rate bound, and the corner constant
are derived in
\cref{app:lower_bound_binary_promise}.

\section{Computational Complexity of Chern-Number Estimation}
\label{sec:bqp_hardness}

We now prove the succinct-input complexity classification stated in
\cref{thm:guided_cne_complexity}.

\subsection{Containment in \texorpdfstring{$\FBQP$}{FBQP}}
\label{sec:bqp_containment}
We apply the quantum algorithm of \cref{sec:upper_bound_algorithm}.
Let $L_x,L_y,H_{\max}\leq\poly(n)$ be the promised derivative and norm bounds,
let $\Delta_{\min}\geq1/\poly(n)$ be the promised gap lower bound, and let
$\gamma\geq1/\poly(n)$ bound the guiding-state overlap from below.
We account for both reference-state preparation and the gate cost of the algorithm.

First prepare two approximate copies of the ground state at the single reference point
$\mathbf{k}_0$.
Phase estimation of $H(\mathbf{k}_0)$ on fresh guiding states, with energy resolution less than
$\Delta_{\min}/8$, samples the ground energy with probability at least $\gamma$, up to the
phase-estimation error.
Taking the smallest outcome from $O(\gamma^{-1}\log(1/\epsilon))$ repetitions estimates
$E_0(\mathbf{k}_0)$ to this resolution with failure probability at most $\epsilon$, provided the
error probability of each phase-estimation call is chosen correspondingly smaller.
The resulting estimate defines an energy threshold inside the gap.
Applying an energy filter to fresh guiding states and accepting outcomes below this threshold
prepares the ground state to trace-distance error at most $\epsilon$.
Choosing the filter error polynomially smaller than $\gamma\epsilon^2$ and capping the number of
attempts at $O(\gamma^{-1}\log(1/\epsilon))$ makes both the preparation error and the probability
of exhausting the attempts $O(\epsilon)$.
All operations have polynomial gate cost for inverse-polynomial $\epsilon$.
Repeating this procedure supplies the two initial copies used by the algorithm.
No additional ground-state preparation is needed during its execution.

Next, implement the controlled forward and inverse Hamiltonian evolutions from the succinct local
description.
For fixed locality, collecting the local Pauli terms gives polynomially many terms with
polynomially bounded coefficients: each coefficient of the full Hamiltonian is bounded by its
operator norm.
Their parameter dependence can be evaluated reversibly to the required precision in polynomial
time.
Standard Hamiltonian simulation therefore realizes the required controlled evolutions with gate
cost polynomial in $n$, the evolution time, and the inverse simulation error.
For the membership claim, each GQSP functional transformation may equivalently
be implemented uniformly by coherent phase estimation of its signal unitary,
reversible evaluation of the required scalar function on the phase register,
phase kickback, and uncomputation.  The promised inverse-polynomial spectral
margins and inverse-polynomial target accuracy make the phase precision,
arithmetic, and gate synthesis polynomial in \(n\).  Hence a classical
polynomial-time algorithm generates the complete circuit description, not
merely a polynomial-length query sequence.
The mesh contains polynomially many cells, and the integer readout uses
polynomially many qubits and gates.  In particular, the
promised derivative and gap bounds make the range of possible Chern numbers
polynomially bounded by \cref{eq:alg_chern_bound}, so the classical output
has polynomial length.
In particular, the oracle-time bound becomes
\begin{align}
    T_{\mathrm{upper}}(\eta)
    =\widetilde O\!\left(\frac{L_xL_y}{\Delta_{\min}^3}\right)
    =\poly(n)
\end{align}
for fixed $\eta$, with all suppressed problem parameters polynomially bounded.

Finally, run the algorithm with ideal-input failure probability at most $1/12$.
Choose the total trace-distance error of the two prepared copies to be at most $1/12$, and the
accumulated implementation error of the simulation, GQSP, and readout gates to be at most $1/12$.
These accuracies require only polynomial resources.
Contractivity of trace distance and the circuit error bound then give total failure probability at
most $1/4$.
The circuit therefore outputs the exact integer \(C[P]\) with bounded error
in polynomial time.  This proves membership in \(\FBQP\).  The remainder of
the section proves \(\BQP\)-hardness.

\subsection{\texorpdfstring{$\BQP$}{BQP}-Hardness}
\label{sec:bqp_hardness_reduction}

\Cref{fig:hardness_overview} summarizes the topological part of the reduction.
The notation appearing in the figure is introduced below.

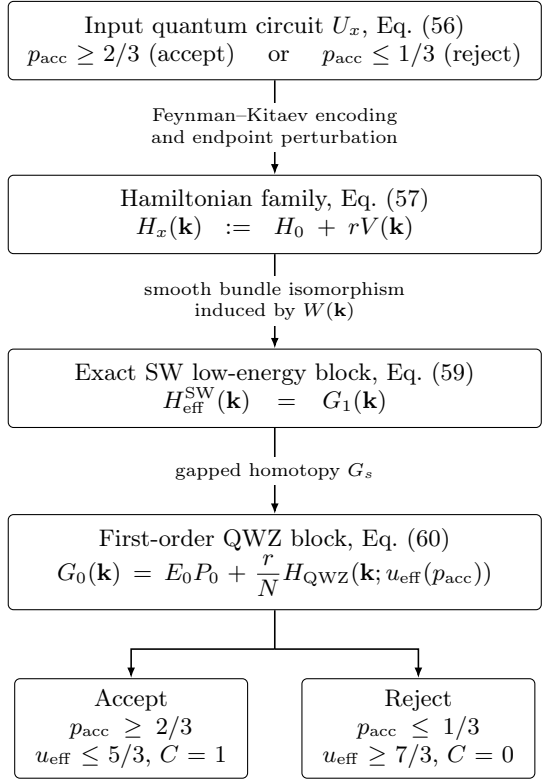
\begin{figure}[t]
  \centering
  \begin{tikzpicture}[
      x=1cm,y=1cm,
      box/.style={draw,rounded corners=2pt,align=center,
        inner sep=5pt,text width=6.7cm,font=\small},
      outcome/.style={draw,rounded corners=2pt,align=center,
        inner sep=4pt,text width=2.8cm,font=\small},
      arrow/.style={-{Latex[length=4pt]},line width=0.6pt},
      tag/.style={fill=white,inner sep=2pt,align=center,font=\scriptsize}
  ]
    \node[box] (circuit) at (0,0)
      {Input quantum circuit \(U_x\), Eq.~\eqref{eq:hardness_circuit_promise}\\
       $\pacc\ge2/3$ (accept) \quad or \quad
       $\pacc\le1/3$ (reject)};
    \node[box] (physical) at (0,-2.3)
      {Hamiltonian family, Eq.~\eqref{eq:hardness_physical_family}\\
       $H_x(\mathbf{k}):=H_0+rV(\mathbf{k})$};
    \node[box] (effective) at (0,-4.6)
      {Exact SW low-energy block, Eq.~\eqref{eq:hardness_exact_sw_block}\\
       \(H_{\rm eff}^{\rm SW}(\mathbf{k})=G_1(\mathbf{k})\)};
    \node[box] (qwz) at (0,-6.9)
      {First-order QWZ block, Eq.~\eqref{eq:hardness_gapped_homotopy}\\
       \(G_0(\mathbf{k})=E_0P_0+
       \dfrac rN H_{\rm QWZ}(\mathbf{k};\ueff(\pacc))\)};
    \node[outcome] (reject) at (1.9,-9.1)
      {Reject\\\(\pacc\le1/3\)\\
       \(\ueff\ge7/3\), \(C=0\)};
    \node[outcome] (accept) at (-1.9,-9.1)
      {Accept\\\(\pacc\ge2/3\)\\
       \(\ueff\le5/3\), \(C=1\)};

    \draw[arrow] (circuit) --
      node[tag]{Feynman--Kitaev encoding\\and endpoint perturbation}
      (physical);
    \draw[arrow] (physical) --
      node[tag]{smooth bundle isomorphism\\induced by \(W(\mathbf{k})\)}
      (effective);
    \draw[arrow] (effective) --
      node[tag]{gapped homotopy \(G_s\)}
      (qwz);
    \coordinate (fork) at (0,-8.05);
    \draw[line width=0.6pt] (qwz.south) -- (fork);
    \draw[arrow] (fork) -| (reject.north);
    \draw[arrow] (fork) -| (accept.north);
  \end{tikzpicture}
  \caption{Proof overview for the hardness reduction.  The circuit acceptance
  probability is encoded into the endpoint perturbation of a Feynman--Kitaev
  Hamiltonian.  The exact SW rotation
  identifies the physical ground bundle with the lower-band bundle of the
  effective Hamiltonian.  The gapped interpolation \(G_s\) then removes the
  higher-order SW remainder without changing the Chern number.  The resulting
  QWZ mass distinguishes rejecting and accepting circuits.}
  \label{fig:hardness_overview}
\end{figure}

\subsubsection{Preliminary: One-Qubit QWZ Model}
\label{sec:qwz_preliminary}

We first introduce a one-qubit QWZ model that will serve as the first-order
effective description of the low-energy sector of the Hamiltonian family
constructed below.

The basic local model behind the construction is a two-band Hamiltonian
\begin{align}
  H(\mathbf{k})&:=d_x(\mathbf{k})\sigma_x+d_y(\mathbf{k})\sigma_y
  +d_z(\mathbf{k})\sigma_z=\mathbf d(\mathbf{k})\cdot\sigma ,
\end{align}
where \(\mathbf{k}\in X\),
\(\mathbf d(\mathbf{k}):=(d_x(\mathbf{k}),d_y(\mathbf{k}),d_z(\mathbf{k}))
\in\mathbb R^3\) is the Bloch vector, and
\(\sigma:=(\sigma_x,\sigma_y,\sigma_z)\) denotes the Pauli matrices.  If
\(\mathbf d(\mathbf{k})\neq0\) for all \(\mathbf{k}\),
the two bands are separated by the gap
\(2\min_{\mathbf{k}}\|\mathbf d(\mathbf{k})\|\).  The normalized vector
\begin{align}
  \widehat{\mathbf d}(\mathbf{k})
  &:=\frac{\mathbf d(\mathbf{k})}{\|\mathbf d(\mathbf{k})\|}
\end{align}
defines a map \(\widehat{\mathbf d}: X\to S^2\).  The Chern number of a band is the
degree of this map, up to the sign convention determined by whether one takes
the upper or lower band and by the orientation of \( X\).  Thus the
topological content of a one-qubit Chern model is the wrapping of the Bloch
sphere.

We use an orientation-reversed form of the Qi--Wu--Zhang
model~\cite{qi_topological_2006}:
\begin{equation}
\label{eq:qwz_model_hardness}
  \begin{aligned}
  H_{\rm QWZ}(\mathbf{k};u)
  &:=\sin k_x\,\sigma_x-\sin k_y\,\sigma_y \\
  &\quad +(u+\cos k_x+\cos k_y)\sigma_z .
  \end{aligned}
\end{equation}
For large positive \(u\), the vector \(\mathbf d(\mathbf{k})\) stays near the north
pole of the Bloch sphere and the lower band is topologically trivial.  As
\(u\) is tuned, the gap closes at the phase boundaries where
\(\mathbf d(\mathbf{k})=0\), and the Chern number can change.  We use the
Chern-number convention of
\cref{eq:chern_number_def,eq:berry_curvature_projector_form}, with the
orientation \(dk_x\wedge dk_y\).  Let \(P_{\rm QWZ}(\mathbf{k};u)\) denote the
lower-band projector of \(H_{\rm QWZ}(\mathbf{k};u)\), and write
\(C_{\rm QWZ}(u):=C[P_{\rm QWZ}(\mathord\cdot;u)]\).
The four ``Dirac points'' are the momenta
\(k_x,k_y\in\{0,\pi\}\), where the sine components vanish.  The remaining
\(\sigma_z\) coefficient is the local mass, whose zero closes the gap.
Their contributions give
\[
  \begin{aligned}
  C_{\rm QWZ}(u)
  &=-\frac12\left[
  \operatorname{sgn}(u+2)-2\operatorname{sgn}(u)
  +\operatorname{sgn}(u-2)\right].
  \end{aligned}
\]
In particular,
\[
  u>2 \Rightarrow C_{\rm QWZ}(u)=0,\qquad
  0<u<2 \Rightarrow C_{\rm QWZ}(u)=1 .
\]
The sign and the uniform gap margin are checked in
\cref{app:qwz_calibration}.
Only this coarse dichotomy is needed below: a trivial region and a neighboring
topological region separated by a gap closing.

In the reduction, \(\tau_x,\tau_y,\tau_z\) act on a separate ancilla qubit
and play the role of the Pauli matrices above.  The circuit acceptance
probability sets \(\ueff(\pacc)\), placing rejecting and accepting circuits in
the trivial and topological regions, respectively.  The remaining argument
shows that this first-order identification survives in the full Hamiltonian
family.

\subsubsection{Construction}

Start with a \(\BQP\) circuit \(U_x\) with acceptance probability \(\pacc\), promised
to satisfy
\begin{equation}
\label{eq:hardness_circuit_promise}
  x\in L_{\rm yes}\Rightarrow \pacc\ge 2/3,\qquad
  x\in L_{\rm no}\Rightarrow \pacc\le 1/3 .
\end{equation}
Let \(\Hhist(U_x)\) be the Feynman--Kitaev history-state Hamiltonian for the
padded circuit described below, with unique ground state
\(\ket{\eta}\)~\cite{kitaev2002classical}.  We use the version without an
output penalty,
\[
  \Hhist(U_x):=H_{\rm in}+H_{\rm prop}+H_{\rm clock},
\]
acting on the full computational and unary-clock Hilbert space, where
\(H_{\rm clock}\) penalizes illegal clock strings.  The total Hilbert space is
\[
  \mathcal H_{\rm comp}\otimes\mathcal H_{\rm clock}
  \otimes\mathcal H_\tau,
  \qquad \mathcal H_\tau\cong\mathbb C^2.
\]
The circuit output qubit belongs to \(\mathcal H_{\rm comp}\), whereas
\(\mathcal H_\tau\) is a separate ancilla qubit with Pauli operators
\(\tau_x,\tau_y,\tau_z\).  Initially it is decoupled, so
\[
  H_0:=\Hhist(U_x)\otimes I_\tau .
\]
Let \(T\) be the number of gates in \(U_x\), and let \(n\) be the number of
computational qubits.  Choose the initial-idling length \(M:=Tn^2\).  This
polynomial choice is made before fixing the
perturbation strength \(r\).  Prepend these \(M\) identity gates to the
computation.  Denote the extended gate sequence with initial-idling by
\[
  \widetilde U_t:=I\quad(1\le t\le M),\qquad
  \widetilde U_{M+j}:=U_j\quad(1\le j\le T).
\]
The history state is then
\[
  \ket{\eta}
  :=
  \frac{1}{\sqrt N}
  \sum_{t=0}^{T+M}
  \widetilde U_t\cdots \widetilde U_1\ket{0^n}\otimes\ket{t}_C,
  \qquad
  N:=T+M+1.
\]
Fix a positive rational constant \(a_{\rm FK}\) below the standard
Feynman--Kitaev gap constant and, throughout the proof, set
\[
  \Dhist:=a_{\rm FK}N^{-3}
\]
so that the actual spectral gap of \(\Hhist(U_x)\) is at least \(\Dhist\)
\cite{kitaev2002classical,gharibian2012hardness}.
Using the unary convention
\(\ket{t}_C:=\ket{1^t0^{T+M-t}}\), the completed computation occurs at
\(t_f:=T+M\).  On the full clock Hilbert space, define the local final-clock
indicator by
\[
  \Pf:=\ket{1}\!\bra{1}_{C,T+M}.
\]
Its restriction to the legal unary-clock subspace is
\(\ket{t_f}\!\bra{t_f}\).  Define the output-acceptance projector and its
final-clock-gated version by
\[
  \Pi_{\rm acc}:=\ket{1}\!\bra{1}_{\rm out},
  \qquad
  \Oacc:=\Pi_{\rm acc}\otimes \Pf .
\]
For constants \(u_0:=3\) and \(c:=2\), define the endpoint perturbation by
\[
  \begin{aligned}
  V(\mathbf{k})
  &:=I_{\rm comp}\otimes\Pf\otimes\Bigl[
  \sin k_x\,\tau_x-\sin k_y\,\tau_y \\
  &\hspace{2.4em} +(u_0+\cos k_x+\cos k_y)\tau_z
  \Bigr]
  -c\Oacc\otimes\tau_z .
  \end{aligned}
\]
Here \(I_{\rm comp}\) acts on the computational register, \(\Pf\) acts on the
clock register, \(\Oacc\) acts on the computational and clock registers, and
the Pauli matrices act on the ancilla qubit.  Identities within
\(\Pi_{\rm acc}\) on computational qubits other than the output qubit are
implicit.  In particular, the
acceptance term contains exactly the one final-clock factor already included
in \(\Oacc\).  The output qubit enters only through \(\Pi_{\rm acc}\); the
Pauli operators \(\tau_a\) never act on it.
The Hamiltonian family is
\begin{equation}
\label{eq:hardness_physical_family}
  H_x(\mathbf{k}):=H_0+rV(\mathbf{k}),
\end{equation}
Here \(r\) is an inverse-polynomial perturbation strength.  Its precise value
will be specified in the Schrieffer--Wolff (SW) analysis below so that
\(r\ll\Dhist\).

The family is local with the same locality convention as the standard
Feynman-Kitaev construction.  
% The history Hamiltonian is local, and the
% endpoint perturbation couples only the output qubit, the final clock qubit,
% and the ancilla qubit \(\tau\).  
By its full-space definition, \(\Pf\) is a one-qubit projector.  Thus
\(I_{\rm comp}\otimes\Pf\otimes\tau_a\) is 2-local and
\(\Oacc\otimes\tau_z\) is 3-local; the full family remains 5-local because
of the Feynman--Kitaev terms.  In particular, \(V(\mathbf{k})\) is
\(O(1)\)-local and has a polynomial-size classical description.  Moreover,
\(V(\mathbf{k})\) is a trigonometric polynomial in \(\mathbf{k}\), so
\(\|\partial_{k_\mu}H_x(\mathbf{k})\|\) is polynomially bounded uniformly
over \( X\), for \(\mu\in\{x,y\}\).

\subsubsection{First-Order Effective QWZ Model}

Let \(P_0:=\ket{\eta}\!\bra{\eta}\otimes I_\tau\) project onto the
two-dimensional ground space of \(H_0\), where the twofold degeneracy is the
ancilla-qubit degree of freedom.  Explicitly,
\[
  \operatorname{range}(P_0)
  =\operatorname{span}\bigl\{
  \ket{\eta}\otimes\ket{0}_\tau,
  \ket{\eta}\otimes\ket{1}_\tau
  \bigr\}.
\]
The dimensionless perturbative parameter is
\(r\sup_{\mathbf{k}}\|V(\mathbf{k})\|/\Dhist\).  When this ratio is sufficiently small, the exact
Schrieffer--Wolff transformation block diagonalizes \(H_0+rV(\mathbf{k})\),
and its restriction to the low-energy block has the convergent expansion
\[
  H_{\rm eff}^{\rm SW}(\mathbf{k})
  =E_0P_0+rP_0V(\mathbf{k})P_0
   +\sum_{j\ge2}r^jH_{{\rm eff},j}(\mathbf{k}).
\]
The next section supplies the precise convergence condition and bounds the
summed higher-order terms by \(O(r^2\|V\|^2/\Dhist)\).  Thus the linear term
selects the target topological model, while the later stability argument shows
that the remaining terms cannot close its gap or change its Chern number.

We now compute this linear term.  Since the history state
has uniform amplitudes over the \(N=T+M+1\) clock sites,
\[
  \bra{\eta}\Pf\ket{\eta}=\frac1N,\qquad
  \bra{\eta}\Oacc\ket{\eta}=\frac{\pacc}{N}.
\]
Thus the first-order effective Hamiltonian on \(\operatorname{range}(P_0)\) is
\[
  \begin{aligned}
  P_0 rV(\mathbf{k})P_0
  &=\frac rN\Bigl[
  \sin k_x\,\tau_x-\sin k_y\,\tau_y \\
  &\hspace{3.9em}
  +\bigl(\ueff(\pacc)+\cos k_x+\cos k_y\bigr)\tau_z
  \Bigr].
  \end{aligned}
\]
where
\[
  \ueff(p):=u_0-cp=3-2p .
\]
This is the two-band Qi--Wu--Zhang model~\cite{qi_topological_2006} up to a
positive scale factor.  We obtain the following straightforward lemma.

\begin{lemma}[QWZ calibration]
For the above choice \(u_0=3,c=2\), no instances have
\(\ueff(\pacc)\ge 7/3>2\), while yes instances have
\(\ueff(\pacc)\le 5/3\in(0,2)\).  With the chosen orientation, the corresponding
QWZ ground-state Chern numbers are \(0\) and \(1\), respectively.
\end{lemma}

\subsubsection{Exact Schrieffer--Wolff Stability}

We use the exact Schrieffer--Wolff transformation with a deliberately
conservative perturbation strength.
Set
\begin{equation}
\label{eq:hardness_perturbation_norm}
  B_V:=|u_0|+c+4=9,
\end{equation}
so that \(\|V(\mathbf{k})\|\le B_V\) uniformly over \( X\).

\begin{lemma}[Exact SW estimate]
\label{lem:exact_sw}
Let \(P_0\) be the rank-two ground eigenspace of \(H_0\), with eigenvalue
\(E_0\), separated from the rest of the spectrum by at least \(\Dhist\), and let
\(V(\mathbf{k})\) be periodic real-analytic with
\(\|V(\mathbf{k})\|\le B_V\) uniformly in \(\mathbf{k}\).
If
\[
  rB_V\le \frac{\Dhist}{32},
\]
then $H_x(\mathbf{k}):=H_0+rV(\mathbf{k})$ has an isolated rank-two
low-energy cluster and a periodic real-analytic direct rotation
\(W(\mathbf{k})\) mapping its spectral projector to \(P_0\).  The effective
Hamiltonian
\begin{equation}
\label{eq:hardness_exact_sw_block}
  H_{\rm eff}^{\rm SW}(\mathbf{k})
  :=
  P_0W(\mathbf{k})H_x(\mathbf{k})W(\mathbf{k})^\dagger P_0
\end{equation}
has remainder
\begin{align}
  R_{\rm SW}(\mathbf{k})
  &:=H_{\rm eff}^{\rm SW}(\mathbf{k})-E_0P_0-rP_0V(\mathbf{k})P_0, \\
  \|R_{\rm SW}(\mathbf{k})\|
  &\le 4\frac{r^2B_V^2}{\Dhist}.
\end{align}
uniformly over \( X\).
\end{lemma}

\begin{proof}
This is a standard consequence of exact Schrieffer--Wolff perturbation
theory~\cite{bravyi2011schriefferwolff}.  We record a short estimate for the
constants used here.  Shift \(E_0\) to zero and
write \(e:=rB_V\) and \(s:=e/\Dhist\le1/32\).  The projector estimate in
Lemma 3.1 of that reference gives
\(\|P-P_0\|\le2s\) for the perturbed cluster projector \(P\).  Its
map \(Y:\operatorname{range}(P_0)\to\operatorname{range}(I-P_0)\) therefore
has norm \(y:=\|Y\|\le2s/\sqrt{1-4s^2}\le3s\).  With
\(S:=(I+Y^\dagger Y)^{1/2}\), \(A:=P_0H_xP_0\), and
\(B:=P_0H_x(I-P_0)\), we have \(\|A\|,\|B\|\le e\) after the shift, and
the canonical direct rotation gives
\[
  \begin{aligned}
  H_{\rm eff}^{\rm SW}&=S(A+BY)S^{-1},\\
  \|R_{\rm SW}\|
  &\le e\bigl(y^2+y\sqrt{1+y^2}\bigr)
  \le4\frac{e^2}{\Dhist}.
  \end{aligned}
\]
Lemma 3.4 of the same reference also places \(s\le1/32\) inside its
absolute-convergence radius.  The Riesz-projector construction makes the
cluster projector, and hence the direct rotation, periodic and real-analytic
in \(\mathbf{k}\).  This proves the claim.
\end{proof}

\subsubsection{Uniform Gap and Chern Number Transfer}
\label{sec:gap_chern_transfer}

We now upgrade the first-order identification to a statement about the full
Hamiltonian family.  There are two gaps to control: the splitting of the two
bands inside the low-energy rank-two cluster, and the separation of that
cluster from the history-excited spectrum.  The basic principle is simple.
If two Hamiltonian families can be joined continuously without closing their
ground-band gap, then no topological phase transition occurs along the way,
and their Chern numbers agree.  We use this gapped interpolation first to
compare the QWZ and exact SW Hamiltonians, and then use the exact SW unitary
to compare the latter with the physical ground band.

Recall that, after identifying \(\operatorname{range}(P_0)\) with the
ancilla-qubit space,
\[
  P_0V(\mathbf{k})P_0
  =
  \frac1N
  H_{\rm QWZ}\bigl(\mathbf{k};\ueff(\pacc)\bigr).
\]
Let \(c_0:=1/3\), the distance of both promise sectors from the nearest QWZ
phase boundary: no instances have \(\ueff\ge7/3\), whereas yes instances have
\(1\le\ueff\le5/3\).  Hence, before multiplication by \(r/N\), the two QWZ
bands are separated by at least \(2c_0\).  Using \(B_V=9\) from
\cref{eq:hardness_perturbation_norm}, choose the positive rational constant
\(\alpha\) by
\[
  \alpha:=\frac1{4096}<
  \min\left\{
    \frac{1}{32B_V},
    \frac{c_0}{16B_V^2},
    \frac{1}{8B_V}
  \right\}=\frac1{3888},
\]
and set
\[
  r:=\frac{\alpha\Dhist}{N}.
\]
The three restrictions on \(\alpha\) respectively place the perturbation in
the SW convergence regime, make the full higher-order tail small compared
with the QWZ gap, and keep the low-energy cluster away from the
history-excited spectrum.

\paragraph{SW remainder.}
By the definition of \(r\),
\[
  \frac{r}{\Dhist}=\frac{\alpha}{N}.
\]
The SW remainder lemma therefore gives
\[
  \|R_{\rm SW}(\mathbf{k})\|
  \le
  4\frac{r^2B_V^2}{\Dhist}
  \le
  4\alpha B_V^2\frac rN
  \le
  \frac{c_0r}{4N}.
\]
This estimate controls the sum of all terms of order two and higher, not only
the second-order correction.  It is therefore enough to compare the exact SW
Hamiltonian directly with its first-order QWZ part.

\paragraph{Uniform spectral gaps.}
The two eigenvalues of the first-order QWZ Hamiltonian are separated by at
least \(2c_0r/N\).  Weyl's inequality shows that the two eigenvalues of the
exact SW Hamiltonian are separated by at least
\[
  \frac{2c_0r}{N}-2\|R_{\rm SW}\|
  \ge\frac{3c_0r}{2N}
  \ge\frac{c_0r}{N}.
\]
Because the SW transformation is unitary, this is also the physical
ground-band gap.  Thus the full Hamiltonian has an isolated, nondegenerate
ground band, and we may choose the supplied gap lower bound as
\begin{align}
  \Delta_{\min}:=\frac{c_0r}{N}=\Omega(N^{-5}).
\end{align}
The choice \(\alpha\le1/(8B_V)\) and another application of
Weyl's inequality keep the rank-two cluster separated from all
history-excited states by at least \(\Dhist/2\).
Thus the internal ground-band gap is \(\Omega(N^{-5})\), while the larger
external gap from the low-energy cluster to the history excitations is
\(\Omega(N^{-3})\).

\paragraph{Chern-number transfer.}
Norm closeness alone does not yet identify an integer topological invariant.
Instead, we turn on the remainder continuously.  On
\(\operatorname{range}(P_0)\), consider the straight-line interpolation
\begin{align}
  G_s(\mathbf{k})
  &:=E_0P_0
  +\frac rN H_{\rm QWZ}\qty(\mathbf{k};\ueff(\pacc))
  \notag\\
  &\quad +sR_{\rm SW}(\mathbf{k}),
  \qquad 0\le s\le1.
  \label{eq:hardness_gapped_homotopy}
\end{align}
The preceding norm estimate keeps its two bands separated for every
\((s,\mathbf{k})\).  Thus the lower-band projector varies smoothly with
\(s\), and its Chern number---an integer---cannot change continuously.
Hence the lower-band Chern number of \(G_s\) is independent of \(s\).  At
\(s=0\) it is the QWZ Chern number, while at \(s=1\) it is the lower-band
Chern number of the exact SW Hamiltonian.

Finally, the periodic real-analytic unitary \(W(\mathbf{k})\) is a smooth
change of basis: it maps the physical ground line bundle isomorphically onto
the lower-band bundle of \(H_{\rm eff}^{\rm SW}(\mathbf{k})\).  A smooth
change of basis does not alter the first Chern number.  These two
steps---gapped interpolation followed by the SW bundle isomorphism---give,
for the original physical Hamiltonian,
\[
  C[P]=C_{\rm QWZ}(\ueff(\pacc)).
\]

The conservative choice
\(r=\Theta(\Dhist/N)=\Theta(N^{-4})\) gives a
slightly smaller band-gap bound than a structure-specific higher-order
analysis, but preserves every inverse-polynomial promise used by the
algorithm.

\subsubsection{Guiding State}

At the fixed base point \(\mathbf{k}_0:=(0,0)\), the effective QWZ Hamiltonian is
\((\ueff(\pacc)+2)\tau_z\), with \(\ueff(\pacc)\in[1,3]\).  Its ground spinor
is therefore the fixed \(\tau_z\)-eigenstate \(\ket{1}_\tau\), using
\(\tau_z\ket{1}=-\ket{1}\).  This choice is independent of the unknown
acceptance probability.
Moreover, at \(\mathbf{k}_0\) the full perturbation contains only
\(I_\tau\) and \(\tau_z\).  Hence \(H_x(\mathbf{k}_0)\), its low-energy
projector, and the direct rotation commute with \(\tau_z\).
The preceding gap-preserving interpolation therefore shows that the exact ground state has
the same fixed ancilla-qubit state \(\ket{1}_\tau\), not merely its first-order
approximation.

For the history-state factor, use the efficiently preparable initial-idling
state
\[
  \ket{c_{\rm hist}}
  :=
  \ket{0^n}\otimes
  \frac{1}{\sqrt{M+1}}\sum_{t=0}^{M}\ket t .
\]
Since the first \(M\) gates are identities,
\[
  \abs{\braket{c_{\rm hist}}{\eta}}^2
  =\frac{M+1}{N}=1-\frac{T}{N}.
\]
Our choice of \(M\) gives \(T/N\le 1/(n^2+1)\le 1/2\).
The full guiding state is
\[
  \ket c:=\ket{c_{\rm hist}}\otimes\ket{1}_\tau,
\]
and satisfies
\[
  \abs{\braket{c}{u(\mathbf{k}_0)}}^2
  \ge
  1-\frac{T}{N}
  -O\!\left(\frac{r}{\Dhist}\right)
  =\Omega(1).
\]
Initial idling controls only the guiding-state overlap; the single final-clock
projector controls only the acceptance signal.

\subsubsection{Verification of the Algorithmic Promises}

We finish the reduction by checking that the constructed family satisfies the
input and access promises used by the algorithmic upper bound.

\paragraph{Succinct description and derivative bounds.}
The reduction outputs a polynomial list of standard Feynman--Kitaev terms,
the supports of the final-clock, output, and ancilla-qubit terms, the rational
constants \(u_0,c,r\), and the symbolic coefficient functions
\(\sin k_\mu,\cos k_\mu\), \(\mu\in\{x,y\}\).  These coefficients can be
evaluated to \(b\) bits in time polynomial in the input size and \(b\).
Moreover,
\[
  \begin{gathered}
  \|H_x(\mathbf{k})\|\le\poly(n), \\
  \max_{\mathbf{k}\in X}
  \|\partial_{k_\mu}H_x(\mathbf{k})\|\le2r,
  \quad \mu\in\{x,y\}.
  \end{gathered}
\]
and
\[
  \begin{gathered}
  \max_{\mathbf{k}\in X}
  \|\partial_{k_\mu}\partial_{k_\nu}H_x(\mathbf{k})\|
  \le2r, \\
  \mu,\nu\in\{x,y\}.
  \end{gathered}
\]
Thus the usual local-term block encoding and Hamiltonian simulation have
polynomial cost.

\paragraph{Connection to the algorithm.}
The explicit guiding state and inverse-polynomial gap provide the two reference-state copies
required by \cref{sec:upper_bound_algorithm}, using the preparation procedure in the containment
argument.
The succinct local description supplies controlled forward and inverse Hamiltonian simulation with
polynomial gate cost.
Thus the constructed instances satisfy the resources needed to apply that algorithm in the
classical-input model.

\paragraph{Completion of the reduction.}
Putting the pieces together: the reduction constructs \(H_x(\mathbf{k})\) in
polynomial time; locality and polynomial derivative bounds hold by construction; the
spectral gap is inverse-polynomial by the perturbative stability analysis; the
true ground band has the same Chern number as the QWZ model by the exact SW
bundle isomorphism and the effective-Hamiltonian homotopy; and the guiding
state is explicit.  Hence yes
instances give \(C=1\) and no instances give \(C=0\).
In this precise sense, the reduction compiles the computational outcome into
whether a uniformly gapped ground band is topologically trivial or carries
unit Chern number.

\section{Discussion and Conclusion}
\label{sec:discussion}

\paragraph{Conclusion.}
We have established a nearly optimal quantum algorithm 
for Chern-number estimation and characterized its
computational complexity.
Using GQSP and two ground-state copies prepared at a single reference point,
the algorithm achieves a Hamiltonian-oracle evolution time of $\widetilde
	O(L_xL_y/\Delta_{\min}^3)$ under the stated access assumptions.
    A matching worst-case oracle lower bound under the same access model
    establishes optimality up to polylogarithmic factors.
The same asymptotic oracle lower-bound scaling holds even under the promise $C\in\qty{0,1}$.
For succinct local Hamiltonian families satisfying the gap and guiding-state
promises, we also prove that exact Chern-number computation is in \(\FBQP\)
and is \(\BQP\)-hard.  Its restriction to Chern number zero or one is
\(\BQP\)-complete.  Thus, even extracting a quantized integer invariant can
encode arbitrary quantum computation.
These results establish both an achievable quantum runtime and limits on its improvement, while
showing that quantization alone does not remove worst-case computational hardness.

\paragraph{Discussion.}
An important next step is to determine how additional physical structure changes these conclusions.
For local many-body systems parametrized by boundary twists, locality and a gap can suppress twist
dependence, enabling estimates without full parameter-space integration in suitable
regimes~\cite{watanabe_twist_2018,kudo_chern_2019}.
Exploiting such structure may lead to more efficient quantum algorithms.
A complementary question is whether the computational
hardness established here persists within physically motivated classes, such as geometrically local
lattice Hamiltonians whose parameters enter only through boundary twists.
This question is not settled by our reduction, which allows
general periodic dependence of local terms
without imposing the boundary-twist restriction.
The same restriction also raises a separate question for oracle complexity.
The matching oracle lower bound is proved for general smooth periodic Hamiltonian families and
holds even under $C\in\qty{0,1}$, independently of the
$\BQP$-completeness result for the classical-input problem.
It remains open whether the same oracle lower bound holds for geometrically local lattice
Hamiltonians whose parameters enter only through boundary twists.
Our results are also related to recent work on the hardness of recognizing phases of matter~\cite{schuster_phases_2025}, where quantum hardness was established for ground states of log-local Hamiltonians.
In contrast, we establish classical hardness of Chern number estimation for constant-local Hamiltonians.
It would be interesting to investigate the relationship between such phase-recognition problems and the complexity of extracting explicit topological invariants such as the Chern number.

Extending the algorithm to an isolated, possibly degenerate ground-state manifold would connect
the framework more directly to fractional Chern insulators~\cite{chen_nonabelian_mote2_2025}.
This requires treating higher-rank projectors and matrix-valued holonomies, together with suitable
preparation resources and a coherent readout of the manifold's Chern number.
Such an extension would broaden the interacting phases accessible to the algorithm beyond the
unique-ground-state setting considered here.

Finally, assessing quantum advantage in concrete models requires resources beyond
Hamiltonian-oracle time.
The system-size dependence of the gap and derivative bounds, the cost of reference-state
preparation, and the gate costs of Hamiltonian simulation and GQSP must be included.
Comparing these total resources with exact diagonalization and tensor-network methods would
identify regimes in which the algorithm offers an advantage for physically relevant systems.
This would connect the present complexity results to the study of correlated topological phases and
their stability at larger system sizes.

\begin{acknowledgments}
We thank Takahiro Morimoto for helpful discussions.
This work was supported by the Center of Innovation for Sustainable Quantum AI
(JST Grant No.~JPMJPF2221) and JST ASPIRE (Grant No.~JPMJAP2319).
S.I. acknowledges support from FoPM (a WINGS Program) and the JSR Fellowship
at the University of Tokyo.
S.A. acknowledges support from the Toyota Riken Overseas Fellowship.
K.S. acknowledges support from MEXT Q-LEAP (Grant No.~JPMXS0120319794)
and JST SPRING (Grant No.~JPMJSP2138).
R.H. acknowledges support from JST PRESTO (Grant No.~JPMJPR23F9)
and JST ASPIRE (Grant No.~JPMJAP26A4), Japan.
C.K. acknowledges support from JST PRESTO (Grant No.~JPMJPR25F1)
and JSPS KAKENHI (Grant No.~JP26K17052).
\end{acknowledgments}

\appendix
\crefalias{section}{appendix}
\crefalias{subsection}{subappendix}
\crefalias{subsubsection}{subsubappendix}

\section{Proof Details for the Chern-Number Algorithm}
\label{app:nonadiabatic_algorithm}

We prove the geometric identity, local implementation guarantees, and
resource bounds used in \cref{sec:upper_bound_algorithm}.

Throughout, the Hamiltonian, ground-state projector, Berry connection,
and supplied bounds are those of \cref{sec:chern_definition}; the mesh
and local primitives are those of \cref{sec:upper_bound_algorithm}.

\subsection{Geometric bounds and the exact phase-sum identity}
\label{app:alg_geometric_proofs}

The geometric estimates serve two purposes: they certify the
early-return case and control the phases on the finite mesh.
We first establish the projector bounds, then use them to prove the
exact phase-sum identity.

\subsubsection{Projector bounds and the zero-Chern certificate}
\label{app:alg_geometry}

The spectral gap controls how rapidly the ground-state projector
can change with the parameters.
Let $Q:=I-P$ and introduce the reduced resolvent
\begin{align}
    S(\mathbf{k})
    :=\qty[Q\qty(H(\mathbf{k})-E_0(\mathbf{k}))Q]^{-1}Q,
    \label{eq:alg_reduced_resolvent}
\end{align}
where the inverse acts on $\operatorname{Ran}Q$. The spectral gap implies
$\norm{S}=1/\Delta(\mathbf{k})$. Differentiating $HP=E_0P$ and
$P^2=P$ gives
\begin{align}
    X_\mu:=Q(\partial_{k_\mu}P)P
    &=-S(\partial_{k_\mu}H)P,
    \notag\\
    P(\partial_{k_\mu}P)P&=Q(\partial_{k_\mu}P)Q=0.
    \label{eq:alg_projector_blocks}
\end{align}
Thus, with respect to $\operatorname{Ran}P\oplus\operatorname{Ran}Q$,
\begin{align}
    \partial_{k_\mu}P
    &=\begin{pmatrix}0&X_\mu^\dagger\\X_\mu&0\end{pmatrix},
    \notag\\
    \norm{\partial_{k_\mu}P}
    &=\norm{X_\mu}
    \leq\frac{\norm{\partial_{k_\mu}H}}{\Delta(\mathbf{k})}.
    \label{eq:alg_projector_proof}
\end{align}
The rank-one projector in \cref{eq:berry_curvature_projector_form}
then yields
\begin{align}
    |\Omega(\mathbf{k})|
    &\leq2\norm{\partial_{k_x}P}\norm{\partial_{k_y}P}
    \notag\\
    &\leq\frac{2\norm{\partial_{k_x}H}\norm{\partial_{k_y}H}}
                      {\Delta(\mathbf{k})^2}.
    \label{eq:alg_curvature_proof}
\end{align}
This proves \cref{eq:alg_projector_bound} and
\cref{eq:berry_curvature_hamiltonian_bound}. Integrating over the torus
of area $4\pi^2$ proves \cref{eq:alg_chern_bound}.

A sufficiently small derivative bound in either direction guarantees
a smooth periodic ground state over the whole torus.
Suppose $L_x<\Delta_{\min}/2$. On the circle $k_x=0$, choose a smooth,
periodic normalized ground state $\ket{u_{\mathrm{ref}}(k_y)}$.
For any $\mathbf{k}$, horizontally transport this state along a shortest
$k_x$ arc from $(0,k_y)$ to $\mathbf{k}$, of length $\ell\leq\pi$.
Writing $P(s)$ for the projector along this arc, with $s$ its arc-length
parameter, denote the transported normalized state by
$\ket{u_{\parallel}(s)}$. It satisfies
\begin{align}
    \ket{u_{\parallel}(0)}&=\ket{u_{\mathrm{ref}}(k_y)},
    \notag\\
    P(s)&=\ket{u_{\parallel}(s)}\bra{u_{\parallel}(s)}.
    \label{eq:alg_zero_transport_state}
\end{align}
We fix its phase along the arc by imposing the horizontal-transport
condition
\begin{align}
    \bra{u_{\parallel}(s)}\partial_s\ket{u_{\parallel}(s)}=0.
    \label{eq:alg_zero_horizontal_transport}
\end{align}
In particular,
$P(\mathbf{k})=\ket{u_{\parallel}(\ell)}\bra{u_{\parallel}(\ell)}$.
Differentiating $P(s)\ket{u_{\parallel}(s)}=\ket{u_{\parallel}(s)}$
and using \cref{eq:alg_zero_horizontal_transport} gives
$\partial_s\ket{u_{\parallel}(s)}=(\partial_sP(s))\ket{u_{\parallel}(s)}$.
Thus, \cref{eq:alg_projector_bound} bounds the length of the
transported state-vector curve on the unit sphere by
\begin{align}
    \int_0^\ell ds\,\norm{\partial_s\ket{u_{\parallel}(s)}}
    &\leq\int_0^\ell ds\,\norm{\partial_sP(s)}
    \notag\\
    &\leq\frac{\pi L_x}{\Delta_{\min}}<\frac{\pi}{2}.
    \label{eq:alg_zero_path_length}
\end{align}
Orthogonal unit vectors form an angle of $\pi/2$.
On the unit sphere, a shortest path between them follows a circular arc
of unit radius through this angle and therefore has length $\pi/2$.
Any curve on the sphere joining them must therefore be at least this long.
Since \cref{eq:alg_zero_path_length} bounds the transported curve's length
strictly below $\pi/2$, $\ket{u_{\parallel}(\ell)}$ is not orthogonal to
$\ket{u_{\mathrm{ref}}(k_y)}$.
The endpoint projector therefore obeys
$P(\mathbf{k})\ket{u_{\mathrm{ref}}(k_y)}\neq 0$, and
\begin{align}
    \ket{u_X(\mathbf{k})}
    :=\frac{P(\mathbf{k})\ket{u_{\mathrm{ref}}(k_y)}}
          {\norm{P(\mathbf{k})\ket{u_{\mathrm{ref}}(k_y)}}}
    \label{eq:alg_zero_global_section}
\end{align}
is a normalized ground state defined in a single smooth, periodic gauge
over the entire torus. In this global gauge, Stokes' theorem gives $C=0$.
Interchanging $x$ and $y$ proves the
other case of \cref{eq:alg_zero_check}.

In the quantum case, $L_\mu/\Delta_{\min}\geq1/2$, and the ceiling function in
\cref{eq:alg_mesh} changes the mesh counts only by constant factors.
Also $L_\mu\,\delta k_\mu\leq\Delta_{\min}/64$. Integrating the Hamiltonian and
projector derivatives along an axis-aligned path inside a cell 
with \cref{eq:alg_projector_bound} gives
\begin{align}
    \norm{H(v)-H(w)}&\leq L_x\,\delta k_x+L_y\,\delta k_y
        \leq\frac{\Delta_{\min}}{32},
    \notag\\
    \norm{P(v)-P(w)}&\leq\frac{L_x\,\delta k_x+L_y\,\delta k_y}
        {\Delta_{\min}}\leq\rho.
    \label{eq:alg_cell_lipschitz_proof}
\end{align}
Cells crossing a periodic seam are represented by a rectangle in a lifted
coordinate chart, so the same bounds apply there.
These cell bounds will ensure that the vertex overlaps stay nonzero
and that the Bargmann phases remain small.

\subsubsection{Bargmann phases and branch bounds}
\label{app:alg_bargmann_bounds}

We bound the Bargmann phase. Rank-one projectors obey
\begin{align}
    |\braket{u(v)}{u(w)}|^2
    =1-\norm{P(v)-P(w)}^2\geq1-\rho^2>0
    \label{eq:alg_overlap_bound}
\end{align}
for vertices of a cell.
To bound the imaginary parts of the overlaps between adjacent vertices,
we choose the phases of the ground states relative to the anchor state.
Choose their phases so that
$r_v:=\braket{u(a)}{u(v)}$ is positive for $v=b,c,d$, and set
$\ket{v_\perp}:=\ket{u(v)}-r_v\ket{u(a)}$, where
$\braket{u(a)}{v_\perp}=0$ and $\norm{v_\perp}\leq\rho$.
For $(v,w)=(b,c),(c,d)$,
\begin{align}
    z_{vw}&:=\braket{u(v)}{u(w)}
    =r_vr_w+\braket{v_\perp}{w_\perp},
    \notag\\
    \operatorname{Re}z_{vw}&\geq1-2\rho^2,
    \quad |\operatorname{Im}z_{vw}|\leq\rho^2.
    \label{eq:alg_aligned_overlap}
\end{align}
The other two factors in $B_{\mathcal C}=r_bz_{bc}z_{cd}r_d$ are
positive. Therefore
\begin{align}
    |\phi_{\mathcal C}|
    \leq2\arctan\frac{\rho^2}{1-2\rho^2}
    <\frac{\pi}{4}.
    \label{eq:alg_bargmann_branch_bound}
\end{align}

\subsubsection{From Stokes' theorem to Bargmann phases}
\label{app:alg_phase_sum}

We prove the phase-sum identity by pairing shared-edge integrals,
evaluating their transition phases, and fixing the remaining branch ambiguity.
Stokes' theorem expresses $2\pi C$ as a sum of Berry-connection
integrals around the cells. Let
$A_{\mathcal C}:=i\ev{d}{u_{\mathcal C}}$ be the Berry connection in a
smooth local gauge on $\mathcal C$, with components
$A_{\mathcal C,\mu}:=i\ev{\partial_{k_\mu}}{u_{\mathcal C}}$
for $\mu\in\qty{x,y}$.
Taking each boundary in the
counterclockwise order $a\to b\to c\to d\to a$ gives
\begin{align}
    2\pi C
    &=\sum_{\mathcal C}\int_{\mathcal C}d^2\mathbf{k}\,\Omega(\mathbf{k})
    \notag\\
    &=\sum_{\mathcal C}\oint_{\partial\mathcal C}A_{\mathcal C}.
    \label{eq:alg_cell_stokes}
\end{align}

\begin{figure}[bt]
    \centering
    \begin{tikzpicture}[
        x=1pt,y=1pt,
        cell/.style={draw,line width=0.5pt},
        loop/.style={-{Latex[length=3.5pt]},line width=0.45pt,gray,
            shorten >=2pt,shorten <=2pt},
        pairx/.style={-{Latex[length=4pt]},line width=1pt,
            blue!70!black,shorten >=2pt,shorten <=2pt},
        pairy/.style={-{Latex[length=4pt]},line width=1pt,
            red!70!black,shorten >=2pt,shorten <=2pt},
        anc/.style={circle,fill,inner sep=1.6pt},
        vtx/.style={circle,draw,fill=white,inner sep=1.3pt},
        lbl/.style={font=\small,inner sep=1.5pt}
    ]
        \def\s{54}   % cell size
        \def\o{7}    % inset of the oriented boundaries
        % cells
        \draw[cell] (0,0) rectangle (\s,\s);
        \draw[cell] (\s,0) rectangle (2*\s,\s);
        \draw[cell] (0,\s) rectangle (\s,2*\s);
        % counterclockwise boundary of C, a->b->c->d->a
        \draw[loop] (\o,\o) -- (\s-\o,\o);
        \draw[pairx] (\s-\o,\o) -- (\s-\o,\s-\o);
        \draw[pairy] (\s-\o,\s-\o) -- (\o,\s-\o);
        \draw[loop] (\o,\s-\o) -- (\o,\o);
        % counterclockwise boundary of C_x (anchored at b)
        \draw[loop] (\s+\o,\o) -- (2*\s-\o,\o);
        \draw[loop] (2*\s-\o,\o) -- (2*\s-\o,\s-\o);
        \draw[loop] (2*\s-\o,\s-\o) -- (\s+\o,\s-\o);
        \draw[pairx] (\s+\o,\s-\o) -- (\s+\o,\o);
        % counterclockwise boundary of C_y (anchored at d)
        \draw[pairy] (\o,\s+\o) -- (\s-\o,\s+\o);
        \draw[loop] (\s-\o,\s+\o) -- (\s-\o,2*\s-\o);
        \draw[loop] (\s-\o,2*\s-\o) -- (\o,2*\s-\o);
        \draw[loop] (\o,2*\s-\o) -- (\o,\s+\o);
        % cell labels
        \node[lbl] at (0.5*\s,0.5*\s) {$\mathcal C$};
        \node[lbl] at (1.5*\s,0.5*\s) {$\mathcal C_x$};
        \node[lbl] at (0.5*\s,1.5*\s) {$\mathcal C_y$};
        % vertices: filled = gauge anchors of C, C_x, C_y
        \node[anc] (a) at (0,0) {};
        \node[anc] (b) at (\s,0) {};
        \node[vtx]    (c) at (\s,\s) {};
        \node[anc] (d) at (0,\s) {};
        \node[lbl,below left]  at (a) {$a$};
        \node[lbl,below]       at (b) {$b$};
        \node[lbl,above right] at (c) {$c$};
        \node[lbl,left=1.5pt]  at (d) {$d$};
        % mesh spacings
        \draw[{Latex[length=3pt]}-{Latex[length=3pt]},line width=0.4pt]
            (0,-12) -- node[lbl,below,font=\scriptsize]{$\delta k_x$} (\s,-12);
        \draw[{Latex[length=3pt]}-{Latex[length=3pt]},line width=0.4pt]
            (-14,\s) -- node[lbl,left,font=\scriptsize]{$\delta k_y$} (-14,2*\s);
        % axes
        \draw[-{Latex[length=4pt]},line width=0.5pt] (\s+26,72) -- (\s+52,72)
            node[lbl,right,font=\scriptsize]{$k_x$};
        \draw[-{Latex[length=4pt]},line width=0.5pt] (\s+26,72) -- (\s+26,102)
            node[lbl,left,font=\scriptsize]{$k_y$};
    \end{tikzpicture}
    \caption{Cell $\mathcal C=(a,b,c,d)$ and its right and upper
    neighbors $\mathcal C_x$ and $\mathcal C_y$. Vertices are labeled
    counterclockwise from the lower-left anchor $a$. Filled dots mark the
    anchors of the local gauges in \cref{eq:alg_cell_gauge}: $a$ for
    $\mathcal C$, $b$ for $\mathcal C_x$, and $d$ for $\mathcal C_y$.
    Arrows show the counterclockwise boundaries in
    \cref{eq:alg_cell_stokes}. Each shared edge is traversed twice with
    opposite orientations: the right edge by $\mathcal C$ ($b\to c$) and
    by $\mathcal C_x$ ($c\to b$), shown in blue; the upper edge by
    $\mathcal C$ ($c\to d$) and by $\mathcal C_y$ ($d\to c$), shown in
    red.}
    \label{fig:alg_cell_pairing}
\end{figure}
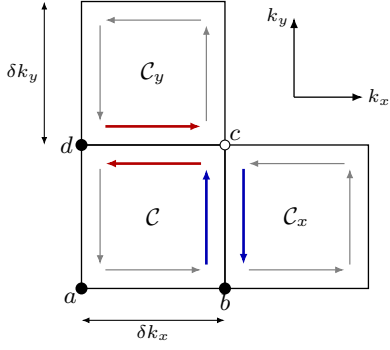

Every shared edge occurs twice with opposite orientations, including
edges identified by periodicity. Each contribution uses the local gauge
of its cell. To count each pair of shared-edge integrals exactly once, we let each
cell collect the paired contributions on its right and upper edges.
Thus, each vertical edge is counted with the cell on its left,
and each horizontal edge with the cell below. For a cell
$\mathcal C=(a,b,c,d)$, let $\mathcal C_x$ and $\mathcal C_y$ be its
right and upper neighbors, as shown in \cref{fig:alg_cell_pairing}. Collect the paired contributions on its
right edge $b\to c$ and upper edge $c\to d$ as
\begin{align}
    E_{\mathcal C}
    &:=\int_b^c dk_y\,\qty(A_{\mathcal C,y}-A_{\mathcal C_x,y})
    \notag\\
    &\quad+\int_c^d dk_x\,\qty(A_{\mathcal C,x}-A_{\mathcal C_y,x}).
    \label{eq:alg_paired_edge_integrals}
\end{align}
Each integral follows the indicated edge; the minus sign accounts for
the neighboring cell's opposite orientation. Every edge pair is counted
once, so \cref{eq:alg_cell_stokes} becomes
\begin{align}
    2\pi C=\sum_{\mathcal C}E_{\mathcal C}.
    \label{eq:alg_transition_sum}
\end{align}
Thus $E_{\mathcal C}$ collects paired edge contributions to the global
sum; it need not equal the Berry flux through $\mathcal C$.

To relate $E_{\mathcal C}$ to the Bargmann invariant, take the lower-left
vertex $a$ of each cell as its anchor, fix a normalized ground state
$\ket{u(a)}$ there, and choose
\begin{align}
    \ket{u_{\mathcal C}(\mathbf{k})}
    :=\frac{P(\mathbf{k})\ket{u(a)}}
          {\norm{P(\mathbf{k})\ket{u(a)}}}.
    \label{eq:alg_cell_gauge}
\end{align}
Its denominator is at least $\sqrt{1-\rho^2}$ by
\cref{eq:alg_cell_lipschitz_proof}, so this gauge is smooth throughout
the cell.

The mismatch between neighboring gauges turns each paired edge integral
into an endpoint phase difference.
The neighbors $\mathcal C_x$ and $\mathcal C_y$ are anchored
at $b$ and $d$, respectively. On the shared edges, choose continuous
transition phases satisfying
\begin{align}
    \ket{u_{\mathcal C_x}}&=e^{i\chi_x}\ket{u_{\mathcal C}},
    &A_{\mathcal C_x, \mu}&=A_{\mathcal C,\mu}-\partial_{k_\mu}\chi_x,
    \notag\\
    \ket{u_{\mathcal C_y}}&=e^{i\chi_y}\ket{u_{\mathcal C}},
    &A_{\mathcal C_y, \mu}&=A_{\mathcal C,\mu}-\partial_{k_\mu}\chi_y
    \label{eq:alg_transition_connections}
\end{align}
for $\mu\in\{x, y\}$. Substituting into \cref{eq:alg_paired_edge_integrals} reduces the
integrals to endpoint differences:
\begin{align}
    E_{\mathcal C}
    &=\chi_x(c)-\chi_x(b)+\chi_y(d)-\chi_y(c).
    \label{eq:alg_paired_edge_phases}
\end{align}

The projector gauges make these transition phases explicit in terms
of ground-state overlaps.
For $z\ne0$, write $\operatorname{ph}(z):=z/|z|$. The rank-one identity
in \cref{eq:alg_cell_gauge} gives
\begin{align}
    e^{i\chi_x(\mathbf{k})}
    &=\operatorname{ph}\qty(\bra{u(a)}P(\mathbf{k})\ket{u(b)}),
    \notag\\
    e^{i\chi_y(\mathbf{k})}
    &=\operatorname{ph}\qty(\bra{u(a)}P(\mathbf{k})\ket{u(d)}).
    \label{eq:alg_transition_explicit}
\end{align}
In particular,
\begin{align}
    e^{i\chi_x(b)}&=\operatorname{ph}\qty(\braket{u(a)}{u(b)}),
    \notag\\
    e^{i\chi_x(c)}&=\operatorname{ph}\qty(
        \braket{u(a)}{u(c)}\braket{u(c)}{u(b)}),
    \notag\\
    e^{i\chi_y(d)}&=\operatorname{ph}\qty(\braket{u(a)}{u(d)}),
    \notag\\
    e^{i\chi_y(c)}&=\operatorname{ph}\qty(
        \braket{u(a)}{u(c)}\braket{u(c)}{u(d)}).
    \label{eq:alg_transition_endpoints}
\end{align}
Multiplying the four phases with the signs in
\cref{eq:alg_paired_edge_phases} cancels the diagonal overlap and gives
\begin{align}
    e^{iE_{\mathcal C}}
    =\operatorname{ph}\qty(B_{\mathcal C}^*)
    =e^{-i\phi_{\mathcal C}}.
    \label{eq:alg_transition_modulo}
\end{align}
Here $\phi_{\mathcal C}=\operatorname{Arg}B_{\mathcal C}$ is the
principal phase of the Bargmann invariant. Thus $E_{\mathcal C}$ equals
$-\phi_{\mathcal C}$ modulo $2\pi$.

To remove the $2\pi$ ambiguity, set
\begin{align}
    z_x(\mathbf{k})&:=\bra{u(a)}P(\mathbf{k})\ket{u(b)},
    &w_x(\mathbf{k})&:=\frac{z_x(\mathbf{k})}{z_x(b)}.
    \label{eq:alg_transition_ratio}
\end{align}
Since $P(b)\ket{u(b)}=\ket{u(b)}$, the denominator is
$z_x(b)=\braket{u(a)}{u(b)}$. For any point $\mathbf{k}$ on the right
edge, subtracting $1$ from the ratio gives
\begin{align}
    |w_x(\mathbf{k})-1|
    &=\frac{|z_x(\mathbf{k})-z_x(b)|}{|z_x(b)|}
    \notag\\
    &=\frac{\abs{\bra{u(a)}\qty[P(\mathbf{k})-P(b)]\ket{u(b)}}}
            {|\braket{u(a)}{u(b)}|}
    \notag\\
    &\leq\frac{\norm{P(\mathbf{k})-P(b)}}
                   {|\braket{u(a)}{u(b)}|}
    \notag\\
    &\leq\frac{\rho}{\sqrt{1-\rho^2}} < 1.
    \label{eq:alg_transition_disk}
\end{align}
The first inequality bounds a matrix element between normalized states
by the operator norm. For the second, $\mathbf{k}$ and $b$ lie in the
same cell, so \cref{eq:alg_cell_lipschitz_proof} bounds the numerator
by $\rho$, while \cref{eq:alg_overlap_bound} bounds the denominator
below by $\sqrt{1-\rho^2}$. The last inequality uses $\rho=1/32$.

This disk lies in the right half-plane. Since $w_x(b)=1$, continuity
fixes $\chi_x(\mathbf{k})-\chi_x(b)=\operatorname{Arg}w_x(\mathbf{k})$
as a real equality along the edge. Hence
\begin{align}
    |\chi_x(c)-\chi_x(b)|<\frac{\pi}{2}.
    \label{eq:alg_transition_x_bound}
\end{align}
The same argument for
$z_y(\mathbf{k}):=\bra{u(a)}P(\mathbf{k})\ket{u(d)}$, compared with its
value at $d$, gives $|\chi_y(d)-\chi_y(c)|<\pi/2$. Thus
$|E_{\mathcal C}|<\pi$. Together with
\cref{eq:alg_bargmann_branch_bound}, the equality modulo $2\pi$ in
\cref{eq:alg_transition_modulo} implies the real equality
\begin{align}
    E_{\mathcal C}=-\phi_{\mathcal C}.
    \label{eq:alg_transition_exact}
\end{align}
Substituting into \cref{eq:alg_transition_sum} proves
\cref{eq:alg_phase_sum}.
This establishes the geometric input for the local circuits analyzed next.

\subsection{Implementation guarantees for local primitives}
\label{app:alg_local_primitives}

We now justify the circuits used to realize the local primitives.
We first establish the control and inverse conventions, then prove the
reflection and fractional-power guarantees and allocate precision
across their nested implementations.

\subsubsection{Controlled and inverse circuits}
\label{app:alg_controls}

We track oracle time, global phase, and state error under circuit
control and inversion.
Let $V$ be a specified unitary circuit on all its data and work
registers, built from the specified Hamiltonian evolutions and
input-independent gates. Let $T(V):=\sum_j\abs{t_j}$ denote the total
Hamiltonian-oracle evolution time of this circuit, where $t_j$ is the
signed evolution time of its $j$th oracle call. Its controlled version is
\begin{align}
    \Lambda_c(V)
    :=\ket0\bra0_c\otimes I+\ket1\bra1_c\otimes V.
    \label{eq:alg_controlled_circuit}
\end{align}
For $V=G_M\cdots G_1$, controlling each gate gives
$\Lambda_c(V)=\Lambda_c(G_M)\cdots\Lambda_c(G_1)$.
Under our oracle model, a controlled Hamiltonian evolution of duration
$t$ has cost $|t|$. When controlling a circuit that already contains
controlled oracle calls, we must impose both the existing and the
additional control conditions. We compute their conjunction into a work
qubit, use it to control the oracle call, and then uncompute it.
These additional gates are input-independent and therefore add no
Hamiltonian-oracle time.
Reversing the circuit and changing the signs of the
evolution times implements $V^\dagger$. Therefore
\begin{align}
    T\qty(\Lambda_c(V))=T(V^\dagger)=T(V).
    \label{eq:alg_control_cost}
\end{align}

When a circuit is controlled, its global phase becomes a relative phase
of the control qubit and must therefore be accounted for explicitly.
To construct a controlled primitive, we use the same control qubit for
every gate in its implementation. If the implementation introduces a
known global phase relative to the ideal operation, we cancel it with
a compensating phase gate on the control qubit.

The same constructions also preserve the local state-error guarantees
on the relevant subspaces.
Suppose $\widetilde V$ approximates $V$ uniformly on an allowed
input subspace. Then
\begin{align}
    \Lambda_c(\widetilde V)-\Lambda_c(V)
    &=\ket1\bra1_c\otimes\qty(\widetilde V-V),
    \notag\\
    \norm{\qty(\widetilde V^\dagger-V^\dagger)V\ket\psi}
    &=\norm{\qty(\widetilde V-V)\ket\psi}.
    \label{eq:alg_control_inverse_errors}
\end{align}
The orthogonality of the control branches preserves the uniform
state-error bound, including for states entangled with other registers
on which the circuit does not act. The inverse bound holds on the ideal
image of the allowed input subspace. These statements apply to every
nested local circuit.

\subsubsection{Phase-step GQSP reflections}
\label{app:alg_reflection}

The reflection construction begins by separating the ground and excited
relative energies on the exact-reference sector.
For any vertex $v$ in the same cell as $a$, the variational principle implies
\begin{align}
    \abs{E_0(v)-E_0(a)}
    \leq\norm{H(v)-H(a)}\leq\frac{\Delta_{\min}}{32}.
    \label{eq:alg_energy_lipschitz}
\end{align}
On the sector with reference $\ket{u(a)}$, the eigenvectors of
$K_{v|a}$ are
$\ket{\chi_m}:=\ket{u(a)}\ket{u_m(v)}$, where $\ket{u_m(v)}$ has
energy $E_m(v)$ under $H(v)$ and $\ket{u_0(v)}:=\ket{u(v)}$.
Their relative energies
$\lambda_m:=E_m(v)-E_0(a)$ satisfy
\begin{align}
    \abs{\lambda_0}&\leq\frac{\Delta_{\min}}{32},
    \notag\\
    \lambda_m&\geq\frac{31\Delta_{\min}}{32}\qquad(m\ne0).
    \label{eq:alg_relative_windows}
\end{align}

The shifted signal unitary converts this energy separation into a phase
gap around zero.
Write the eigenvalues of the unitary
$U_{v|a}$ in \cref{eq:alg_reflection_signal} on the exact-reference
sector as $e^{-i\varphi_m}$, and define
\begin{align}
    \varphi_m&:=\tau_R\qty(\lambda_m-\frac{\Delta_{\min}}{2}),
    &\gamma_R&:=\frac{15}{32}\tau_R\Delta_{\min}.
    \label{eq:alg_reflection_phase_gap}
\end{align}
Here, $-\Delta_{\min}/32\leq\lambda_m\leq2H_{\max}$ for every $m$,
and $0<\Delta_{\min}\leq2H_{\max}$. 
Then $\varphi_0\leq-\gamma_R$ and $\varphi_m\geq\gamma_R$ for
$m\ne0$. Moreover,
\begin{align}
    \varphi_m&\geq-\frac{17}{32}\tau_R\Delta_{\min}
        >-\frac{\pi}{2},
    \notag\\
    \varphi_m&\leq\tau_R\qty(2H_{\max}-\frac{\Delta_{\min}}{2})
        <\frac{\pi}{2}.
    \label{eq:alg_reflection_phase_interval}
\end{align}
Thus the desired signs are $s_0:=+1$ and $s_m:=-1$ for $m\ne0$,
equivalently $s_m=-\operatorname{sgn}(\sin\varphi_m)$.

We next approximate the required reflection signs by a bounded polynomial.
For $0<\delta<1/2$, the bounded sign approximation of
Ref.~\cite[Lemma~14]{gilyen_qsvt_2019} supplies a real polynomial
$S_{d_R}$ of degree $d_R$ such that
\begin{align}
    \abs{S_{d_R}(x)}&\leq1\qquad(x\in[-1,1]),
    \notag\\
    \abs{S_{d_R}(x)-\operatorname{sgn}x}
    &\leq\delta\qquad(\sin\gamma_R\leq\abs{x}\leq1),
    \label{eq:alg_sign_polynomial}\\
    d_R&=O\qty(\frac{H_{\max}}{\Delta_{\min}}\log\frac1\delta).
    \label{eq:alg_reflection_degree}
\end{align}

To use this polynomial in GQSP, we express it as a Laurent polynomial
of the signal eigenvalue.
For $z:=e^{-i\varphi}$, define the Laurent polynomial
\begin{align}
    p_{d_R}(z)
    :=-S_{d_R}\qty(\frac{z^{-1}-z}{2i}).
    \label{eq:alg_reflection_laurent}
\end{align}
Since $(z^{-1}-z)/(2i)=\sin\varphi$, this polynomial is bounded
in magnitude by one on the entire unit circle and approximates
$s_m$ on the promised spectrum. GQSP implements it with one signal
ancilla and $O(d_R)$ queries to $U_{v|a}$ and its
inverse~\cite[Corollary~5 and Theorem~6]{motlagh_generalized_qsp_2024}.
The scalar shift in \cref{eq:alg_reflection_signal} is supplied by an
input-independent phase gate.

The full circuit must approximate the reflection together with a clean
signal ancilla.
Let $\widetilde{\mathcal R}_{v|a}$ be this full GQSP circuit, with
any known overall phase corrected. For
$z_m:=e^{-i\varphi_m}$ and $p_m:=p_{d_R}(z_m)$, its action is
\begin{align}
    \widetilde{\mathcal R}_{v|a}\ket0\ket{\chi_m}
    =\qty(p_m\ket0+q_m\ket1)\ket{\chi_m},
    \label{eq:alg_reflection_ancilla_action}
\end{align}
where $\abs{p_m}^2+\abs{q_m}^2=1$ by unitarity.
Here $p_m$ is real and $\abs{p_m-s_m}\leq\delta$. The squared
distance from the ideal output $s_m\ket0\ket{\chi_m}$ is
\begin{align}
    \abs{p_m-s_m}^2+\abs{q_m}^2
    &=2\qty(1-s_mp_m)
    \notag\\
    &=2\abs{s_m-p_m}\leq2\delta.
    \label{eq:alg_reflection_eigenstate_error}
\end{align}

Errors associated with distinct eigenvectors are orthogonal.
Consequently, choosing $\delta:=\epsilon^2/2$ gives, uniformly for
normalized $\ket\psi$ in the other system register,
\begin{align}
    &\norm{
        \widetilde{\mathcal R}_{v|a}\ket0\ket{u(a)}\ket\psi
        -\ket0\ket{u(a)}\mathcal R_v\ket\psi
    }
    \leq\epsilon.
    \label{eq:alg_reflection_error}
\end{align}
The same estimate holds with arbitrary control registers and with other
registers on which the circuit does not act. It includes the complementary
signal branch; no
postselection is performed.

By \cref{eq:alg_relative_evolution},
each signal query costs $2\tau_R$. Combining this duration with
\cref{eq:alg_reflection_degree} proves
\cref{eq:alg_reflection_cost}. The single signal ancilla is reusable.

\subsubsection{Direct rotations and fractional phase modulation}
\label{app:alg_qsp}

The direct-rotation construction relies on the spectrum of a product
of two ground-state reflections at vertices in the same cell.
For vertices $v,w$ in the same cell, define the principal angle between
their ground rays by
\begin{align}
    \theta_{vw}&:=\arcsin\norm{P(v)-P(w)}\leq\theta_*,
    &\theta_*&:=\arcsin\rho.
    \label{eq:alg_principal_angle}
\end{align}
For $\theta_{vw}>0$, choose the orthonormal basis
\begin{align}
    \ket{e_1}&:=\ket{u(v)},
    \notag\\
    \ket{e_2}
    &:=\frac{\ket{u(w)}^{(v)}-\cos\theta_{vw}\ket{e_1}}
        {\sin\theta_{vw}}.
    \label{eq:alg_rotation_basis}
\end{align}
With its phase chosen to give a real and strictly positive overlap
with $\ket{u(v)}$, the target state $\ket{u(w)}^{(v)}$ has
coordinates $(\cos\theta_{vw},\sin\theta_{vw})$ in this basis.
The rank-one projectors onto these two states have the matrix
representations
\begin{align}
    P(v)&=\begin{pmatrix}1&0\\0&0\end{pmatrix},
    \notag\\
    P(w)&=\begin{pmatrix}
        \cos^2\theta_{vw}&\cos\theta_{vw}\sin\theta_{vw}\\
        \cos\theta_{vw}\sin\theta_{vw}&\sin^2\theta_{vw}
    \end{pmatrix}.
    \label{eq:alg_rotation_projectors}
\end{align}
Using $\mathcal R_x=2P(x)-I$ for $x=v,w$, we obtain
\begin{align}
    \mathcal R_v&=\begin{pmatrix}1&0\\0&-1\end{pmatrix},
    \notag\\
    \mathcal R_w&=\begin{pmatrix}
        \cos2\theta_{vw}&\sin2\theta_{vw}\\
        \sin2\theta_{vw}&-\cos2\theta_{vw}
    \end{pmatrix}.
    \label{eq:alg_rotation_reflections}
\end{align}
Multiplying these matrices gives
\begin{align}
    \mathcal R_w\mathcal R_v
    =\begin{pmatrix}
        \cos2\theta_{vw}&-\sin2\theta_{vw}\\
        \sin2\theta_{vw}&\cos2\theta_{vw}
    \end{pmatrix}
    \label{eq:alg_rotation_matrix}
\end{align}
on this plane. On its orthogonal complement, both reflections act
as $-I$, so their product is the identity.

Its principal square root rotates by $\theta_{vw}$, proving
\cref{eq:alg_rotation_action}. Coincident rays give $D_{wv}=I$.
In all cases,
\begin{align}
    \operatorname{spec}\qty(\mathcal R_w\mathcal R_v)
    &\subseteq\qty{1,e^{2i\theta_{vw}},e^{-2i\theta_{vw}}},
    \label{eq:alg_rotation_spectrum}\\
    \operatorname{dist}\qty(-1,\operatorname{spec}\qty(\mathcal R_w\mathcal R_v))
    &\geq2\sqrt{1-\rho^2}.
    \label{eq:alg_rotation_branch_distance}
\end{align}

To implement the principal square root, we must match the branch
convention of the fractional-query construction to that of the direct rotation.
The fractional-query construction of
Ref.~\cite[Theorem~10]{motlagh_generalized_qsp_2024} approximates
the principal power on this fixed spectral arc with a number of
signal queries logarithmic in the inverse error. That construction
uses eigenphases in $[0,2\pi)$, so its branch cut lies at $+1$ on
the unit circle. Our principal power uses eigenphases in $(-\pi,\pi)$,
with the branch cut at $-1$. To account for this difference, we apply
the construction to $-U$, whose eigenphases lie near $\pi$ and are
separated from the cut at $+1$. The shift
$\theta\mapsto\theta+\pi$ introduces a phase $e^{i\pi\alpha}$,
which we cancel as follows:
\begin{align}
    U^\alpha
    =e^{-i\pi\alpha}(-U)_{[0,2\pi)}^\alpha.
    \label{eq:alg_fractional_branch_convention}
\end{align}
These phase adjustments are implemented coherently.
With the control convention of \cref{eq:alg_controlled_circuit},
each controlled query to $-U$ is implemented by a controlled
query to $U$ followed by a Pauli-$Z$ gate on the query control.
For a controlled implementation of $U^\alpha$, the prefactor
$e^{-i\pi\alpha}$ is implemented by the phase gate
$\operatorname{diag}\qty(1,e^{-i\pi\alpha})$ on the control
of the fractional-power circuit.
Both gates are input-independent and add no Hamiltonian-oracle time.

To evaluate the theorem's dependence on the minimum singular value,
let $A_U$ be the Hermitian phase generator defined by $-U=e^{iA_U}$
with eigenvalues in $(0,2\pi)$.
If the principal eigenphases of $U$ lie in $[-\theta_{\max},\theta_{\max}]$ with
$0\leq\theta_{\max}<\pi$, the eigenvalues of $A_U$ are $\pi+\theta$, and hence
\begin{align}
    \operatorname{spec}(A_U)
    &\subseteq\qty[\pi-\theta_{\max},\pi+\theta_{\max}].
    \notag
\end{align}
Since $A_U$ is positive definite and Hermitian, its minimum singular
value equals its smallest eigenvalue. Therefore,
\begin{align}
    \sigma_{\min}(A_U)&\geq\pi-\theta_{\max}.
    \label{eq:alg_fractional_generator_gap}
\end{align}
For $U=\mathcal R_w\mathcal R_v$, we may take
$\theta_{\max}=2\theta_*=2\arcsin\rho$ by
\cref{eq:alg_principal_angle,eq:alg_rotation_spectrum}.
The lower bound in \cref{eq:alg_fractional_generator_gap} is therefore
a positive constant, independently of $\Delta_{\min}$, and the spectrum
is separated from $2\pi$ by the same constant.
The theorem's inverse dependence on this lower bound contributes only
a constant to the query count: for approximation error $\zeta$, it
requires $O\qty(\log(1/\zeta))$ queries and $O(1)$ ancilla qubits.
The query bound is uniform for $\alpha\in(0,1)$.
Taking $\alpha:=1/2$ and
$U:=\mathcal R_w\mathcal R_v$ constructs $D_{wv}$.

The state-error convention includes all ancilla qubits used in the construction. 
If the construction is stated as a block approximation with error
$\zeta$, choosing $\zeta=O(\epsilon^2)$ makes the full state error
$O(\epsilon)$.
By unitarity, a block-approximation error of $\zeta$ bounds the norm
of the output component in which the ancilla qubits do not all return
to zero by $O(\sqrt{\zeta})$.
This choice preserves the logarithmic query count.

We next apply the same construction to the closed product around a cell.
Let $\varphi_{vw}$ denote the principal argument of
$\braket{u(v)}{u(w)}$ along an oriented edge $v\to w$.
\Cref{eq:alg_rotation_action} implies
\begin{align}
    D_{wv}\ket{u(v)}
    =e^{-i\varphi_{vw}}\ket{u(w)}.
    \label{eq:alg_edge_phase_action}
\end{align}
Multiplying the four edge phases gives
$e^{-i\phi_{\mathcal C}}$, proving \cref{eq:alg_loop_action}.

Taking a fractional power requires a branch gap for the full loop
operator, including its excited-state sector.
Also the spectrum of $D_{wv}$ is contained in
$\qty{1,e^{i\theta_{vw}},e^{-i\theta_{vw}}}$, with
$0\leq\theta_{vw}\leq\theta_*$, so
\begin{align}
    \norm{D_{wv}-I}\leq2\sin\qty(\frac{\theta_*}{2}).
    \label{eq:alg_edge_identity_bound}
\end{align}
The inequality $\norm{AB-I}\leq\norm{A-I}+\norm{B-I}$ for unitaries
then gives
\begin{align}
    \norm{L_{\mathcal C}-I}
    &\leq8\sin\qty(\frac{\theta_*}{2})
    \leq8\rho=\frac14,
    \notag\\
    \operatorname{dist}\qty(-1,\operatorname{spec}(L_{\mathcal C}))
    &\geq2-\norm{L_{\mathcal C}-I}\geq\frac74.
    \label{eq:alg_loop_spectrum_distance}
\end{align}
This controls the entire system spectrum, including the excited-state
sector.
For $U=L_{\mathcal C}$, the norm bound also implies that every principal
eigenphase satisfies $\abs{\theta}\leq2\arcsin(1/8)$.
Taking $\theta_{\max}=2\arcsin(1/8)$ in \cref{eq:alg_fractional_generator_gap}
gives $\sigma_{\min}(A_U)\geq\pi-2\arcsin(1/8)>0$.
Thus the dependence on the minimum singular value contributes only
a constant for the loop signal as well.
The same fractional-query construction therefore applies to
$L_{\mathcal C}$ uniformly for $\alpha\in(0,1)$.
By \cref{eq:alg_loop_action,eq:alg_bargmann_branch_bound}, the ground
state $\ket{u(a)}$ has principal eigenphase
$-\phi_{\mathcal C}\in(-\pi/4,\pi/4)$ under $L_{\mathcal C}$.
Taking the principal power $L_{\mathcal C}^{\alpha}$ multiplies this
eigenphase by $\alpha$, so $\ket{u(a)}$ acquires the phase
$e^{-i\alpha\phi_{\mathcal C}}$, as in \cref{eq:alg_fractional_action}.

During a complete phase modulation, the reference is fixed at $a$,
and every reflection target is a vertex of that cell. Thus
\cref{eq:alg_reflection_error} applies uniformly at every ideal
internal query.

\subsubsection{Nested precision and local costs}
\label{app:alg_nested_costs}

We allocate the target error across the nested GQSP layers to obtain
the costs of completed local operations.
For direct rotations, use state error at most $\epsilon/2$ for square-root GQSP with ideal
reflection-product queries. Its degree is
$d_{\mathrm{rot}}=O\qty(\log(1/\epsilon))$.
Each query uses two reflections. Assign each reflection state error
\begin{align}
    \epsilon_R:=\frac{\epsilon}{4d_{\mathrm{rot}}}.
    \label{eq:alg_rotation_internal_precision}
\end{align}
The ideal-trajectory bound proved in \cref{app:alg_error} then gives
\begin{align}
    \epsilon_{\mathrm{GQSP},D}
        +2d_{\mathrm{rot}}\epsilon_R
    &\leq\epsilon,
    \notag\\
    T_D(\epsilon)
    &=O\qty[
        d_{\mathrm{rot}}T_{\mathcal R}\qty(\epsilon_R)
    ]
    \notag\\
    &=O\qty(\frac{1}{\Delta_{\min}}\log^2\frac1\epsilon).
    \label{eq:alg_direct_rotation_cost}
\end{align}
The uniform reflection bound applies even when the other system
register is entangled with the outer GQSP signal.

For fractional phase modulation, the outer GQSP approximation has
error at most $\epsilon/2$ and degree
$d_{\mathrm{frac}}=O\qty(\log(1/\epsilon))$.
Each loop query uses four direct rotations, so choose
\begin{align}
    \epsilon_D:=\frac{\epsilon}{8d_{\mathrm{frac}}}.
    \label{eq:alg_fractional_internal_precision}
\end{align}
The corresponding bounds are
\begin{align}
    \epsilon_{\mathrm{GQSP},W}
        +4d_{\mathrm{frac}}\epsilon_D
    &\leq\epsilon,
    \notag\\
    T_W(\epsilon)
    &=O\qty[
        4d_{\mathrm{frac}}T_D\qty(\epsilon_D)
    ]
    \notag\\
    &=O\qty(\frac{1}{\Delta_{\min}}\log^3\frac1\epsilon).
    \label{eq:alg_fractional_cost}
\end{align}
These allocations include all work registers. Each inner workspace is
clean on the ideal trajectory and can be reused; errors in the actual
workspace are carried by the same trajectory bound.
\Cref{eq:alg_control_inverse_errors} covers controlled and inverse
queries at every layer.

The constants in
\cref{eq:alg_fractional_cost} are independent of $\alpha$, so the
cost is uniform in $m/q$. Together with
\cref{eq:alg_reflection_cost,eq:alg_direct_rotation_cost}, this
establishes the costs in \cref{tab:alg_primitive_costs}.

\subsection{Global correctness and resource analysis}
\label{app:alg_global_analysis}

We first verify the ideal scan and readout, then bound how local
implementation errors affect the output. The local precision and primitive
counts determine the runtime; a separate count gives the total number of qubits required.

\subsubsection{Closed-scan action and Fourier reconstruction}
\label{app:alg_scan_readout}

A closed scan must restore both ground-state copies while retaining
the cell phases on the readout register.
For a move from $a:=a_\ell$ to its successor
$v:=a_{(\ell+1)\bmod N_{\square}}$, let $\ket{u(v)}^{(a)}$ denote the
normalized target in \cref{eq:alg_rotation_action}.
Apply $D_{va}$ to copy 2 using copy 1 at $a$ as reference, and then
to copy 1 using copy 2 at $v$ as reference. Both references are in
the same cell as the reflection targets. The ideal action is
\begin{align}
    \ket{u(a)}_1\ket{u(a)}_2
    \longmapsto\ket{u(v)}_1^{(a)}\ket{u(v)}_2^{(a)}.
    \label{eq:alg_two_copy_transport}
\end{align}
Denote the move from anchor $a_\ell$ to its successor by $T_\ell$.
The raster path in \cref{eq:alg_raster} closes on the periodic torus,
so for some real phase $\gamma_{\mathrm{tr}}$,
\begin{align}
    T_{N_{\square}-1}\cdots T_0\ket{A_0}
    =e^{i\gamma_{\mathrm{tr}}}\ket{A_0},
    \label{eq:alg_transport_closure}
\end{align}
where $\ket{A_0}:=\ket{u(\mathbf{k}_0)}_1\ket{u(\mathbf{k}_0)}_2$
is the initial pair of ground-state copies.

At each cell, \cref{eq:alg_fractional_action} shows that the
controlled operation of $W_{\mathcal C}$ returns the same anchor pair on both control
branches and places the scalar $e^{-izm\phi_{\mathcal C}/q}$ on
branch $z\in\{0,1\}$. These phase factors are preserved by the subsequent unconditional transports.
By \cref{eq:alg_phase_sum}, the complete
scan therefore has the action
\begin{align}
    \mathcal S_{q,m}\qty(\ket z\ket{A_0})
    =e^{i\gamma_{\mathrm{tr}}}e^{2\pi izmC/q}\ket z\ket{A_0}.
    \label{eq:alg_closed_scan_action}
\end{align}
The restored anchor pair is independent of the readout value.
Hence either copy can serve as the reference in the next scan
without retaining any record of that value.

Applying the scans to all readout bits produces the Fourier encoding
of the Chern-number residue.
For a readout basis label $n:=\sum_j2^jz_j$, the $r_C$ scans
contribute phase $e^{2\pi inC/q}$ and the common factor
$e^{ir_C\gamma_{\mathrm{tr}}}$. This gives
\cref{eq:alg_fourier_state}. 
After the inverse QFT, the amplitude of $\ket{R'}$ is
\begin{align}
    \frac1q\sum_{n=0}^{q-1}e^{2\pi in(C-R')/q}
    =\begin{cases}
        1,&R'=C\bmod q,\\
        0,&\text{otherwise},
    \end{cases}
    \label{eq:alg_fourier_orthogonality}
\end{align}
for $0\leq R'<q$. Because $\abs{C}\leq C_{\max}<q/2$,
\cref{eq:alg_decode} uniquely recovers $C$ from the ideal residue.
There is no intrinsic readout failure in the ideal circuit.

\subsubsection{Error accumulation along the ideal trajectory}
\label{app:alg_error}

The global error bound uses local guarantees only along the ideal trajectory.
Let $U_j$ be the $j$th ideal primitive and $\widetilde U_j$ its
implemented unitary on all data and work registers.
Define the ideal and actual trajectories by
$\ket{\phi_j}:=U_j\ket{\phi_{j-1}}$ and
$\ket{\psi_j}:=\widetilde U_j\ket{\psi_{j-1}}$.
Assume the local state-error guarantee
\begin{align}
    \norm{\qty(\widetilde U_j-U_j)\ket{\phi_{j-1}}}\leq\epsilon_j.
    \label{eq:alg_ideal_local_error}
\end{align}
Adding and subtracting $\widetilde U_j\ket{\phi_{j-1}}$ gives
\begin{align}
    \ket{\psi_j}-\ket{\phi_j}
    &=\widetilde U_j\qty(\ket{\psi_{j-1}}-\ket{\phi_{j-1}})
    \notag\\
    &\quad+\qty(\widetilde U_j-U_j)\ket{\phi_{j-1}}.
    \label{eq:alg_trajectory_difference}
\end{align}
Unitarity and the triangle inequality imply
\begin{align}
    \norm{\ket{\psi_N}-\ket{\phi_N}}
    \leq\norm{\ket{\psi_0}-\ket{\phi_0}}+\sum_{j=1}^N\epsilon_j.
    \label{eq:alg_trajectory_stability}
\end{align}

Only ideal inputs appear in the local error assumption.
In particular, imperfections of a moved reference or a reused
ancilla are propagated in the first term of
\cref{eq:alg_trajectory_difference}; the actual imperfect reference
is never assumed to satisfy the exact-reference spectral promise.
The argument also applies to the internal sequences used in
\cref{app:alg_nested_costs}.

Each scan has $N_{\square}$ controlled fractional phase modulations
and $2N_{\square}$ direct rotations. Thus the number of
complete high-level primitives satisfies
\begin{align}
    N_{\mathrm{prim}}\leq3N_{\square}r_C.
    \label{eq:alg_primitive_count}
\end{align}
With exact initial preparation and the local precision in
\cref{eq:alg_local_precision}, the final state-vector error is at
most $\eta$. The ideal output, including clean work registers, is
\begin{align}
    \ket{\Phi_{\mathrm{out}}}
    :=e^{ir_C\gamma_{\mathrm{tr}}}
        \ket R\ket{A_0}\ket0_{\mathrm{work}}.
    \label{eq:alg_ideal_output}
\end{align}

To obtain the success probability, we convert the state-vector error
into a bound on measurement statistics.
For normalized pure states, the trace distance is
\begin{align}
    d_{\mathrm{tr}}\qty(\ket\psi,\ket\phi)
    &:=\frac12\norm{\ket\psi\bra\psi-\ket\phi\bra\phi}_1
    \notag\\
    &=\sqrt{1-\abs{\braket{\psi}{\phi}}^2}
    \leq\norm{\ket\psi-\ket\phi}.
    \label{eq:alg_trace_distance}
\end{align}
Any measurement-event probability differs by at most this distance.
Since the ideal circuit yields $R$ with certainty, the measured
residue $\widehat R$ and output
$\widehat C:=\operatorname{cent}_q(\widehat R)$ obey
\begin{align}
    \Pr\qty[\widehat C\ne C]
    =\Pr\qty[\widehat R\ne R]\leq\eta.
    \label{eq:alg_failure_bound}
\end{align}
The output in the early-return case is exact by \cref{app:alg_geometry}.
This proves the correctness part of
\cref{thm:nonadiabatic_algorithm}.

\subsubsection{Total runtime and register count}
\label{app:alg_resources}

We now translate the correctness requirements into resource bounds.
Combining \cref{eq:alg_direct_rotation_cost,eq:alg_fractional_cost}
with $N_{\square}r_C$ fractional phase modulations,
$2N_{\square}r_C$ direct rotations, and the local precision in
\cref{eq:alg_local_precision} yields the runtime bound in
\cref{eq:alg_resource_composition}. The early-return case uses no
oracle calls, as stated in \cref{eq:alg_runtime_zero}.

The reflection degree in \cref{eq:alg_reflection_degree} scales with
$H_{\max}/\Delta_{\min}$, while each signal query lasts
$\Theta(H_{\max}^{-1})$. Their product gives
$O\qty(\Delta_{\min}^{-1}\log(1/\epsilon))$.
Thus the norm-to-gap ratio affects the number of signal queries and
input-independent gates, but cancels from the Hamiltonian-oracle
time. The additional GQSP layers contribute the logarithmic factors
displayed in \cref{eq:alg_resource_composition}.

The data registers use $2n$ qubits. Coherent mesh addressing requires
$O\qty(\log N_{\square})$ qubits and the readout requires
$r_C=O\qty(\log\qty(2C_{\max}+1))$ qubits by \cref{eq:alg_modulus}.
Each of the reflection, square-root, and fractional-power GQSP layers
has constant-size GQSP ancilla overhead. There are only three nested
layers, so their simultaneous signal registers and the work bits
for combining control conditions total $O(1)$ qubits and can be
reused between calls. This yields
\begin{align}
    N_{\mathrm{qubit}}
    &=2n+O\qty[\log N_{\square}+\log\qty(2C_{\max}+1)]+O(1)
    \notag\\
    &=2n+O\qty[\log\qty(1+\frac{L_xL_y}{\Delta_{\min}^2})],
    \label{eq:alg_register_count}
\end{align}
in the quantum case, proving \cref{eq:alg_qubits}.
Reducing the allowed state error $\epsilon$ increases the GQSP
sequence lengths without requiring additional ancilla qubits. 
Gate descriptions and classical
arithmetic are not quantum workspace in this oracle model.
The early-return case uses no quantum registers.

For completeness, finite state-preparation and gate-synthesis errors
can be budgeted separately. Replace the local precision by
$\eta/(9N_{\square}r_C)$, allow total two-copy preparation error at
most $\eta/3$, and allow accumulated state error at most $\eta/3$
from input-independent gate synthesis, including the inverse QFT.
\Cref{eq:alg_trajectory_stability,eq:alg_trace_distance} then still
give failure probability at most $\eta$.
The constant change of local precision preserves \cref{eq:alg_runtime}.
State-preparation costs are not included in the oracle time,
and any additional workspace needed for state preparation or gate
implementation is not included in \cref{eq:alg_register_count}.
Together with the correctness analysis, these bounds establish
\cref{thm:nonadiabatic_algorithm}.

\section{Proof Details for the Chern-Number Lower Bound}
\label{app:chern_lower_bound_proof}

We verify the three hard families of \cref{sec:lower_bound_hard_instance}
and prove the common adversary estimate used in
\cref{thm:chern_lower_bound,cor:binary_chern_lower_bound}.
The early-return branch is already established by the zero-Chern
certificate in \cref{app:alg_geometry}.
Throughout this appendix, $a_\mu=L_\mu/\Delta_{\min}$ and the
supplied bounds are held fixed across all inputs in each hard family.

\subsection{Smooth interpolation}
\label{app:lower_bound_interpolation}

The tile and slab constructions use an interpolation that is flat to
all orders at its endpoints while keeping its maximum slope close to
one.
Begin with the smooth step
\begin{align}
    \xi(u)&:=
    \begin{cases}
        0,&u\leq0,\\
        e^{-1/u},&u>0,
    \end{cases}
    \notag\\
    S(u)&:=\frac{\xi(u)}{\xi(u)+\xi(1-u)}.
    \label{eq:app_lower_bound_smooth_step}
\end{align}
For every integer $m\geq0$, the $m$th derivative of $\xi$ with
respect to $u$ tends to zero as $u\to0^+$.
The denominator of $S$ is strictly positive.
Thus $S$ is smooth, equals zero for $u\leq0$, equals one for
$u\geq1$, and satisfies $S(u)+S(1-u)=1$.

For $0<\varepsilon<1/2$, define
\begin{align}
    b_\varepsilon(u)
    &:=S\qty(\frac{u}{\varepsilon})
       S\qty(\frac{1-u}{\varepsilon}),
    \notag\\
    B_\varepsilon(u)
    &:=\frac1{1-\varepsilon}\int_0^u b_\varepsilon(v)\,dv,
    \qquad 0\leq u\leq1.
    \label{eq:app_lower_bound_interpolation}
\end{align}
Since $\varepsilon<1/2$ and $S(u)=1$ for $u\geq1$, the window
function $b_\varepsilon(u)$ reduces to
\begin{align}
    b_\varepsilon(u)=
    \begin{cases}
        S\qty(\frac{u}{\varepsilon}),&0\leq u\leq\varepsilon,\\
        1,&\varepsilon\leq u\leq1-\varepsilon,\\
        S\qty(\frac{1-u}{\varepsilon}),&1-\varepsilon\leq u\leq1.
    \end{cases}
    \label{eq:app_lower_bound_window_pieces}
\end{align}
Thus it consists of two smooth ramps and a central interval on which
it equals one.

To compute its integral, first note that the symmetry
$S(u)+S(1-u)=1$ implies
\begin{align}
    2\int_0^1S(v)\,dv
    &=\int_0^1\qty[S(v)+S(1-v)]\,dv=1.
    \label{eq:app_lower_bound_step_integral}
\end{align}
Splitting the integral of $b_\varepsilon$ over the three intervals
and rescaling each endpoint interval to $[0,1]$ therefore gives
\begin{align}
    \int_0^1b_\varepsilon(v)\,dv
    &=\int_0^\varepsilon S\qty(\frac{v}{\varepsilon})\,dv
      +1-2\varepsilon
    \notag\\
    &\quad+\int_{1-\varepsilon}^1
        S\qty(\frac{1-v}{\varepsilon})\,dv
    \notag\\
    &=2\varepsilon\int_0^1S(v)\,dv+1-2\varepsilon
    \notag\\
    &=1-\varepsilon.
    \label{eq:app_lower_bound_interpolation_integral}
\end{align}
Consequently, $B_\varepsilon(0)=0$, $B_\varepsilon(1)=1$, and
$0\leq B_\varepsilon'\leq1/(1-\varepsilon)$.
All derivatives of $b_\varepsilon$ vanish at both endpoints, so
every positive-order derivative of $B_\varepsilon$ does as well.
This proves \cref{eq:lower_bound_interpolation_properties}.

The derivative budget determines a positive interpolation parameter in
every direction whose ratio exceeds $1/2$.
Indeed, for $a_\mu>1/2$, the choice in
\cref{eq:lower_bound_tile_counts} obeys
\begin{align}
    \frac{a_\mu}{2}\leq n_\mu<2a_\mu.
    \label{eq:app_lower_bound_partition_counts}
\end{align}
For $1/2<a_\mu<1$, this follows from $n_\mu=1$; for
$a_\mu\geq1$, use $\lfloor a_\mu\rfloor\geq a_\mu/2$.
If $J$ is the set of interpolated directions, choose, for example,
\begin{align}
    \varepsilon
    :=\min\qty{\frac1{10},\,
        \min_{\mu\in J}\frac12\qty(1-\frac{n_\mu}{2a_\mu})}.
    \label{eq:app_lower_bound_smoothing_choice}
\end{align}
Then $0<\varepsilon<1/2$ and
$n_\mu/[2(1-\varepsilon)]\leq a_\mu$ for every $\mu\in J$.
Use $J:=\qty{x,y}$ for tiles and only the direction with
$a_\mu>1/2$ for slabs.

\subsection{Tile construction}
\label{app:lower_bound_tiles}

For $a_x,a_y>1/2$, partition the torus into the half-open rectangles
\begin{align}
    I_{\mu,i}&:=[iw_\mu,(i+1)w_\mu),
    \notag\\
    \mathcal R_{ij}&:=I_{x,i}\times I_{y,j},
    \label{eq:app_lower_bound_tile_regions}
\end{align}
with $0\leq i<n_x$ and $0\leq j<n_y$.
On each rectangle, use the texture in
\cref{eq:lower_bound_polar_excursion,eq:lower_bound_tile_azimuth,eq:lower_bound_spherical_texture,eq:lower_bound_tile_texture}.
The half-open convention assigns each oracle-grid point to exactly
one region; the smooth texture extends across all boundaries.
The construction does not require these boundaries to coincide with
the oracle grid.

Smoothness follows from the flat endpoint behavior of the angular
interpolation.
At $s=0,1$, the vector equals the north pole, while at $s=1/2$
it equals the south pole.
All partial derivatives of positive total order of $\mathbf n_z$
with respect to $s$ and $t$ vanish at these values of $s$, so the
azimuth can switch at the south pole without creating a singularity.
At $t=0,1$, the azimuth is an integer multiple of $2\pi$, and every
positive-order partial derivative of $\mathbf n_z$ with respect to $t$
vanishes.
The textures for both bit values therefore match,
together with all mixed derivatives,
the common meridian $\mathbf m(\vartheta(s),0)$ at $t=0,1$.
These properties also ensure periodic matching across the torus seams.
At $k_x=0$ the texture is the north pole for every $k_y$, in
particular at the reference point $(0,0)$.

The angular derivatives directly control the Hamiltonian derivatives.
For the spherical texture of \cref{eq:lower_bound_spherical_texture},
\begin{align}
    \abs{\partial_\theta\mathbf m}&=1,
    &\abs{\partial_\phi\mathbf m}&=\abs{\sin\theta}.
    \label{eq:app_lower_bound_spherical_derivatives}
\end{align}
Away from the flat seams, the polar angle depends only on $k_x$
and the azimuth depends only on $k_y$ within each half-tile.
The bounds in \cref{eq:lower_bound_interpolation_properties} give
\begin{align}
    \abs{\partial_{k_x}\vartheta}
    &\leq\frac{2\pi}{w_x(1-\varepsilon)}
    =\frac{n_x}{1-\varepsilon},
    \notag\\
    \abs{\partial_{k_y}\varphi_{z_{ij}}}
    &\leq\frac{2\pi}{w_y(1-\varepsilon)}
    =\frac{n_y}{1-\varepsilon}.
    \label{eq:app_lower_bound_angular_derivatives}
\end{align}
Since $\norm{\mathbf a\cdot\bm\sigma}=\abs{\mathbf a}$ for
real $\mathbf a$, it follows that
\begin{align}
    \norm{\partial_{k_\mu}H_z}
    \leq\frac{\Delta_{\min}n_\mu}{2(1-\varepsilon)}
    \leq L_\mu.
    \label{eq:app_lower_bound_tile_derivatives}
\end{align}
Smoothness extends the same bounds to the seams, proving
\cref{eq:lower_bound_tile_derivatives}.

The ground-state Chern number is the degree of the texture with the
Berry-connection convention used here~\cite{qi_topological_2006}.
For any of the two-band Hamiltonians in
\cref{eq:lower_bound_hard_hamiltonian}, Pauli algebra gives
\begin{align}
    P_-&=\frac{I-\mathbf n_z\cdot\bm\sigma}{2},
    \label{eq:app_lower_bound_ground_projector}\\
    \Omega
    &=i\operatorname{Tr}\qty(P_-
        \qty[\partial_{k_x}P_-,\partial_{k_y}P_-])
    \notag\\
    &=\frac12\mathbf n_z\cdot\qty(
        \partial_{k_x}\mathbf n_z\times\partial_{k_y}\mathbf n_z).
    \label{eq:app_lower_bound_berry_curvature}
\end{align}
An inactive tile and the return half of every tile have no $k_y$
dependence and hence contribute zero curvature.
On the first half of an active tile,
\begin{align}
    \mathbf n_z\cdot\qty(
        \partial_{k_x}\mathbf n_z\times\partial_{k_y}\mathbf n_z)
    =\sin\vartheta\,
      (\partial_{k_x}\vartheta)(\partial_{k_y}\varphi).
    \label{eq:app_lower_bound_tile_density}
\end{align}
The polar angle runs from $0$ to $\pi$ and the azimuth from $0$
to $2\pi$, so the tile contributes
\begin{align}
    C_{ij}
    &=\frac1{4\pi}\int_0^\pi\sin\vartheta\,d\vartheta
                    \int_0^{2\pi}d\varphi
    =1
    \qquad(z_{ij}=1).
    \label{eq:app_lower_bound_tile_chern}
\end{align}
Adding the curvature integrals over the tiles yields
\begin{align}
    C(H_z)&=\sum_{i,j}z_{ij},
    \label{eq:app_lower_bound_chern_sum}\\
    M=n_xn_y&\geq\frac{a_xa_y}{4},
    \label{eq:app_lower_bound_tile_count}
\end{align}
where the count follows from \cref{eq:app_lower_bound_partition_counts}.
Changing one bit changes the texture only inside its tile; on shared
boundaries the two textures agree.
This proves all the required properties of the interior family.

For a negative defect, \cref{eq:lower_bound_signed_tile_azimuth}
reflects the active texture as
$\mathbf n_-=(n_{+,x},-n_{+,y},n_{+,z})$.
This orthogonal reflection preserves every derivative norm and reverses
the scalar triple product in \cref{eq:app_lower_bound_berry_curvature}.
Thus the Berry curvature changes sign and the region contributes $-1$.
Endpoint flatness and the common boundary meridian hold for azimuths
ending at either $+2\pi$ or $-2\pi$, so arbitrary signed and inactive
tiles join smoothly.
The unit texture still has gap $\Delta_{\min}$, norm
$\Delta_{\min}/2$, and the north-pole reference texture, hence the
same reference ground state $\ket1$.

\subsection{Slab construction}
\label{app:lower_bound_slabs}

Suppose $a_y=1/2$ and $a_x>1/2$.
Take the regions $\mathcal R_i:=I_{x,i}\times S^1$, where $S^1$
denotes the $2\pi$-periodic $k_y$ circle.
Use \cref{eq:lower_bound_slab_azimuth} with the same polar excursions
as in the tile construction.
The $k_x$ seams are smooth by endpoint flatness,
and the $k_y$ dependence
is smooth and periodic because it enters through $\cos(z_i k_y)$
and $\sin(z_i k_y)$.
On each first half-slab, $\abs{\partial_{k_y}\varphi}=z_i\leq1$,
while the return half has no $k_y$ dependence.
Thus \cref{eq:app_lower_bound_spherical_derivatives} implies
\begin{align}
    \norm{\partial_{k_y}H_z}
    \leq\frac{\Delta_{\min}}2\abs{\sin\vartheta}
    \leq\frac{\Delta_{\min}}2=L_y.
    \label{eq:app_lower_bound_slab_equality_direction}
\end{align}
The $k_x$ derivative obeys \cref{eq:app_lower_bound_tile_derivatives}
with the interpolation parameter chosen for $k_x$ alone.
The reference texture at $(0,0)$ is again the north pole.

Each active slab contributes one to the Chern number.
The first half has exactly the angular ranges used in
\cref{eq:app_lower_bound_tile_chern}, now with $\varphi=k_y$,
and the return half contributes zero.
Consequently,
\begin{align}
    C(H_z)&=\sum_i z_i,
    &M=n_x&\geq\frac{a_x}{2}=a_xa_y.
    \label{eq:app_lower_bound_slab_count}
\end{align}
Different bits have disjoint slab supports, and their textures agree
on every shared boundary.

For $a_x=1/2$ and $a_y>1/2$, use polar excursions along $k_y$
and set the azimuth on each first half-slab to $-z_j k_x$.
The coordinate exchange and azimuth reversal compensate in the
curvature orientation:
\begin{align}
    \mathbf n_z\cdot\qty(
        \partial_{k_x}\mathbf n_z\times\partial_{k_y}\mathbf n_z)
    &=-\sin\vartheta\,
        (\partial_{k_y}\vartheta)(\partial_{k_x}\varphi)
    \notag\\
    &=z_j\sin\vartheta\,\partial_{k_y}\vartheta.
    \label{eq:app_lower_bound_swapped_slab_orientation}
\end{align}
The contribution of each active slab is therefore still $+1$.
Interchanging the derivative estimates gives
$\norm{\partial_{k_x}H_z}\leq\Delta_{\min}/2=L_x$ and
$\norm{\partial_{k_y}H_z}\leq L_y$, with
$M=n_y\geq a_xa_y$ and the same reference ground state.

For a negative defect, take $\varphi:=-k_y$ when $a_y=1/2<a_x$,
and $\varphi:=+k_x$ when $a_x=1/2<a_y$, on the first half-slab;
retain $\varphi:=0$ on the return half.
In the swapped case, \cref{eq:app_lower_bound_swapped_slab_orientation}
gives
\begin{align}
    \mathbf n_-\cdot\qty(
        \partial_{k_x}\mathbf n_-\times\partial_{k_y}\mathbf n_-)
    =-\sin\vartheta\,\partial_{k_y}\vartheta.
    \label{eq:app_lower_bound_signed_slab_orientation}
\end{align}
The negative defect therefore contributes $-1$ in either slab construction.
As in the tile construction, the sign change reflects one texture component and
preserves the derivative bounds, smooth periodic matching, gap, norm,
and common reference ground state.

\subsection{Flattened QWZ construction at the corner}
\label{app:lower_bound_corner}

When $a_x=a_y=1/2$, consider the pair in
\cref{eq:lower_bound_corner_hamiltonian_0,eq:lower_bound_corner_hamiltonian_1}.
The normalization of its Bloch vector is nonsingular because
\begin{align}
    \abs{\mathbf d(\mathbf k)}^2
    =1+2(1+\cos k_x)(1+\cos k_y)\geq1.
    \label{eq:app_lower_bound_corner_norm}
\end{align}
Writing $\widehat{\mathbf d}:=\mathbf d/\abs{\mathbf d}$ gives
\begin{align}
    \partial_{k_\mu}\widehat{\mathbf d}
    &=\frac{\qty(I-\widehat{\mathbf d}\widehat{\mathbf d}^{\mathsf T})
                 \partial_{k_\mu}\mathbf d}{\abs{\mathbf d}},
    \notag\\
    \abs{\partial_{k_\mu}\widehat{\mathbf d}}
    &\leq\frac{\abs{\partial_{k_\mu}\mathbf d}}{\abs{\mathbf d}}
    \leq1,
    \label{eq:app_lower_bound_corner_derivatives}
\end{align}
since $\abs{\partial_{k_x}\mathbf d}
=\abs{\partial_{k_y}\mathbf d}=1$.
Both inputs therefore satisfy the derivative bounds
$\norm{\partial_{k_\mu}H_z}\leq\Delta_{\min}/2=L_\mu$.
Their eigenvalues are $\pm\Delta_{\min}/2$ everywhere, and
$\widehat{\mathbf d}(0,0)=\mathbf e_z$, so their gaps, norms,
and reference ground states coincide.

The Hamiltonian $H_1$ is obtained from
$H_{\rm QWZ}(\mathbf k;1)$ by multiplication by a positive scalar function.
This spectral flattening changes the eigenvalues but leaves the lower-band
projector unchanged.  The QWZ calibration in
\cref{app:qwz_calibration} therefore gives $C(H_1)=1$.
The constant projector of $H_0$ gives $C(H_0)=0$.
With $M=1$ and $\mathcal R_1:=X$, this pair satisfies
$C(H_z)=z_1$ and $M\geq a_xa_y/4$.
The pointwise difference between the two Hamiltonians has norm at
most $\Delta_{\min}$, as required for the common adversary argument.

\subsection{Common adversary estimate}
\label{app:chern_lower_bound_adversary}

We adapt the continuous-time adversary argument of
Refs.~\cite{farhi_quantum_2008,yonge-mallo_adversary_2011}
to the present spatially supported oracle family.
For all three constructions, the reference ground state at $(0,0)$
is always $\ket1$.
The preparation circuit, its inverse, and its controlled version
are the same for every input.
Thus reference-state preparation provides no input-dependent information.
Input-independent operations can be included in the driver of
\cref{eq:hamiltonian_oracle_evolution}.
Ancillas may purify any randomness or deferred measurements, so it
suffices to analyze pure algorithm states $\ket{\psi_z(t)}$ and
the progress function in \cref{eq:chern_progress_function}, with
$W(0)=0$.

Successful estimation imposes a uniform separation between the final
states of neighboring inputs.
For each $z$, let
\begin{align}
    I_z:=\qty(C(H_z)-\frac12,C(H_z)+\frac12).
    \label{eq:app_lower_bound_success_interval}
\end{align}
The intervals $I_z$ and $I_{z^{(j)}}$ are disjoint because their
centers differ by one.
If $E_z$ is the measurement effect for outputs in $I_z$ and
$\rho_z(T):=\ket{\psi_z(T)}\!\bra{\psi_z(T)}$, then
\begin{align}
    \operatorname{Tr}\qty(E_z\rho_z(T))&\geq1-\eta,
    \notag\\
    \operatorname{Tr}\qty(E_z\rho_{z^{(j)}}(T))&\leq\eta.
    \label{eq:app_lower_bound_measurement_probabilities}
\end{align}
The trace distance of these states is at least $1-2\eta$.
The pure-state trace-distance formula consequently gives
\begin{align}
    \abs{\braket*{\psi_z(T)}{\psi_{z^{(j)}}(T)}}
    \leq\beta_\eta:=2\sqrt{\eta(1-\eta)}<1.
    \label{eq:app_lower_bound_final_overlap}
\end{align}
This bounds the squared distance used by the progress function:
\begin{align}
    \norm{\ket{\psi_z(T)}-\ket{\psi_{z^{(j)}}(T)}}^2
    &=2-2\operatorname{Re}
        \braket*{\psi_z(T)}{\psi_{z^{(j)}}(T)}
    \notag\\
    &\geq2(1-\beta_\eta).
    \label{eq:app_lower_bound_neighbor_distance}
\end{align}
Substituting into \cref{eq:chern_progress_function} yields
\begin{align}
    W(T)
    &\geq\frac1{2^M}\sum_z\sum_{j=1}^M2(1-\beta_\eta)
    \notag\\
    &=2\qty[1-2\sqrt{\eta(1-\eta)}]M,
    \label{eq:app_lower_bound_required_progress_derivation}
\end{align}
which is \cref{eq:lower_bound_required_progress}.

To bound the growth rate, assign an address projector to each bit region
using the mesh vertices in \cref{eq:alg_mesh_points}:
\begin{align}
    \Pi_j:=I_{\mathrm c}\otimes
        \sum_{\mathbf k_{j_x,j_y}\in\mathcal R_j}
        \ket{j_x,j_y}\!\bra{j_x,j_y}\otimes I.
    \label{eq:app_lower_bound_region_projector}
\end{align}
Here $I_{\mathrm c}$ acts on the control qubit and the last
identity includes the data and workspace registers.
Use the tile or slab regions, with their half-open boundary convention,
and use $\Pi_1=I$ for the corner pair.
In every case,
\begin{align}
    \Pi_j\Pi_\ell&=0\quad(j\ne\ell),
    &\sum_{j=1}^M\Pi_j&\leq I.
    \label{eq:app_lower_bound_orthogonal_projectors}
\end{align}
Flipping bit $j$ changes the Hamiltonian only in $\mathcal R_j$.
Because the textures agree on shared boundaries, this support property
holds irrespective of how the construction partition meets the oracle grid.

The norm bound needed for each neighboring pair follows from the unit
length of the two textures:
\begin{align}
    \norm{H_z(\mathbf k)-H_{z^{(j)}}(\mathbf k)}
    &=\frac{\Delta_{\min}}2
        \abs{\mathbf n_z(\mathbf k)-\mathbf n_{z^{(j)}}(\mathbf k)}
    \notag\\
    &\leq\Delta_{\min}.
    \label{eq:app_lower_bound_pointwise_difference}
\end{align}
The sampled oracle is block diagonal in its address, and the control
factor $\ket{1}\!\bra{1}_{\mathrm c}$ in
\cref{eq:controlled_hamiltonian_oracle} has norm one.
Thus, for $\delta\widetilde{\mathcal O}_{z,j}
:=\widetilde{\mathcal O}_{H_z}-\widetilde{\mathcal O}_{H_{z^{(j)}}}$,
\begin{align}
    \delta\widetilde{\mathcal O}_{z,j}
    &=\Pi_j\delta\widetilde{\mathcal O}_{z,j}\Pi_j,
    \label{eq:app_lower_bound_oracle_support}\\
    \norm{\delta\widetilde{\mathcal O}_{z,j}}&\leq\Delta_{\min}.
    \label{eq:app_lower_bound_oracle_difference}
\end{align}
The input-independent block added in
\cref{eq:lower_bound_multiband_embedding} does not affect these estimates.

For a neighboring pair, write
\begin{align}
    d_{z,j}(t)
    :=\norm{\ket{\psi_z(t)}-\ket{\psi_{z^{(j)}}(t)}}^2.
    \label{eq:app_lower_bound_neighbor_distance_definition}
\end{align}
Substituting \cref{eq:hamiltonian_oracle_evolution} cancels the
input-independent driver and gives
\begin{align}
    \frac{d}{dt}d_{z,j}(t)
    =2g(t)\operatorname{Im}
        \mel{\psi_z(t)}{\delta\widetilde{\mathcal O}_{z,j}}{\psi_{z^{(j)}}(t)}.
    \label{eq:app_lower_bound_distance_derivative}
\end{align}
Defining the probability weight in region $j$ by
\begin{align}
    p_{z,j}(t):=\norm{\Pi_j\ket{\psi_z(t)}}^2,
    \label{eq:app_lower_bound_region_probability}
\end{align}
the support and norm bounds imply
\begin{align}
    \abs{\frac{d}{dt}d_{z,j}(t)}
    \leq2\Delta_{\min}\abs{g(t)}
        \sqrt{p_{z,j}(t)p_{z^{(j)},j}(t)}.
    \label{eq:app_lower_bound_neighbor_rate}
\end{align}
This bound applies to controlled queries and to either sign of $g(t)$,
so forward and inverse evolution have the same cost in the estimate.

Averaging over inputs and summing over the orthogonal regions bounds
the total rate independently of the number of bits:
\begin{align}
    \abs{\frac{dW}{dt}}
    &\leq\frac{2\Delta_{\min}\abs{g(t)}}{2^M}
        \sum_z\sum_{j=1}^M\sqrt{p_{z,j}p_{z^{(j)},j}}
    \notag\\
    &\leq\frac{2\Delta_{\min}\abs{g(t)}}{2^M}
        \sum_z\sum_{j=1}^M p_{z,j}
    \notag\\
    &\leq2\Delta_{\min}\abs{g(t)}.
    \label{eq:app_lower_bound_progress_rate}
\end{align}
The second line uses the arithmetic--geometric mean inequality and
$\sum_zp_{z^{(j)},j}=\sum_zp_{z,j}$, since flipping one bit
permutes all inputs.
The last line uses $\sum_jp_{z,j}\leq1$, which follows from
\cref{eq:app_lower_bound_orthogonal_projectors}.
Input-independent gates preserve all pairwise state distances, so
inserting exact gates, including register swaps and reference-state
preparation, between oracle evolutions does not change the rate estimate.

Finally, integrate \cref{eq:app_lower_bound_progress_rate} using
$W(0)=0$ and the oracle-time definition in
\cref{eq:oracle_time_definition}:
\begin{align}
    W(T)\leq2\Delta_{\min}\int_0^Tdt\,\abs{g(t)}
    =2\Delta_{\min}T_{\mathrm{oracle}}.
    \label{eq:app_lower_bound_integrated_rate}
\end{align}
Combining this with
\cref{eq:app_lower_bound_required_progress_derivation} and
$M\geq a_xa_y/4$ gives
\begin{align}
    T_{\mathrm{oracle}}
    &\geq\frac{1-2\sqrt{\eta(1-\eta)}}{\Delta_{\min}}M
    \notag\\
    &\geq c_\eta\frac{L_xL_y}{\Delta_{\min}^3}.
    \label{eq:app_lower_bound_runtime}
\end{align}
This proves \cref{eq:chern_lower_bound_theorem_result} for the
tile, slab, and corner families with the same constant $c_\eta$.

\subsubsection{\texorpdfstring{Binary-Chern promise}{Binary-Chern promise}}
\label{app:lower_bound_binary_promise}

For the binary-Chern tile and slab
families of
\cref{sec:lower_bound_binary_family}, use the two central Hamming layers
\begin{align}
    X_m&:=\qty{z\in\qty{0,1}^{M_{\mathrm{var}}}:|z|=m},
    \notag\\
    Y_{m+1}&:=\qty{z\in\qty{0,1}^{M_{\mathrm{var}}}:|z|=m+1},
    \label{eq:app_lower_bound_binary_layers}\\
    R&:=\qty{(z,z+e_j):z\in X_m,\ z_j=0},
    \label{eq:app_lower_bound_binary_relation}
\end{align}
where $e_j$ is the standard basis vector with a one only in coordinate
$j$, and $M_{\mathrm{var}}=2m+1$.
Each $z\in X_m$ has $M_{\mathrm{var}}-m=m+1$ zero coordinates, each defining one
neighbor in $Y_{m+1}$.
Each $w\in Y_{m+1}$ has $m+1$ one coordinates, each defining one
neighbor $w-e_j\in X_m$.
Thus $R$ is a bipartite relation regular of degree $m+1$ on both sides,
and
\begin{align}
    K:=|X_m|=\binom{M_{\mathrm{var}}}{m}
    &=\binom{M_{\mathrm{var}}}{m+1}=|Y_{m+1}|,
    \notag\\
    |R|&=(m+1)K.
    \label{eq:app_lower_bound_binary_edge_count}
\end{align}

As in \cref{eq:chern_progress_function}, average over the $2K$ allowed
inputs and sum over their neighbors.
Each edge of $R$ is then counted twice, giving
\begin{align}
    W_{\mathrm{bin}}(t):=\frac1K\sum_{(z,w)\in R}
        \norm{\ket{\psi_z(t)}-\ket{\psi_w(t)}}^2.
    \label{eq:app_lower_bound_binary_progress}
\end{align}
The initial state is common to all allowed inputs, so
$W_{\mathrm{bin}}(0)=0$.
For every $(z,w)\in R$, \cref{eq:lower_bound_binary_chern_values}
gives $C(H_z)=0$ and $C(H_w)=1$.
Their success intervals are disjoint, so the measurement-effect and
pure-state trace-distance argument in
\cref{eq:app_lower_bound_measurement_probabilities,eq:app_lower_bound_final_overlap}
gives $|\braket*{\psi_z(T)}{\psi_w(T)}|\leq\beta_\eta$.
Consequently, since $|R|=(m+1)K$,
\begin{align}
    W_{\mathrm{bin}}(T)
    &\geq\frac1K\sum_{(z,w)\in R}2(1-\beta_\eta)
    \notag\\
    &=2\qty[1-2\sqrt{\eta(1-\eta)}](m+1).
    \label{eq:app_lower_bound_binary_required_progress}
\end{align}

Index only the variable regions by $j=1,\ldots,M_{\mathrm{var}}$ and use their
projectors $\Pi_j$ from \cref{eq:app_lower_bound_region_projector}.
They are mutually orthogonal and satisfy $\sum_{j=1}^{M_{\mathrm{var}}}\Pi_j\leq I$;
the fixed negative and inactive regions need not be included.
For a pair $(z,w)\in R$ with $w=z+e_j$,
the two Hamiltonians agree outside variable
region $j$, including all fixed negative defects.
Thus the same support and norm argument gives
\begin{align}
    \delta\widetilde{\mathcal O}_{z,w}
    &:=\widetilde{\mathcal O}_{H_z}-\widetilde{\mathcal O}_{H_w}
      =\Pi_j\delta\widetilde{\mathcal O}_{z,w}\Pi_j,
    \notag\\
    \norm{\delta\widetilde{\mathcal O}_{z,w}}&\leq\Delta_{\min}.
    \label{eq:app_lower_bound_binary_oracle_difference}
\end{align}
With $d_{z,w}:=\norm{\ket{\psi_z}-\ket{\psi_w}}^2$ and
$p_{v,j}:=\norm{\Pi_j\ket{\psi_v}}^2$ for either endpoint $v$,
driver cancellation and Cauchy--Schwarz yield
\begin{align}
    \abs{\frac{d}{dt}d_{z,w}}
    &=2\abs{g(t)}\abs{\operatorname{Im}
        \mel{\psi_z}{\delta\widetilde{\mathcal O}_{z,w}}{\psi_w}}
    \notag\\
    &\leq2\Delta_{\min}\abs{g(t)}\sqrt{p_{z,j}p_{w,j}}.
    \label{eq:app_lower_bound_binary_edge_rate}
\end{align}

Write $j(z,w)$ for the unique coordinate changed on edge $(z,w)$.
A sum over $(z,w)\in R$ can be written as a sum over
$z\in X_m$ and coordinates $j$ with
$z_j=0$,
setting $w=z+e_j$.
Equivalently, it can be written as a sum over
$w\in Y_{m+1}$ and coordinates $j$ with
$w_j=1$,
setting $z=w-e_j$.
Therefore the two endpoint sums obey
\begin{align}
    \sum_{(z,w)\in R}p_{z,j(z,w)}
    &=\sum_{z\in X_m}\sum_{j:z_j=0}p_{z,j}
    \notag\\
    &\leq\sum_{z\in X_m}\sum_{j=1}^{M_{\mathrm{var}}}p_{z,j}
    \leq K,
    \notag\\
    \sum_{(z,w)\in R}p_{w,j(z,w)}
    &=\sum_{w\in Y_{m+1}}\sum_{j:w_j=1}p_{w,j}
    \notag\\
    &\leq\sum_{w\in Y_{m+1}}\sum_{j=1}^{M_{\mathrm{var}}}p_{w,j}
    \leq K.
    \label{eq:app_lower_bound_binary_endpoint_sums}
\end{align}
The last inequalities use $\sum_jp_{v,j}\leq1$ separately for each
normalized input state.
Applying the arithmetic--geometric mean inequality edge by edge now gives
\begin{align}
    &\sum_{(z,w)\in R}\sqrt{p_{z,j(z,w)}p_{w,j(z,w)}}
    \notag\\
    &\quad\leq\frac12\sum_{(z,w)\in R}
        \qty[p_{z,j(z,w)}+p_{w,j(z,w)}]\leq K.
    \label{eq:app_lower_bound_binary_probability_sum}
\end{align}

Combining the edge rate with this bound proves
\begin{align}
    \abs{\frac{dW_{\mathrm{bin}}}{dt}}
    &\leq\frac{2\Delta_{\min}\abs{g(t)}}{K}
        \sum_{(z,w)\in R}\sqrt{p_{z,j(z,w)}p_{w,j(z,w)}}
    \notag\\
    &\leq2\Delta_{\min}\abs{g(t)}.
    \label{eq:app_lower_bound_binary_progress_rate}
\end{align}
Input-independent gates preserve these distances, as in the proof of
\cref{thm:chern_lower_bound}.
Integrating the rate and using
\cref{eq:app_lower_bound_binary_required_progress} yields
\begin{align}
    W_{\mathrm{bin}}(T)
    &\leq2\Delta_{\min}T_{\mathrm{oracle}},
    \notag\\
    T_{\mathrm{oracle}}
    &\geq\frac{1-2\sqrt{\eta(1-\eta)}}{\Delta_{\min}}(m+1).
    \label{eq:app_lower_bound_binary_runtime}
\end{align}
Let $r\in\qty{0,1,2}$ be the remainder when $M-1$ is divided by three.
Since $m=\lfloor(M-1)/3\rfloor$, we have $M=3m+r+1$ and hence
$m+1-M/3=(2-r)/3\geq0$.
Together with $M\geq a_xa_y/4$, this proves
\cref{eq:binary_chern_lower_bound} for both tiles and slabs,
including $M=1,2,3$, for which $m=0$.

At the corner, the flattened-QWZ pair used in the proof of
\cref{thm:chern_lower_bound} already has Chern numbers zero and one.
The lower bound established for this pair in
\cref{eq:chern_lower_bound_theorem_result} is stronger than
\cref{eq:binary_chern_lower_bound}, completing the proof of
\cref{cor:binary_chern_lower_bound}.

\section{QWZ Calibration for the Hardness Construction}
\label{app:qwz_calibration}

We verify the phase and gap estimates used in the \(\BQP\) reduction for the
orientation-reversed QWZ model introduced in \cref{sec:qwz_preliminary}.
The Dirac-point evaluation below is standard for two-band Chern
insulators~\cite{sticlet_geometrical_2012}.  We include the short calculation
to fix our orientation and sign conventions and to obtain the uniform gap
margin needed in the reduction.
The Hamiltonian is given in \cref{eq:qwz_model_hardness}.  We use the
Chern-number convention of
\cref{eq:chern_number_def,eq:berry_curvature_projector_form}, with orientation
\(dk_x\wedge dk_y\).  Its Bloch vector is
\[
  \mathbf d(\mathbf{k};u)
  =\bigl(\sin k_x,-\sin k_y,u+\cos k_x+\cos k_y\bigr).
\]
The eigenvalues are \(\pm\|\mathbf d\|\).  A band touching requires
\(\sin k_x=\sin k_y=0\), so it can occur only at the four points
\(\mathbf{k}_*\in\{0,\pi\}^2\), when the corresponding Dirac mass
\(m_*:=d_z(\mathbf{k}_*;u)\) vanishes.

With the orientation \(dk_x\wedge dk_y\) and the projector convention of
the main text, the lower-band projector is
\(P_-=(I-\widehat{\mathbf d}\cdot\sigma)/2\), where
\(\widehat{\mathbf d}:=\mathbf d/\|\mathbf d\|\).  Pauli algebra gives
\[
  C[P_-]
  =\frac1{4\pi}\int_{\mathbb T^2}
  \widehat{\mathbf d}\cdot
  \bigl(\partial_{k_x}\widehat{\mathbf d}\times
        \partial_{k_y}\widehat{\mathbf d}\bigr)\,d^2\mathbf{k}.
\]
Thus \(C[P_-]\) is the degree of \(\widehat{\mathbf d}\), with no additional minus
sign.  Write $\mathbf q:=\mathbf{k}-\mathbf{k}_*$ near a Dirac point and
expand to first order in \(\mathbf q\).  The Bloch vector then has the local
form
\[
  \mathbf d(\mathbf{k}_*+\mathbf q)
  =(v_xq_x,v_yq_y,m_*)
  +O(\|\mathbf q\|^2),
\]
where
$v_x:=\cos k_{x,*}$, $v_y:=-\cos k_{y,*}$, and
\(\chi_*:=\operatorname{sgn}(v_xv_y)\).  The masses and chiralities are
\[
  \begin{array}{c|c|c}
    \mathbf{k}_* & m_* & \chi_* \\ \hline
    (0,0)       & u+2 & -1 \\
    (\pi,0)     & u   & +1 \\
    (0,\pi)     & u   & +1 \\
    (\pi,\pi)   & u-2 & -1
  \end{array}
\]
Dropping the quadratic terms gives the local massive-Dirac model.  Its
curvature integrates to
\[
  \begin{aligned}
  &\Omega_*(\mathbf q)
  =\frac{m_*v_xv_y}
  {2\bigl(v_x^2q_x^2+v_y^2q_y^2+m_*^2\bigr)^{3/2}},\\
  \frac1{2\pi}&\int_{\mathbb R^2}\Omega_*(\mathbf q)\,d^2\mathbf{q}
  =\frac12\chi_*\operatorname{sgn}(m_*).
  \end{aligned}
\]
When \(m_*\) changes from negative to positive, the corresponding change in
the lattice Chern number is
\[
  \Delta C_*
  =\frac12\chi_*(+1)-\frac12\chi_*(-1)
  =\chi_*.
\]
We use only this integer change across the gap closing.  An isolated continuum
cone is not assigned a half-integer lattice Chern number.  For \(u>2\),
\(d_z>0\) everywhere and the degree is zero.  Applying these changes at
\(u=2,0,-2\) can be made explicit as follows.  The Chern number is constant
while the gap remains open.  Starting from the trivial region \(u>2\) and
decreasing \(u\), each mass changes from positive to negative, so its
contribution to the change is \(-\chi_*\).  Therefore
\begin{align}
  C_{\rm QWZ}(u>2)&=0,\notag\\
  C_{\rm QWZ}(0<u<2)
    &=0-\chi_{(\pi,\pi)}=1,\notag\\
  C_{\rm QWZ}(-2<u<0)
    &=1-\chi_{(\pi,0)}-\chi_{(0,\pi)}=-1,\notag\\
  C_{\rm QWZ}(u<-2)
    &=-1-\chi_{(0,0)}=0.
\end{align}
At \(u=-2,0,2\) the gap closes, so the band Chern number is not defined.
Equivalently, the expression
\[
  \frac12\sum_*\chi_*\operatorname{sgn}(m_*)
\]
vanishes for \(u>2\) and has exactly the same integer changes at every gap
closing.  It therefore equals the lattice Chern number in every gapped
region.  Substituting the masses and chiralities from the table gives
\[
  C_{\rm QWZ}(u)
  =-\frac12\left[
    \operatorname{sgn}(u+2)-2\operatorname{sgn}(u)
    +\operatorname{sgn}(u-2)
  \right]
\]
whenever the band is isolated.  In particular, \(0<u<2\) gives \(C=1\),
whereas \(u>2\) gives \(C=0\).

The same four points give the exact gap margin.  Setting
\(a:=\cos k_x\) and \(b:=\cos k_y\), we have
\[
  \|\mathbf d(\mathbf{k};u)\|^2
  =2+u^2+2u(a+b)+2ab.
\]
This is affine in each of \(a,b\in[-1,1]\), so its minimum is attained
at a corner.  Hence
\[
  \min_{\mathbf{k}}\|\mathbf d(\mathbf{k};u)\|
  =\min\bigl\{|u+2|,|u|,|u-2|\bigr\}.
\]
For \(\ueff(\pacc)=3-2\pacc\), the no sector has
\(\ueff\in[7/3,3]\), and the yes sector has \(\ueff\in[1,5/3]\).
Both are at least \(c_0=1/3\) from the nearest gap closing, so the
unscaled two-band separation is at least \(2c_0\), as used in
\cref{sec:gap_chern_transfer}.

\bibliography{ref}
\end{document}